\documentclass[aps,floatfix,prl,reprint,superscriptaddress]{revtex4-1}
\usepackage{amsmath,amssymb,graphicx,mathtools}
\usepackage{float}

\usepackage[T1]{fontenc}

\usepackage{color}
\usepackage{colortbl}
\usepackage{nicefrac}
\usepackage[hypertexnames=false,colorlinks,urlcolor=blue,citecolor=blue]{hyperref}
\usepackage{natbib}
\usepackage{tabularx}

\usepackage{dsfont}   %

\makeatletter
\DeclareRobustCommand{\cev}[1]{%
  \mathpalette\do@cev{#1}%
}

\makeatother

\newcounter{theoremcounter}
\stepcounter{theoremcounter}

\newtheorem{theorem}{Theorem}

\newtheorem{lemma}{Lemma}

\newenvironment{proof}[1][Proof]{\begin{trivlist}
\item[\hskip \labelsep {\bfseries #1}]}{\end{trivlist}}

\newcommand{\qed}{\hfill $\blacksquare$}

\newcommand{\ket}[1]{|#1\rangle}
\newcommand{\braket}[2]{\langle #1|#2\rangle}
\newcommand{\expval}[3]{\langle #1|#2|#3\rangle}

\newcommand{\partder}[2]{\frac{\partial #1}{\partial #2}}
\newcommand{\der}[2]{\frac{\mbox d #1}{\mbox d #2}}

\let\Re\relax

\DeclareMathOperator{\Re}{\text{Re}}

\definecolor{blue(ryb)}{rgb}{0.01, 0.28, 1.0}
\definecolor{scarlet}{rgb}{1.0, 0.13, 0.0}

\makeatletter
\newcommand{\vast}{\bBigg@{4}}
\newcommand{\Vast}{\bBigg@{5}}
\makeatother

\newcommand{\snlca}{Quantum Algorithms and Applications Collaboratory, Sandia National Laboratories, Livermore CA 94550, U.S.A.}
\newcommand{\snlnm}{Quantum Algorithms and Applications Collaboratory, Sandia National Laboratories, Albuquerque NM 87185, U.S.A.}

\newcommand{\papertitle}{Self-healing of Trotter error in digital adiabatic state preparation} %

\begin{document}
\title{\papertitle}
\author{Lucas K. Kovalsky}
\email{lkocia@sandia.gov}
\affiliation{\snlca}
\author{Fernando A. Calderon-Vargas}
\affiliation{\snlca}
\author{Matthew D. Grace}
\affiliation{\snlca}
\author{Alicia B. Magann}
\affiliation{\snlnm}
\author{James B. Larsen}
\affiliation{\snlnm}
\affiliation{Department of Mathematics, Brigham Young University, Provo, Utah 84602, U.S.A.}
\author{Andrew D. Baczewski}
\email{adbacze@sandia.gov}
\affiliation{\snlnm}
\author{Mohan Sarovar}
\email{mnsarov@sandia.gov}
\affiliation{\snlca}
\date{\today }
\begin{abstract}
Adiabatic time evolution can be used to prepare a complicated quantum many-body state from one that is easier to synthesize and Trotterization can be used to implement such an evolution digitally.
The complex interplay between non-adiabaticity and digitization influences the infidelity of this process.
We prove that the first-order Trotterization of a complete adiabatic evolution has a cumulative infidelity that scales as $\mathcal O(T^{-2} \delta t^2)$ instead of $\mathcal O(T^2 \delta t^2)$ expected from general Trotter error bounds, where $\delta t$ is the time step and $T$ is the total time.
This result suggests a self-healing mechanism and explains why, despite increasing $T$, infidelities for fixed-$\delta t$ digitized evolutions still decrease for a wide variety of Hamiltonians.
It also establishes a correspondence between the Quantum Approximate Optimization Algorithm (QAOA) and digitized quantum annealing.
\end{abstract}
\maketitle

Preparing the ground state of a quantum many-body Hamiltonian is generically difficult~\cite{kitaev2002,kempe2006complexity,Poulin09}.
Nevertheless, because we frequently observe systems near their ground state in nature, we expect to be able to efficiently prepare these states in laboratories or on quantum computers for a wide range of physical Hamiltonians~\cite{feynman1982simulating,Deutsch85,brooke1999quantum,du2010nmr,preskill2018simulating}.
One approach to ground state preparation is through an adiabatic evolution that interpolates between a Hamiltonian with an easy-to-prepare ground state ($H_1$) and a Hamiltonian with the target ground state ($H_2$).
This has applications in quantum computation~\cite{aharonov2008adiabatic,albash2018adiabatic}, linear algebra~\cite{subacsi2019quantum,costa2022optimal,an2022quantum}, optimization~\cite{farhi2000quantum}, and simulation~\cite{Alan05,jordan2012quantum,lee2023evaluating}.
Realizing these applications requires an understanding of the sources of error in adiabatic state preparation (ASP).

The total error is often quantified as the infidelity, $\mathcal{I}$, of the prepared state relative to the ideal target state.
An ever-present contribution to $\mathcal{I}$ is due to the fact that such an evolution cannot proceed infinitely slowly in practice~\cite{Landau32,Zener32}.
Digitizing the evolution into $r$ time steps via Trotterization~\cite{Trotter59,Suzuki76}, as would be necessary on a gate-based quantum computer~\cite{Lloyd96,van2001powerful}, introduces a second influence on $\mathcal{I}$ due to the fact that we cannot exactly represent the ideal continuous-time dynamics.
This Letter explores the interplay between these effects.

We show that certain errors cancel out over the course of a \emph{complete} adiabatic evolution from $H_1$ to $H_2$ requiring time $T$.
One should expect digitization to degrade $\mathcal{I}$ for larger $T$ with a fixed time step $\delta t=T/r$, \textit{i.e.,} more time steps lead to more accumulation of errors.
Indeed, a generic upper bound on $\mathcal{I}$ for first-order Trotterization suggests that its error scales as $\mathcal{O}(T^2\delta t^2)$.
This would mean that $\delta t$ needs to decrease to realize a fixed $\mathcal{I}$ as $T$ increases.
However, numerical results suggest that $\mathcal{I}$ \emph{decreases} with increasing $T$, even for fixed $\delta t$.

We present a less generic upper bound on $\mathcal{I}$ in Theorem~\ref{th:layden_timedep_informal}, similar to one recently proved for time-independent Hamiltonians by Layden~\cite{Layden21}.
This bound relies on adiabaticity, but not on the evolution being complete.
While it improves the scaling with $T$, it does not explain reductions in $\mathcal{I}$ with increasing $T$ for fixed $\delta t$.

However, Theorem~\ref{th:adiabatic_trotter_error} gives conditions under which $\mathcal{I}$ instead \emph{improves} with the duration of the evolution, consistent with numerics.
This bound relies on both adiabaticity and the evolution being complete.
Critically, it allows for further Trotterization of $H_1$ and $H_2$, as might be required in applications like quantum simulation.
It also does not rely on the ordering of the Trotterization.
Fig.~\ref{fig:sketch} summarizes the primary consequence of Theorem~\ref{th:adiabatic_trotter_error}.

\begin{figure}[t]
\includegraphics[scale=1.0]{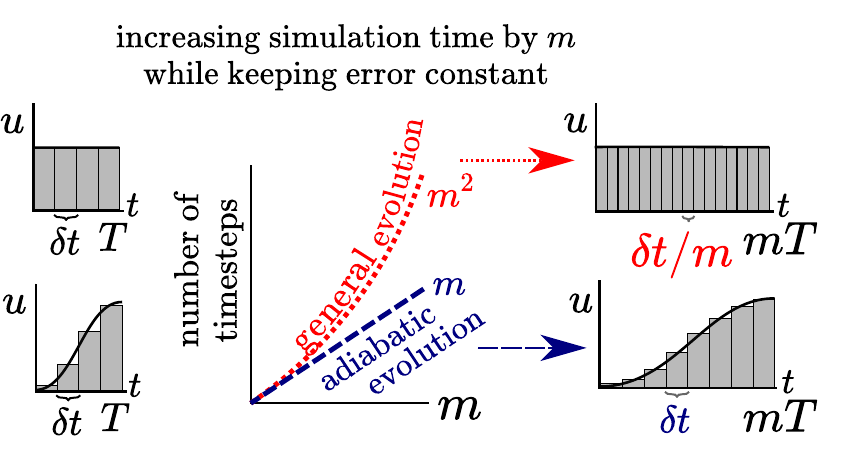}
\caption{
Consider a Hamiltonian with two terms, $H\left[u(t)\right] = (1-u(t))H_1 + u(t) H_2$.
A time evolution from $t=0$ to $t=T$ can be broken into steps of size $\delta t$ through first-order Trotterization (left).
For a general evolution of this sort, previously established bounds suggest that cumulative infidelity should grow with total time $T$ for fixed $\delta t$.
Then, scaling $T$ by $m$ should require dividing $\delta t$ by $m$ to preserve a fixed infidelity (top right).
However, for an adiabatic evolution from $u(0)=0$ to $u(T)=1$, we show that cumulative infidelity actually scales as $\mathcal{O}(T^{-2}\delta t^2) + \mathcal{O}(T^{-2})$ (Th.~\ref{th:adiabatic_trotter_error}).
This allows us to keep $\delta t$ fixed if we scale $T$ by $m$, to achieve a fixed or decreasing infidelity, consistent with the continuous-time limit (bottom right).
Thus, the total number of time steps will be at least quadratically lower than expected (center).}
\label{fig:sketch}
\end{figure}

The crossover between the bounds in Theorems~\ref{th:layden_timedep_informal} and \ref{th:adiabatic_trotter_error} suggest a \emph{self-healing} mechanism for complete adiabatic evolutions.
There is a sense in which $\mathcal{I}$ gets worse before it gets better, and this is supported by numerical results.
We will also show that this explains a relationship between Trotterized quantum annealing and the Quantum Approximate Optimization Algorithm (QAOA)~\cite{Steffen03,Farhi16,Zhou20}.

We begin by considering an adiabatic evolution generated by a time-dependent Hamiltonian, \(H[u(t)] = (1-u(t))H_1 + u(t)H_2\), for \(u(t) \in [0,1]\) where $u(t)$ is $0$ at $t=0$ and $1$ at $t=T$.
The unitary associated with the continuous-time dynamics is $U(t) = \mathcal{T} \exp(-i\int_0^t H[u(t')]{\rm d}t{'})$, where $\mathcal{T}$ is the time-ordering operator.
$U(t)$ is approximated by digitizing the time evolution with first-order Trotterization~\cite{Trotter59,Suzuki76,Lloyd96,Childs19,Childs21}.
We will be particularly interested in the Trotterization of the complete adiabatic evolution,
\begin{equation}
  \label{eq:Utrot}
  U(T) \approx U^{(1)}(T) = \prod_{k=1}^r \prod_{i=1}^2 U_i((k{-}1)\delta t, k\delta t),
\end{equation}
where \(U_1((k{-}1) \delta t, k\delta t) = e^{-iH_1 \int_{(k-1) \delta t}^{k \delta t} \text d t' (1 - u(t'))}\) and \(U_2((k{-}1) \delta t, k\delta t) = e^{-iH_2 \int_{(k-1) \delta t}^{k \delta t} \text d t' u(t')}\)~\footnote{Note that this definition is not unique. The ordering could be reversed in $U^{(1)}$, \textit{i.e.,} $U_2U_1$. While the precise value of the Trotter error will depend on the order, the scaling with $\delta t$ and $T$ is independent of this choice.}.

The error incurred by splitting the exponential this way is typically called the Trotter error and scales as $||U(T) - U^{(1)}(T)||=\mathcal{O}(T \delta t||[H_1,H_2]||)$, where $||\cdot ||$ is the operator norm.
Note that $T\delta t = r\delta t^2 = T^2/r$, and this scaling represents the leading-order contribution to the error in $\delta t$.
It can be derived by bounding the error in a single time step and applying the triangle inequality to aggregate the error over all $r$ steps~\cite{Trotter59,Suzuki85}.
When \(\mathcal I\) is defined as \(\mathcal{I}(T)\equiv 1 - |\expval{\psi}{U^\dagger(\infty) U^{(1)}(T)}{\psi}|\)~\footnote{Notice that our definition of infidelity differs a bit from the more common \(1 - |\expval{\psi}{U^\dagger(\infty) U^{(1)}(T)}{\psi}|^2\).}, then it accounts for infidelity from both digitization and non-adiabaticity.
Specifically, \(\mathcal I\) is upper bounded by the squared sum of the Trotter error \(\|U(T) - U^{(1)}(T)\|\) and an energy gap-dependent \(\mathcal O(T^{-1})\) term accounting for non-adiabaticity~(see Lemma 4~\cite{SMref}). Theorems~\ref{th:layden_timedep_informal} and~\ref{th:adiabatic_trotter_error} improve on this bound for less generic adiabatic evolutions.

That this is possible is motivated by numerical investigations.
Fig.~\ref{fig:scaling} shows $\mathcal{I}$ as a function of $\delta t$ for a simple two-level system described in terms of Pauli matrices.
For this example, $H_1=X, H_2=Z$, and the schedule $u(t)=t/T$ is a linear ramp.
While this is the simplest possible example, the phenomenology that it captures generalizes to more complicated choices for $H_1$ and $H_2$, some of which are considered in the Supplemental Materials (SM)~\cite{SMref}.
In the red-shaded region of Fig.~\ref{fig:scaling}, $\delta t \|H(t)\| \in \Omega(1)$ and the Trotter product formula is non-convergent in this region.
We do not expect predictable scaling of error with $\delta t$ in this region and focus on the behavior in the white and green-shaded regions.

We draw attention to two interesting features in Fig.~\ref{fig:scaling}.
First, the observed scaling of $\mathcal{I}$ is much more favorable than the scaling suggested by the generic $\mathcal{O}(T^2 \delta t^2)$ upper bound given above.
We instead see a $\mathcal{O}(T^{-2}\delta t^2)$ scaling, which is completely inconsistent with the expectation that error should increase with $T$.
This inconsistency is due to the use of the triangle inequality in deriving the generic Trotter error bound, which neglects error cancellation effects evident in Fig.~\hyperref[fig:scaling]{\ref*{fig:scaling}(c)}.
The second feature is that $\mathcal{I}$ asymptotes to a $\delta t$-independent quantity as $\delta t\rightarrow 0$ in Fig.~\hyperref[fig:scaling]{\ref*{fig:scaling}(a)}.
In this green-shaded region the digitization error is smaller than the finite-$T$ non-adiabaticity error, and thus $\mathcal{I}$ is independent of $\delta t$ and consistent with $\mathcal O(T^{-2})$ bounds on the continuous-time adiabatic evolution~\cite{Mackenzie06,Cheung11}.

\begin{figure}[t]
\includegraphics[scale=1.0]{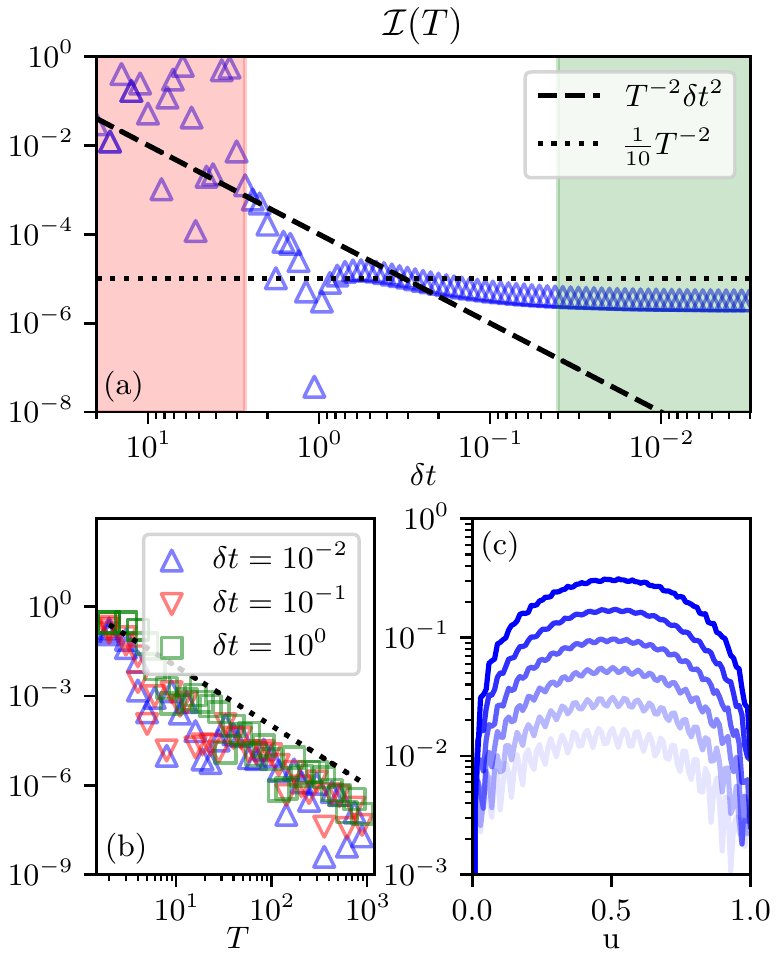}
\caption{
Infidelity of digitized ASP for a two-level system.
(a) Scaling of $\mathcal{I}(T)$ with Trotter time step size $\delta t$, for fixed evolution time $T=100$. %
The dotted/dashed lines show scalings that dominate the upper bound on $\mathcal{I}(T)$ for different values of $\delta t$.
(b) Scaling of $\mathcal{I}(T)$ with $T$ for different values of $\delta t \leq 1$.
The dashed line shows $T^{-2}$ scaling, rather than a generic ${\sim}T^2$ scaling that might be expected due to digitization error.
(c) Intermediate infidelities as a function of $u$, with increasing $\delta t$ corresponding to increasing opacity.
The larger the value of \(\delta t\), the higher the peak of intermediate infidelity at \(u=1/2\).
A self-healing mechanism for complete adiabatic evolutions is evident, as the infidelity comes back down for $u=1$.
}
\label{fig:scaling}
\end{figure}

Our first result is an improved upper bound on the Trotter error associated with the approximation in Eq.~\ref{eq:Utrot}.
\begin{theorem}[Informal]
\label{th:layden_timedep_informal}
Given a gapped Hamiltonian \(H[u(t)] = (1-u(t)) H_1 + u(t) H_2\) and a unitary \(U(T) = \mathcal{T} e^{-i\int_0^T H[u(t')] {\rm d} t'} \), where \(u(t) = s(t/T)\) and $s\!:\![0,1]\!\rightarrow\![0,1]$, if this unitary is first-order Trotterized into \(H_1\) and \(H_2\) terms with \(r\) time steps of size \(\delta t\), as in Eq. \eqref{eq:Utrot}, then as \(T \rightarrow \infty\), $\mathcal{I}(T)$ is upper bounded by
\(\min\{1, C_2T^2 \delta t^2, (C_1 \delta t + C_3 T \delta t^2)^2\} + \mathcal O(\delta t) + \mathcal O(T^{-2})\).
\end{theorem}
The formal version of Theorem~\ref{th:layden_timedep_informal} defines the coefficients $C_i$ that determine the relative magnitudes of the contributions to the error~\cite{SMref}. The coefficients of the non-adiabatic \(\mathcal O(\delta t)\) and \(\mathcal O(T^{-2})\) terms depend on the energy gap~\cite{Mackenzie06,Cheung11}.
We note that Theorem~\ref{th:layden_timedep_informal} only applies to two-term first-order Trotterization, which means that $U_1$ and $U_2$ cannot be further Trotterized.

Theorem~\ref{th:layden_timedep_informal} tightens the generic \(\mathcal O(T^2 \delta t^2)\) scaling of \(\mathcal I(T)\) to \(\mathcal O(T^2 \delta t^4)\) and \(\mathcal O(\delta t^2)\) in some regimes dependent on the \(C_i\) coefficients.
The proof involves combining the first-order Trotter error from subsequent time steps into cumulative second-order Trotter error, similar to a recent bound for evolution under a time-independent Hamiltonian~\cite{Layden21}.
The main ingredients are bounds on the time-dependent first- and second-order Trotter expansion errors that do not require treating discretization error explicitly, in contrast to prior approaches, and they sidestep explicit Magnus expansion~\cite{Yi21,yi2021success,Sahinouglu21}.
As in the time-independent case, the two contributions \(C_1 \delta t + C_3 T \delta t^2\) dominate in different parameter regimes.
\(\mathcal O (T \delta t^2)\) scaling occurs when the evolution is long enough that the endpoints are insignificant, while $T$-independent \(\mathcal O(\delta t)\) scaling occurs when the evolution is short enough that the endpoints dominate~\cite{Layden21}.

While Theorem~\ref{th:layden_timedep_informal} introduces a scaling independent of \(T\) in the short-time regime unlike generic bounds, it does not capture the decrease in $\mathcal{I}$ with increasing $T$ for fixed-$\delta t$ evolutions, evident in Fig.~\hyperref[fig:scaling]{\ref*{fig:scaling}(b)}.
This is because the proof technique relies on adiabaticity (\textit{i.e.,} that $u'$ and $u''$ go to 0 and the Hamiltonian is gapped) but it does not rely on the fact that the adiabatic evolution is complete (\textit{i.e.,} that $u$ goes from $0$ to $1$).
Our second result shows that accounting for this leads to a bound with the anticipated behavior.
\begin{theorem}
  Given a gapped Hamiltonian \(H[u(t)] = (1-u(t)) H_1 + u(t) H_2\) and a unitary \(U(T) = \mathcal T  e^{-i\int_0^T H[u(t')]{\rm d} t'}\), where \(u(t) = s(t/T)\) is smooth and $s\!:\![0,1]\!\rightarrow\![0,1]$, if $U(T)$ is first-order Trotterized with fixed time steps \(\delta t \in \mathcal O(\min_t \|H[u(t)]\|^{-1})\), then as \(T \rightarrow \infty\) with \(u(0) \rightarrow 0\) and \(u(T) \rightarrow 1\) the final state infidelity is bounded by \(\mathcal O ( T^{-2} \delta t^2) + \mathcal O(T^{-2})\).
  Moreover, initially at a given fixed \(t/T \ll 1\), state infidelity increases as \(\mathcal O(t^2 \delta t^2)\).
\end{theorem}
The \(\mathcal O(T^{-2})\) term is again due to energy gap-dependent non-adiabaticity error associated with finite-\(T\) evolution. %
Note that the short-time $\mathcal{O}(t^2 \delta t^2)$ bound matches the generic \(\mathcal O(T^2 \delta t^2)\) Trotter error bound. %

The proof technique~\cite{SMref} involves analyzing the coefficients of the discretized time-evolution of the ground state in the adiabatic basis using first-order time-dependent perturbation theory.
This approach reveals that the leading-order error from Trotterization is due to an off-diagonal harmonic perturbation with amplitude $\mathcal{O}(\delta t)$.
While its amplitude is independent of \(T\), its frequency scales as \(T^{-1}\).
Thus, as $T$ increases this low-frequency perturbation becomes increasingly off-resonant and it induces transitions out of the ground state with a probability \(\mathcal O (T^{-2} \delta t^2)\), \textit{i.e.,} similar to the Lorentzian tail that appears in the solution of the Rabi problem~\cite{allen1987optical}.

Theorem~\ref{th:adiabatic_trotter_error}'s scaling holds for all gapped Hamiltonians, even with simple linear control ramps \(u(t)\), as corroborated by numerical results~\cite{SMref}.
We believe this generality explains the widespread inverse-in-\(T\) scaling reported elsewhere~\cite{Zhou20,Brady21}.

\begin{figure}[t]
\includegraphics[scale=1.0]{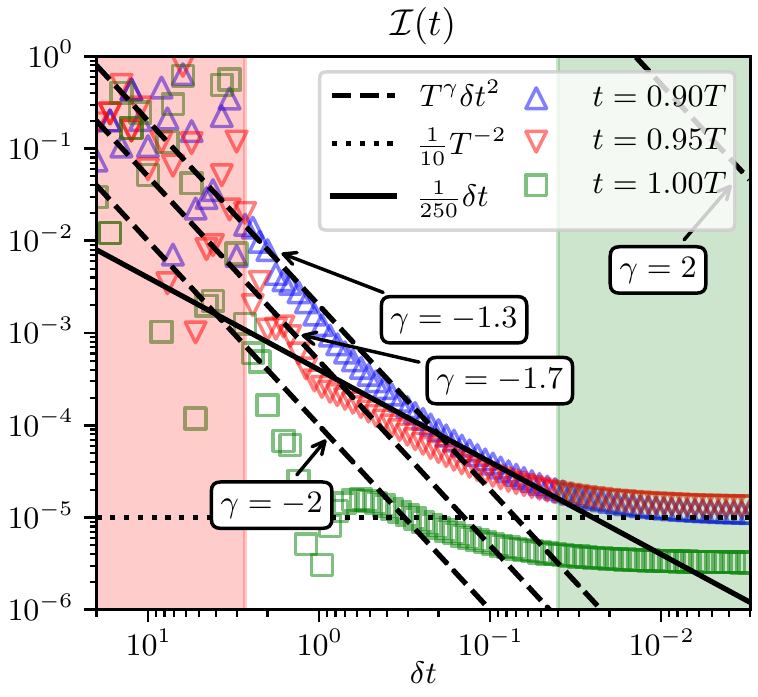}
\caption{Infidelity of digitized ASP for the same two-level system as in Fig. \ref{fig:scaling}, but including intermediate evolutions $U(t)$ for $t\leq T$.
$T=100$ is fixed for all curves.
Here $\mathcal{I}$ considers the target state as the instantaneous ground state of $H\left[u(t)\right]$ for the indicated values of $t$.
The dotted/dashed lines show various scalings with $\delta t$ and illustrate how the infidelity transitions from being dependent (\({\sim} \delta t^2 T^2\)) to inversely dependent (\({\sim} T^{-2} \delta t^2\)) on total time \(T\) as the evolution is completed.
}
\label{fig:partialramps}
\end{figure}

Fig.~\ref{fig:partialramps} illustrates the error bounds in Theorems~\ref{th:layden_timedep_informal} and \ref{th:adiabatic_trotter_error} for the same simple two-level system studied in Fig.~\ref{fig:scaling}.
To illustrate the bounds in Theorem~\ref{th:layden_timedep_informal} it is necessary to consider incomplete adiabatic evolutions, as to avoid the superior scaling that complete evolutions achieve according to Theorem~\ref{th:adiabatic_trotter_error}.
Thus we examine $\mathcal{I}(t)$ for $t \leq T$, in which the states relative to which infidelities are evaluated are the instantaneous ground states of $H[u(t)]$.

For the $t\ll T$ evolution, $\mathcal{I}$ initially scales as the generic bound $\mathcal{O}(T^2 \delta t^2)$ (dashed line) in Theorem~\ref{th:layden_timedep_informal} and then crosses over to $\mathcal{O}(\delta t)$ scaling (solid line).
The constant coefficients in Theorem~\ref{th:layden_timedep_informal} dictate the size of the $\mathcal{O}(\delta t)$-scaling region, which is due to a cross-term involving both Trotter error and non-adiabatic error, and we see that this region dominates the $t = 0.95T$ scaling for most of the relevant \(\delta t\) values~\cite{SMref}. %
The $t=T$ curve achieves $\mathcal{O}(T^{-2}\delta t^2)$ scaling because it is a complete evolution.
As in Fig. \ref{fig:scaling}, all three curves plateau at small $\delta t$ (large $r$) once the digitization error is dominated by the finite-$T$ non-adiabaticity error.
We briefly note that the size of the \(\mathcal O(T^2 \delta t^2)\)-scaling region compared to \(\mathcal O(\delta t)\) can be changed if variable time steps are used~\cite{SMref}.

Fig.~\ref{fig:partialramps} also reveals an interesting crossover in the dependence of $\mathcal{I}(t)$ on $\delta t$ as $t\rightarrow T$.
This is evident in the white region, in which the error bounds given by Theorems \ref{th:layden_timedep_informal} and \ref{th:adiabatic_trotter_error} are $\mathcal{O}(T^2\delta t^2)$ and $\mathcal{O}(T^{-2}\delta t^2)$, respectively.
The $T$-dependent part of the error can be written as \(\mathcal O(T^\gamma \delta t^{2})\), where $\gamma$ transitions from \(2\) to \(-2\) as $t\rightarrow T$.
As this limit is approached, the upper bound in Theorem~\ref{th:layden_timedep_informal} becomes looser and looser as the scaling transitions to the tighter upper bound in Theorem~\ref{th:adiabatic_trotter_error}. %
It appears that the prefactor of the \(\mathcal O(T^{\gamma} \delta t^2)\) term remains relatively constant during this transition, and so the different $T$ behavior is apparent by the shifted dashed curves in Fig.~\ref{fig:partialramps} as its power changes.
For $t=T$, Theorem~\ref{th:adiabatic_trotter_error}'s bound becomes valid and $\mathcal{I}(t=T)$ decreases with $T$ as \(\mathcal O(T^{-2} \delta t^2)\) for all relevant values of \(\delta t\), as seen in Fig.~\hyperref[fig:scaling]{\ref*{fig:scaling}(b)}.

Moreover, the transition from the initial \(\mathcal O(t^2 \delta t^2)\) scaling to the \(\mathcal O(T^{-2}\delta t^2 )\) cumulative scaling implies a cancellation of errors incurred at intermediate times, evident in Fig.~\hyperref[fig:scaling]{\ref*{fig:scaling}(c)}.
We find this \(\mathcal O(t^2 \delta t^2)\) upper bound to be tight and that increasing ground state infidelity at intermediate times is reversed as $t\rightarrow T$ after traversing the system's avoided crossing~\cite{SMref}.
We do not prove the mechanism of the reversal, but the same phenomenon can be observed in more complex systems, where transitions to many excited states at intermediate times are reversed as $t\rightarrow T$~\cite{SMref}.
This remarkable property of \emph{self-healing} digitized adiabatic evolutions has been empirically observed, but had otherwise defied explanation~\cite{Honda22,Albash22}.

This begs the question of how self healing impacts higher-order Trotterization for adiabatic evolutions and resource requirements for ASP~\cite{SMref}.
A $p$th-order generalization of Theorem~\ref{th:adiabatic_trotter_error} will still include a $\delta t$-independent $\mathcal{O}(T^{-2})$ term, and the interplay of this diabatic error with improved $\mathcal{O}(\delta t^{p})$ digitization error requires further study.
Our analysis also implies reductions in circuit depths for Trotterized ASP relative to generic bounds, with potentially significant consequences for resource estimates of ASP.
However, the optimal approach to ground-state preparation is likely to be problem dependent and a comparison of Trotterized ASP to alternatives~\cite{Reiher17,Ge19,Lin20,Wan20,Lemieux21} is a topic for future work.

In the context of optimization algorithms, our results establish a bijective correspondence between QAOA and digitized quantum annealing~\cite{Zhou20,Brady21}. As in Fig.~\ref{fig:QAOA_anneal}, it is often possible to find an injective correspondence between a set of optimal angles for QAOA and a Trotterization of an optimized quasi-adiabatic (or annealing) evolution between the driver ($H_d$)  and problem ($H_p$) Hamiltonians. But justification for the surjective correspondence (QAOA angles from Trotterizing a given quasi-adiabatic evolution) has remained elusive.

\begin{figure}[ht]
\includegraphics[scale=0.6]{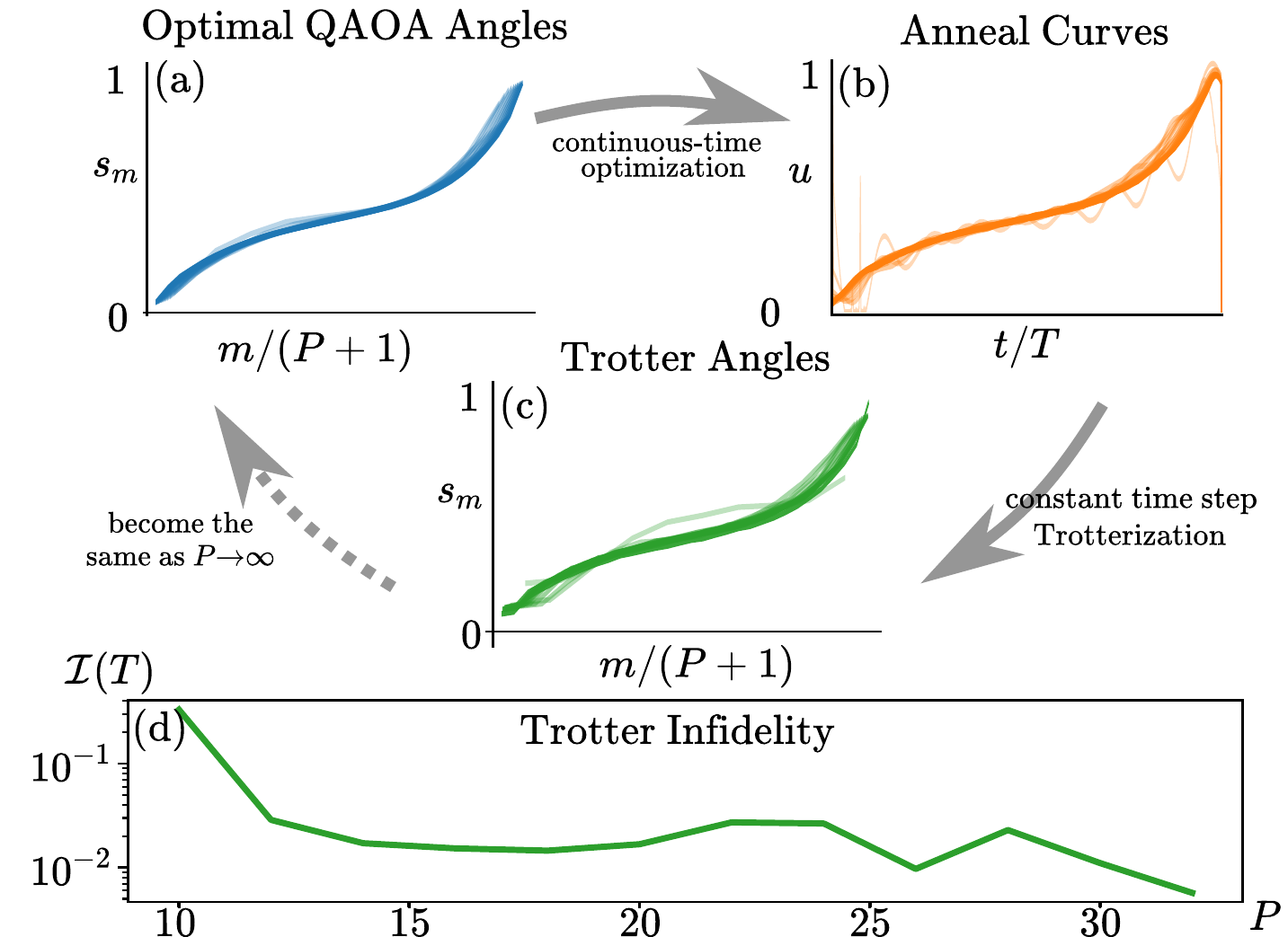}
\caption{Correspondence between QAOA and digitized quantum annealing.
We consider a 3-regular Ising model (MAXCUT) on $N=8$ qubits with periodic boundary conditions.
\(s_m = \gamma_m/(\gamma_m+\beta_m)\), where \(m \in [1,\,\ldots,\,P]\) and \(\gamma_m\) (\(\beta_m\)) are the \(m\)th QAOA angle corresponding to \(H_p\) (\(H_d\)).
(a) QAOA curves found for \(P=\{10,\,\dots,\, 32\}\) layers by bootstrapped seeding~\cite{Pichler18,Mbeng19,Mbeng19_2} from smaller \(P\) solutions (darker curves correspond to larger \(P\)).
(b) Optimal anneal curves found by seeding from the corresponding QAOA curves in (a).
(c) Constant time step Trotterization of these anneal curves.
(d) The infidelity with the ground state of $H_p$ after evolution under the digitized annealing schedule.
The decrease in infidelity demonstrates an instance when QAOA optimal angles can be considered to be proximal to a \emph{constant non-vanishing $\delta t$} Trotterization of a universal family of anneal curves.
This correspondence becomes exact as \(P\rightarrow \infty\)~\cite{SMref}.}
\label{fig:QAOA_anneal}
\end{figure}

This is because for a given set of \(P\) pairs of optimal QAOA angles (Fig.~\hyperref[fig:QAOA_anneal]{\ref*{fig:QAOA_anneal}(a)}) the corresponding continuous anneal control curve often has a similar total integrated time \(T\) (Fig.~\hyperref[fig:QAOA_anneal]{\ref*{fig:QAOA_anneal}(b)}).
This forces the timestep of the QAOA to scale as \(\delta t \propto T/P\) when viewed as a Trotterization.
Since \(P\) is generally found to be proportional to \(T\) in unrestricted QAOA, this becomes a fixed-\(\delta t\) Trotterization.
Such a Trotterization has been phenomenologically found to be the best discretization to match the oscillating curves of the adiabatic anneal, whose period scales as the ($T$-independent) energy gap~\cite{Brady21}.
Prior efforts were unable to prove the surjective correspondence between Fig.~\hyperref[fig:QAOA_anneal]{\ref*{fig:QAOA_anneal}(a)} and~\hyperref[fig:QAOA_anneal]{\ref*{fig:QAOA_anneal}(c)} because generic bounds suggest that error increases with $T$ for fixed $\delta t$.

Here we have shown that for the broad class of Hamiltonians that satisfy Theorem~\ref{th:adiabatic_trotter_error}, $\mathcal{I}$ is constant or decreasing with fixed $\delta t$ and increasing $T$.
This facilitates proving the final relationship between Fig.~\hyperref[fig:QAOA_anneal]{\ref*{fig:QAOA_anneal}(c)} and~\hyperref[fig:QAOA_anneal]{\ref*{fig:QAOA_anneal}(a)}, producing a fully bijective relationship between optimal QAOA angles and continuous anneal curves.
Trotterization of these curves with $\delta t = T/P$, produces angles (Fig.~\hyperref[fig:QAOA_anneal]{\ref*{fig:QAOA_anneal}(c)}) that, as \(P \rightarrow \infty\), approach the original QAOA angles and evolve the initial state to the ground state with increasing fidelity (Fig.~\hyperref[fig:QAOA_anneal]{\ref*{fig:QAOA_anneal}(d)}).

Thus at sufficiently large depth, QAOA can become a fixed-$\delta t$ digitization of an underlying set of quantum annealing curves, a limit that differs from the $\delta t\rightarrow 0$ digitization of the adiabatic limit traditionally considered~\cite{Farhi14}.
These annealing curves approach a single asymptotic curve as the total integrated time goes to infinity.
This correspondence allows for high-depth QAOA instances to be seeded by interpolating low-depth instances~\cite{Pichler18,Mbeng19,Mbeng19_2,Zhou20,Brady21} that converge quickly to the high-depth instance's minimum.
Theorem~\ref{th:adiabatic_trotter_error} justifies this widely used ``bootstrap'' procedure.

Based on prior Trotter bounds, it might have been expected that digitization error would dominate the cumulative infidelity of digitized ASP with increasing $T$.
For a fixed $\delta t$, more time steps should lead to more error.
However, thanks to a self-healing property of complete adiabatic evolutions this is not the case.
We have applied this to establish a correspondence between QAOA and digitized quantum annealing, but future work remains in exploring the mechanism of the self-healing property and consequences for other quantum algorithms that rely on ASP.

\vskip5pt
\noindent--

We acknowledge useful conversations with Jonathan Wurtz and Tameem Albash about portions of this work.
We are also grateful to the anonymous referees whose thorough comments greatly improved its presentation.
This material is based upon work supported by the U.S. Department of Energy, Office of Science, Office of Advanced Scientific Computing Research, under the Quantum Computing Application Teams (QCAT) and Accelerated Research in Quantum Computing (ARQC) programs, the National Nuclear Security Administration's Advanced Simulation and Computing program, and the National Science Foundation under Grant No. NSF PHY-1748958.
A.D.B., A.B.M., and J.B.L. acknowledge support from the Sandia National Laboratories Truman Fellowship Program, which is funded by the Laboratory Directed Research and Development (LDRD) program.
Sandia National Laboratories is a multimission laboratory managed and operated by National Technology \& Engineering Solutions of Sandia, LLC, a wholly owned subsidiary of Honeywell International Inc., for the U.S. Department of Energy's National Nuclear Security Administration under contract DE-NA0003525.
This paper describes objective technical results and analysis.
Any subjective views or opinions that might be expressed in the paper do not necessarily represent the views of the U.S. Department of Energy or the United States Government.

\clearpage
\widetext
\begin{center}
\textbf{\large Supplemental Materials: \papertitle}
\end{center}

\setcounter{secnumdepth}{2}
\setcounter{section}{0}
\setcounter{page}{1}
\setcounter{theorem}{0}

The Supplemental Materials include proofs of the major theorems in the main body of the paper, as well as further technical details and numerical evidence relevant to both.
We also include a third Theorem that isn't included in the main body of the paper that describes the scaling of Trotter error with variable time steps.

\begin{itemize}
	\item Appendix~\ref{app:timedep_trotter} provides and proves Lemmas that bound the first- and second-order Trotter error for time-dependent adiabatic evolutions between $H_1$ and $H_2$. It also includes a short proof that $\mathcal{I}(T)$ is upper bounded by the Trotter error and an \(\mathcal O(T^{-2})\) term.
	These Lemmas are used to prove Theorem~\ref{th:layden_timedep}, which bounds the first-order Trotter error via a technique similar to Ref.~\cite{Layden21}.
	\item Appendix~\ref{app:adiabatic_trotter_error} provides and proves Theorem~\ref{th:adiabatic_trotter_error}.
	It also includes numerical results that corroborate the associated short-time scaling of $\mathcal{I}$.
	\item Appendix~\ref{app:counterdiabatic} provides a third theorem that describes bounds for variable time step Trotterization.
	It also includes numerical results that corroborate the predicted scaling.
	\item Appendix~\ref{app:QAOA} contains more details concerning the correspondence between QAOA and Trotterized annealing, including numerical results.
	\item Appendix~\ref{app:numerics} 	demonstrates that the scalings predicted in Theorems 1 and 2 are also evident in a few more general Hamiltonians than the two-level system considered as the motivation in the main text.
	\item Appendix~\ref{app:resource_estimates} discusses the implications for our improved scaling that are relevant to ground state preparation in quantum simulation.
\end{itemize}

\appendix

\section{Theorem 1 - Statement and proof}
\label{app:timedep_trotter}
\renewcommand{\theequation}{A\arabic{equation}}
\renewcommand{\thefigure}{A\arabic{figure}}
\setcounter{figure}{0}

Here we bound the first-order Trotter (operator norm) error for a time-dependent Hamiltonian.

The operator norm of the Hilbert space \(\mathbb H\) is defined as
\begin{equation}
    \|A \| \equiv \inf \{c > 0: \|A v\|_\text{vec} \le c \|v\|_\text{vec} \, \forall v \in \mathbb H\},
\end{equation}
for any norm \(\|\cdot\|_\text{vec}\) on \(\mathbb H\).

Lemma~\ref{le:first_order_trotter} first establishes a ``na\"ive'' bound on first-order Trotter error for time-dependent Hamiltonians and Lemma~\ref{le:second_order_trotter} subsequently establishes a bound on second-order Trotter error for time-dependent Hamiltonians. Theorem~\ref{th:layden_timedep} then combines these two results for a tighter bound on first-order Trotter error, in the same manner as~\cite{Layden21} accomplished for the time-independent case.

We begin with a ``na\"ive'' bound on first-order Trotter error for time-dependent Hamiltonians.
\begin{lemma}
\label{le:first_order_trotter}
Given a unitary, \begin{equation}U(T) = \mathcal T \exp \left(-i\int^T_0 \text d t' H(t')\right),\end{equation} %
and its first-order Trotterization, \begin{eqnarray}U^{(1)}(T) &=& \prod_{k=1}^r \exp \left(-i \int_{(k-1) \delta t}^{k \delta t} \text d t' H_2(t')\right)\\ &\times&\exp \left(-i \int_{(k-1) \delta t}^{k \delta t} \text d t' H_1(t')\right),\nonumber \end{eqnarray} where \(\delta t = T/r\), then
  \begin{equation}
      \|U(T) - U^{(1)}(T)\| \le \|[H_1,H_2]\| T\delta t  + \mathcal O(T \delta t^2).
  \end{equation}
\end{lemma}
See \cite{Huyghebaert90} and Appendix B3 in~\cite{Brady21} for a derivation.
Note that this result is the same as that obtained from bounding first-order Trotter error for time-independent Hamiltonians.
In fact, there is no actual dependence on adiabaticity in the proof of Lemma~\ref{le:first_order_trotter}, which means that it holds for all time-dependent Hamiltonians.
We will see that this does not remain true for the traditional second-order Trotter error.

The main tool that is used in proving Lemma~\ref{le:first_order_trotter} is a single application of a variant on Kubo's formula~\cite{Suzuki85,kubo1957statistical}.
This identity will play a critical role in the second-order generalization in Lemma~\ref{le:second_order_trotter} and thus we provide a self-contained statement and proof.
\begin{lemma}
  \label{le:kubo}
  The following identity is true,
\begin{equation}
  \left[\exp\left(i \int_{t_{k-1}}^{t_k} \text d t' A(t')\right), B\right] = \int_{t_{k-1}}^{t_k} \text d s \exp\left(i \int_{t_{k-1}}^{s} \text d t' A(t')\right) \left[iA(s),B\right] \exp\left(-i \int_{t_k}^{s} \text d t' A(t')\right).~\label{eq:kubo_formula}
\end{equation}
\end{lemma}

\begin{proof}
  \begin{subequations}
  \begin{align}
    \left[\exp\left(i \int_{t_{k-1}}^{t_k} \text d t' A(t')\right), B\right] = & \exp\left(i \int_{t_{k-1}}^{t_k} \text d t' A(t')\right)B - B \exp\left(i \int_{t_{k-1}}^{t_k} \text d t' A(t')\right) \\
    =& \left(\exp\left(i \int_{t_{k-1}}^{s} \text d t' A(t')\right) B \exp\left(-i \int_{t_k}^{s} \text d t' A(t')\right)\right) \Biggr|_{s=t_{k-1}}^{s=t_k} \\
     =& \int_{t_{k-1}}^{t_k} \text{d} s \der{}{s} \left[\exp\left(i \int_{t_{k-1}}^{s} \text d t' A(t')\right)\right]B \exp\left(-i \int_{t_k}^{s} \text d t' A(t')\right) \nonumber \\
     &+\int_{t_{k-1}}^{t_k} \text{d} s \exp\left(i \int_{t_{k-1}}^{s} \text d t' A(t')\right)B \der{}{s} \left[\exp\left(-i \int_{t_k}^{s} \text d t' A(t')\right)\right] \\
    =& \int_{t_{k-1}}^{t_k} \text d s \exp\left(i \int_{t_{k-1}}^{s} \text d t' A(t')\right) \left[iA(s),B\right] \exp\left(-i \int_{t_k}^{s} \text d t' A(t')\right)
  \end{align}
  \end{subequations}
  \qed
\end{proof}

We next obtain a bound similar to the one in Lemma~\ref{le:first_order_trotter}, but for second-order Trotter error.
\begin{lemma}
\label{le:second_order_trotter}
Given a unitary, \begin{equation}U(T) = \mathcal T \exp \left(-i\int^T_0 \text d t' H(t')\right),\end{equation} where \(H(t) = H_1(t) + H_2(t)\), \(H_1(t) = (1-u(t))H_1\), \(H_2(t) = u(t) H_2\), \(u(t) = s(t/T)\) and \(s\) is a function from \([0,1]\) to \([0,1]\), and its second-order Trotterization, \begin{eqnarray}U^{(2)}(T) &=& \prod_{k=1}^r \exp \left(-i \int_{(k-1/2) \delta t}^{k \delta t} \text d t' H_1(t')\right)\\ &&\times \exp \left(-i \int_{(k-1) \delta t}^{k \delta t} \text d t' H_2(t')\right)\nonumber\\
&& \times \exp \left(-i \int_{(k-1)\delta t}^{(k-1/2) \delta t} \text d t' H_1(t')\right),\nonumber\end{eqnarray} where \(\delta t = T/r\), then in the limit that \(T\rightarrow \infty\) and \(H'_1(t),\, H_1''(t), \, H'_2(t),\, H''_2(t) \in \mathcal O(T^{-1})\),
  \begin{align}
    &\|U(T) - U^{(2)}(T)\|\nonumber\\
    &\le \left(\frac{1}{2}\|[H_1,[H_1,H_2]]\|+\|[H_2,[H_2,H_1]]\|\right) \frac{1}{12} T \delta t^2 + \mathcal O(T\delta t^3).
  \end{align}
\end{lemma}
A proof where \(H_2\) is assumed to be not time-independent can be found in~\cite{Huyghebaert90}.
Otherwise, the general proof follows.

\begin{proof}

We begin by considering a single arbitrary time step on the interval $\left[(k-1)\delta t,k \delta t\right]=\left[t_{k-1},t_k\right]$, being careful to explicitly indicate that $t_j=j \delta t$, \textit{e.g.,} $t_{k-1/2} = (k-1/2)\delta t$.
The Trotter error accumulated in this interval is $\epsilon(t_k,t_{k-1})=\|U(t_k,t_{k-1})-U^{(2)}(t_k,t_{k-1})\|$, where
\begin{equation}
U(t_k,t_{k-1}) = \mathcal T \exp \left(-i\int^{t_k}_{t_{k-1}} \text d t' H(t')\right),\end{equation}
and
\begin{align}
U^{(2)}(t_k, t_{k-1}) =& \exp\left(-i \int^{t_k}_{t_{k-1/2}} \text d t' H_1(t') \right) \nonumber\\
                         & \times \exp\left(-i \int^{t_k}_{t_{k-1}} \text d t' H_2(t') \right) \\
                         & \times \exp\left(-i \int^{t_{k-1/2}}_{t_{k-1}} \text d t' H_1(t') \right). \nonumber
\end{align}
We are also careful to note that we are working with the ``natural'' time ordering for which $t_k \geq t_{k-1}$.

By unitarity we can rewrite $\epsilon(t_k,t_{k-1})$ as
\begin{equation}
\epsilon(t_k,t_{k-1}) = \|\left(U^{(2)}(t_k,t_{k-1})\right)^{\dagger} U(t_k,t_{k-1})-\mathds{1}\|, \label{eq:single_step_error}
\end{equation}
noting that it is intuitive to think of this as reframing the error as the extent to which the inverse of $U^{(2)}$ fails to invert $U$.
Next, we introduce the notation $F(t_{k},t_{k-1}) = (U^{(2)}(t_k,t_{k-1}))^\dagger U(t_k,t_{k-1})$.
Because $F(t_{k},t_{k-1})$ is unitary, we expect (and subsequently verify) that it  obeys a first-order differential equation of the form
\begin{equation}
 \frac{\partial F(t_{k},t_{k-1})}{\partial t_k} = C(t_{k},t_{k-1}) F(t_{k},t_{k-1}). \label{eq:F_eom}
\end{equation}
Here, the notation $C$ is intended to presage that this will involve many commutators.

Integrating both sides of Eq.~\ref{eq:F_eom} we find that
\begin{equation}
 F(t_k, t_{k-1}) - \mathds{1} = \int^{t_k}_{t_{k-1}} dt' C(t',t_{k-1}) F(t',t_{k-1}) = \left(U^{(2)}(t_k,t_{k-1})\right)^{\dagger} U(t_k,t_{k-1})-\mathds{1},
\end{equation}
which can be inserted into Eq.~\ref{eq:single_step_error} to yield
\begin{equation}
\epsilon(t_k,t_{k-1}) = \|\int^{t_k}_{t_{k-1}} dt' C(t',t_{k-1}) F(t',t_{k-1})\|. \label{eq:single_step_error_integral}
\end{equation}
This is upper bounded by
\begin{equation}
\epsilon(t_k,t_{k-1}) \leq \int^{t_k}_{t_{k-1}} dt' \|C(t',t_{k-1}) F(t',t_{k-1})\| \leq \int^{t_k}_{t_{k-1}} dt' \|C(t',t_{k-1})\| \|F(t',t_{k-1})\| \leq \int^{t_k}_{t_{k-1}} dt' \|C(t',t_{k-1})\|, \label{eq:single_step_C_bound}
\end{equation}
where we have invoked the triangle inequality, submultiplicativity of the operator norm, and unitarity of $F$.
We are thus left with the problem of expressing $C$ consistent with Eq.~\ref{eq:F_eom} and bounding the integral of its norm.

To determine $C(t',t_{k-1})$ we explicitly differentiate $F(t_k,t_{k-1})$ with respect to $t_k$ and substitute $t_k\rightarrow t'$.
It is then convenient to first consider $F(t_k,t_{k-1})=f_1(t_k,t_{k-1}) f_2(t_k,t_{k-1}) f_3(t_k,t_{k-1}) f_4(t_k,t_{k-1})$ in full,
\begin{equation}
\exp\left(i \int^{t_{k-1/2}}_{t_{k-1}} \text d t' H_1(t') \right) \times \exp\left(i \int^{t_k}_{t_{k-1}} \text d t' H_2(t') \right) \times \exp\left(i \int^{t_k}_{t_{k-1/2}} \text d t' H_1(t') \right) \times \mathcal T \exp \left(-i\int^{t_k}_{t_{k-1}} \text d t' H(t')\right), \label{exp:F_long}
\end{equation}
where \(t_{k-1/2} = (t_k - t_{k-1})/2\) and so is treated as a function of \(t_k\) and \(t_{k-1}\). It is straightforward to apply the product rule to differentiate this expression, yielding a sum of four terms with the first being the one in which the leftmost factor ($f_1(t_k,t_{k-1})$) is differentiated and the fourth term being the one in which the rightmost factor ($f_4(t_k,t_{k-1})$) is differentiated.
While the first term will then be proportional to $F(t_{k},t_{k-1})$, the second through fourth terms will need to be further manipulated to appear proportional to $F(t_k,t_{k-1})$ due to the lack of commutativity among the factors in expression \ref{exp:F_long}.
Such a manipulation is realized by applying the identity $AB=BA+[A,B]$,
\begin{align}
\frac{\partial F(t_{k},t_{k-1})}{\partial t_k} = & i\left(\frac{1}{2}H_1(t_{k-1/2})+H_2(t_k)+H_1(t_k)-\frac{1}{2}H_1(t_{k-1/2})-H_1(t_k)-H_2(t_k)\right) F(t_{k},t_{k-1}) \nonumber \\
+ &\left[f_1(t_k,t_{k-1}),iH_2(t_k)\right]f_2(t_k,t_{k-1})f_3(t_k,t_{k-1})f_4(t_k,t_{k-1}) \nonumber \\
+ &\left[f_1(t_k,t_{k-1})f_2(t_k,t_{k-1}),i\left(H_1(t_k)-\frac{1}{2}H_1(t_{k-1/2})\right)\right]f_3(t_k,t_{k-1})f_4(t_k,t_{k-1}) \nonumber \\
+ &\left[f_1(t_k,t_{k-1})f_2(t_k,t_{k-1})f_3(t_k,t_{k-1}),-i\left(H_1(t_k)+H_2(t_k)\right)\right]f_4(t_k,t_{k-1}),
\end{align}
where we have retained a combination of $H_1$ and $H_2$ that obviously cancel in the first line, as to be suggestive of where the other three lines came from.
Eliminating the first line, applying the identities $[AB,C]=A[B,C]+[A,C]B$ and $[ABC,D]=AB[C,D]+A[B,D]C+[A,D]BC$, and suppressing arguments on $f_j$,
\begin{align}
\frac{\partial F(t_{k},t_{k-1})}{\partial t_k} = &\left[f_1,iH_2(t_k)\right]f_2f_3f_4 + \left( f_1\left[f_2,i\left(H_1(t_k)-\frac{1}{2}H_1(t_{k-1/2})\right)\right]f_3f_4 + \left[f_1,i\left(H_1(t_k)-\frac{1}{2}H_1(t_{k-1/2})\right)\right]f_2f_3f_4\right) \nonumber \\
& - f_1f_2\left[f_3,i\left(H_1(t_k)+H_2(t_k)\right)\right]f_4 - \left( f_1\left[f_2,i\left(H_1(t_k)+H_2(t_k)\right)\right]f_3f_4 + \left[f_1,i\left(H_1(t_k)+H_2(t_k)\right)\right]f_2f_3f_4 \right).
\end{align}
We can simplify this expression by eliminating commutators that involve commuting quantities (\textit{e.g.,} $\left[f_1,H_1\right]=0$),
\begin{align}
\frac{\partial F(t_{k},t_{k-1})}{\partial t_k} = &\left[f_1,iH_2(t_k)\right]f_2f_3f_4 +f_1\left[f_2,i\left(H_1(t_k)-\frac{1}{2}H_1(t_{k-1/2})\right)\right]f_3f_4 \nonumber \\
& - f_1f_2\left[f_3,iH_2(t_k)\right]f_4 - f_1\left[f_2,iH_1(t_k)\right]f_3f_4 - \left[f_1,iH_2(t_k)\right]f_2f_3f_4,
\end{align}
and among the remaining six terms all but the third and fourth cancel,
\begin{equation}
\frac{\partial F(t_{k},t_{k-1})}{\partial t_k} = -\frac{i}{2} f_1\left[f_2,H_1(t_{k-1/2})\right]f_3f_4 - i f_1f_2\left[f_3,H_2(t_k)\right]f_4.
\end{equation}

We now apply Eq.~\ref{eq:kubo_formula}, which leads to
\begin{align}
\frac{\partial F(t_{k},t_{k-1})}{\partial t_k} = &\frac{1}{2} f_1 \int^{t_k}_{t_{k-1}} \text d s \exp\left(i \int^{s}_{t_{k-1}} \text d t' H_2(t') \right) \left[H_2(s),H_1(t_{k-1/2})\right] \exp\left(-i \int^s_{t_k} \text d t' H_2(t') \right) f_3 f_4 \nonumber \\
&+ f_1 f_2 \int^{t_k}_{t_{k-1}} \text d s \exp\left(i\int^{s}_{t_{k-1/2}} \text d t' H_1(t') \right) \left[H_1(s),H_2(t_k)\right] \exp\left(-i\int^s_{t_k} \text d t' H_1(t') \right) f_4. \label{eq:F_eom_post-Kubo}
\end{align}
We see that the first (second) line of Eq.~\ref{eq:F_eom_post-Kubo} can be made proportional to $F(t_k,t_{k-1})$ by inserting $(f_1 f_2)^{-1}f_1 f_2$ ($(f_1 f_2 f_3)^{-1} f_1 f_2 f_3$).
This leaves us with an expression for $C(t_k,t_{k-1})$,
\begin{equation}
C(t_k,t_{k-1}) = f_1 I_{A} (f_1 f_2)^{-1} + f_1 f_2 I_{B} (f_1 f_2 f_3)^{-1}, \label{eq:C_def}
\end{equation}
where
\begin{subequations}
\begin{align}
I_A = &\frac{1}{2} \int^{t_k}_{t_{k-1}} \text d s \exp\left(i \int^{s}_{t_{k-1}} \text d t' H_2(t') \right) \left[H_2(s),H_1(t_{k-1/2})\right] \exp\left(-i \int^s_{t_k} \text d t' H_2(t') \right)~\text{and} \\
I_B = &\int^{t_k}_{t_{k-1}} \text d s \exp\left(i\int^{s}_{t_{k-1/2}} \text d t' H_1(t') \right) \left[H_1(s),H_2(t_k)\right] \exp\left(-i\int^s_{t_k} \text d t' H_1(t') \right).
\end{align}
\end{subequations}
It is convenient to further simplify this by again exploiting the unitary invariance of the relevant norm,
\begin{equation}
\|C(t_k,t_{k-1})\| = \|f_2^{-1} I_A + I_B f_3^{-1}\|. \label{eq:C_norm_def_simple}
\end{equation}

However, we need $\|C(t',t_{k-1})\|$ in order to bound the integral in Eq.~\ref{eq:single_step_C_bound}.
It is straightforward but tedious to substitute $t_k$ with the integration variable from Eq.~\ref{eq:single_step_C_bound}, noting the need to disambiguate this use of $t'$ from its use in integrals internal to the definition of $C(t_k,t_{k-1})$ in Eq.~\ref{eq:C_norm_def_simple}.
This leads us to
\begin{subequations}
\begin{align}
\epsilon(t_k,t_{k-1}) \leq \int^{t_k}_{t_{k-1}} dt' &\|C(t',t_{k-1})\|, \\
\leq \int^{t_k}_{t_{k-1}} dt' &\| \frac{1}{2} \int^{t'}_{t_{k-1}} \text d s \exp\left(i \int^{s}_{t'} \text d t'' H_2(t'') \right) \left[H_2(s),H_1((t'+t_{k-1})/2)\right] \exp\left(-i \int^s_{t'} \text d t'' H_2(t'') \right) \nonumber \\
+ &\int^{t'}_{(t'+t_{k-1})/2} ds \exp\left(i \int_{(t'+t_{k-1})/2}^{s} dt'' H_1(t'')\right) \left[H_1(s),H_2(t')\right] \exp\left(-i \int_{(t'+t_{k-1})/2}^{s} dt'' H_1(t'')\right) \|.
\end{align}
\end{subequations}

The remaining work of bounding this integral requires using the identity
\begin{equation}
\exp\left(A\right)B\exp\left(-A\right) = B + \left[A,B\right] + \frac{1}{2}\left(A\left[A,B\right] + \left[B,A\right]A\right)+\ldots \label{eq:commutator_series}
\end{equation}
We can truncate beyond singly nested commutators to arrive at a bound that is $\mathcal{O}(\delta t^4)$, \textit{e.g.,} the terms involving doubly-nested commutators will involve a 4-dimensional integral, and so on.
Retaining the indicated terms we see that
\begin{equation}
\epsilon(t_k,t_{k-1}) \leq \int^{t_k}_{t_{k-1}} dt' \left(\Gamma_1(t') + \Gamma_2(t')\right) + \mathcal{O}(\delta t^4),
\end{equation}
where $\Gamma_1$ involves integrals with a single commutator and $\Gamma_2$ involves integrals with singly nested commutators, specifically
\begin{subequations}
\begin{align}
 \Gamma_1(t') = & \|\int_{t_{k-1}}^{t'} \text d s \frac{1}{2}\left[H_2(s),H_1\left(\frac{t'+t_{k-1}}{2}\right)\right] + \int_{(t'+t_{k-1})/2}^{t'} \text d s \left[H_1(s),H_2(t')\right]\|\label{eq:gamma1_original}~\text{and} \\
 \Gamma_2(t') = & \|\int_{t_{k-1}}^{t'} \text d s \int_{t'}^{s} \text d s_1 \frac{1}{2}\left[H_2(s_1),\left[H_2(s),H_1\left(\frac{t'+t_{k-1}}{2}\right)\right]\right] + \int_{(t'+t_{k-1})/2}^{t'} \text d s \int_{(t'+t_{k-1})/2}^{s} \text d s_1 \left[H_1(s_1),\left[H_1(s),H_2(t')\right]\right]\|.
\end{align}
\end{subequations}

In Eq.~\ref{eq:gamma1_original}, the limits of integration in the two integrals can be made identical through a change of variables
\begin{equation}
\Gamma_1 = \|\frac{1}{2}\int_{t_{k-1}}^{t'} \text d s \lbrace\left[H_2(s),H_1\left(\frac{t'+t_{k-1}}{2}\right)\right] +  \left[H_1\left(\frac{s+t'}{2}\right),H_2(t')\right]\rbrace\|,
\end{equation}
which will also be useful in bounding $\Gamma_2$.
The factor of $H_2(s)$ can be Taylor expanded to second order about $s=t'$ and the factor of $H_1((s+t')/2)$ to second order about $s=t_{k-1}$, such that the commutators can be combined.
The zeroth-order terms cancel, leaving only contributions that depend on derivatives of the Hamiltonians (\textit{i.e.,} $H_1'$, $H_2'$, $H_1''$, and $H_2''$).
These terms can be replaced by constants that take the largest value that they would have over the domain of integration, and the remaining integral can be evaluated and used as an upper bound.
This same general strategy can be applied to upper bound $\Gamma_2$.

After upper bounding $\Gamma_1$ and $\Gamma_2$ on an arbitrary interval, we can apply the triangle inequality to upper bound the cumulative error over all steps to arrive at
\begin{align}
	\|U(T) - U^{(2)}(T)\| \le &\left(\frac{1}{2}\|[H_1,[H_1,H_2]]\|+\|[H_2,[H_2,H_1]]\|\right) \frac{1}{12} T \delta t^2 \nonumber \\
	&+\left(\|[H_1,H_2']\|+\|[H_1,H_2''\| + \frac{1}{2}\|[H_1',H_2]\|+\frac{1}{2}\|[H_1'',H_2]\|\right)\frac{1}{12} T \delta t^2 \nonumber \\
	&+\left(\frac{1}{2}\|[H_1',[H_2,H_1]]\|+\|[H_2',[H_2,H_1]]\|\right) \frac{1}{12} T\delta t^2 \nonumber \\
    &+\mathcal O(T \delta t^3).
\end{align}
The second line comes from $\Gamma_1$ and the third line comes from $\Gamma_2$.
\(H_1', \, H_2',\, H_1'', \, H_2'' \in \mathcal O(T^{-1})\), and so the first two terms are \(\mathcal O(\delta t^2 T)\) while the remaining terms contribute \(\mathcal O(\delta t^2)\). Here we are considering the limit \(T \rightarrow \infty\) and a fixed \(\delta t\) that produces a convergent Trotterization (i.e.~\(\delta t \in \mathcal O(\min_t \|H[u(t)]\|^{-1})\)). Therefore, in the adiabatic limit (\(T \rightarrow \infty\)), the first two terms dominate and we arrive at the desired upper bound,
\begin{align}
	\|U(T) - U^{(2)}(T)\| \le &\left(\frac{1}{2}\|[H_1,[H_1,H_2]]\|+\|[H_2,[H_2,H_1]]\|\right) \frac{1}{12} T \delta t^2 \nonumber \\
    &+\mathcal O(T \delta t^3).
\end{align}

\end{proof}

Note that the leading order terms in the adiabatic limit are the same bound as the traditional second-order Trotter error bound for time-independent Hamiltonians.
However, unlike Lemma~\ref{le:first_order_trotter}, this is only due to the adiabatic limit allowing for dropping higher order derivatives of the Hamiltonian; unlike Lemma~\ref{le:first_order_trotter}, Lemma~\ref{le:second_order_trotter} is not valid outside of that limit for general time-dependent Hamiltonians.

\begin{lemma}[Trotter Error and Infidelity]
\label{le:adiabatic_fidelity}
Given a unitary $U(T) = \mathcal T \exp \left(-i\int^T_0 \text d t' H(t')\right)$ where \(H[u(t)] = (1-u(t)) H_1 + u(t) H_2\) is gapped, and its first-order Trotterization $U^{(1)}(T)$, then in the limit that \(T \rightarrow \infty\) state infidelity \(\mathcal{I}(T)\equiv 1 - |\expval{\psi}{U^\dagger(\infty) U^{(1)}(T)}{\psi}|\) is upper bounded by the squared sum of its Trotter operator-norm error $\|U(T) - U^{(1)}(T)\|$ and an \(\mathcal O(T^{-1})\) non-adiabatic correction:
\begin{equation}
\mathcal I(T) \le \|U(T) - U^{(1)}(T)\|^2 + \mathcal O(\|U(T) - U^{(1)}(T)\| T^{-1}) + \mathcal O(T^{-2}).
\end{equation}
\end{lemma}

\begin{proof}
Let \(\{\ket{\phi_i}\}_i\) be the eigenstates of \(H[u(0)]\). For all \(\ket{\phi_j}\),
    \begin{align}
      \mathcal I(T) \le& 2 (1 - |\expval{\psi}{U^\dagger(\infty) U^{(1)}(T)}{\psi}|)\\
      \le& 2 (1 - |\Re \expval{\psi}{U^\dagger(\infty) U^{(1)}(T)}{\psi}|)\\
      \le& 2 (1 - \Re \expval{\psi}{U^\dagger(\infty) U^{(1)}(T)}{\psi})\\
         =& \max_i \expval{\phi_i}{(U(\infty) - U^{(1)}(T))^\dagger(U(\infty) - U^{(1)}(T))}{\phi_i}\\
      =& \|U(\infty) - U^{(1)}(T)\|^2\\
        \le& (\|U(T) - U^{(1)}(T)\| + \|U(\infty) - U(T)\|)^2\\
        \le& (\|U(T) - U^{(1)}(T)\| + \|(I - P_i)U(T)\ket{\phi_i}\|)^2\\
      \le& \|U(T) - U^{(1)}(T)\|^2 + \mathcal O(\|U(T) - U^{(1)}(T)\| T^{-1}) + \mathcal O(T^{-2}),
    \end{align}
    where we bounded \(\|U(\infty) - U(T)\|\) by the approximation error \(\|(I - P_i)U(T)\ket{\phi_i}\|\) for \(P_i\) the projector on the final \(i\)th eigenstate state, which we subsequently bounded by \(\mathcal O(T^{-1})\)~\cite{Mackenzie06,Cheung11} assuming \(T \rightarrow \infty\). The coefficients of this term depend on the minimal energy gap at points during the evolution but we omit these as we are only interested in its scaling with \(T\). This term is the bound on continuous-time evolved infidelity for finite time \(T\), \textit{i.e.,} when there is no Trotter error you obtain the same infidelity as from un-Trotterized evolution.

    \qed
\end{proof}

We combine Lemmas~\ref{le:first_order_trotter}-\ref{le:adiabatic_fidelity} for a time-dependent version of the first-order Trotter error:
\begin{theorem}
\label{th:layden_timedep}
Given a gapped Hamiltonian \(H[u(t)] = (1-u(t)) H_1 + u(t) H_2\) %
where \(U(T) = \mathcal T e^{-i\int_0^T H[u(t')]\text d t'}\), \(u(t) = s(t/T)\) and \(s\) is a function from \([0,1]\) to \([0,1]\), if this unitary is first-order Trotterized into \(H_1\) and \(H_2\),
\begin{eqnarray}U^{(1)}(T) &=& \prod_{k=1}^r \exp \left(-iH_2 \int_{(k-1) \delta t}^{k \delta t} \text d t' u_k(t')\right)\\ &\times&\exp \left(-iH_1 \int_{(k-1) \delta t}^{k \delta t} \text d t' (1-u_k(t'))\right),\nonumber \end{eqnarray}
where \(u_k = u((k-1)\delta t)\) and $\delta t = T/r$ is fixed, then in the limit that \(T \rightarrow \infty\) state infidelity \(\mathcal{I}(T)\equiv 1 - |\expval{\psi}{U^\dagger(\infty) U^{(1)}(T)}{\psi}|\) is upper bounded by
\begin{align}
  \min \{C_2 T^2\delta t^2, (C_1 \delta t + C_3 T \delta t^2)^2, 2 \|I\|\} + \mathcal O(\delta t) + \mathcal O(T^{-2}), \nonumber
\end{align}
where
\begin{eqnarray}
  C_1 &=& \min \{\|H_1\|, \|H_2\|\}, \\
  C_2 &=& \frac{1}{4} \|[H_1, H_2]\|^2, \\
  C_3 &=& \frac{1}{12}\left[ \min S + \frac{1}{2} \max S\right],
\end{eqnarray}
\begin{equation}
  S = \{\|[H_1,[H_1,H_2]]\|, \|[H_2,[H_2,H_1]]\|\}.
\end{equation}
The coefficients of the \(\mathcal O(\delta t)\) and \(\mathcal O(T^{-2})\) terms depend on the minimal energy gap at points during the evolution.
\end{theorem}
\begin{proof}
  Since Lemma~\ref{le:first_order_trotter} and Lemma~\ref{le:second_order_trotter} are the same as their time-independent versions, the same proof strategy as employed in~\cite{Layden21} can be used. Specifically, it can be shown that \(\|U(T)-U^{(1)}(T)\| \le \|U(T) - U^{(2)}(T)\| + \max_r \|[U(T), \exp \left(-iH_2 \int_{(r-1) \delta t}^{r \delta t}\text dt' u_r(t')\right) ]\|\), where \(U^{(2)}(T)\) is the second-order Trotterization defined in Lemma~\ref{le:second_order_trotter}. The first term can be bounded by appealing to Lemma~\ref{le:second_order_trotter}. The second term can be upper bounded by applying Kubo's formula to produce \(\frac{\delta t}{2} \|[U(T), H_2]\|\), which can in turn be upper bounded by \(\delta t\|H_2\|\), \(\delta t\|H_1\|\), or \(\frac{T \delta t}{2} \|[H_1, H_2]\|\) (by another application of Kubo's formula and then an appeal to Lemma~\ref{le:first_order_trotter}). Repeating the same steps with \(H_1\) and \(H_2\) reversed produces the final result.

  This produces the following bound on the first-order Trotter error squared:
\begin{align}
  &\|U(T) - U^{(1)}(T)\|^2\\
  \leq& \min \{C_2 T^2\delta t^2, (C_1 \delta t + C_3 T \delta t^2)^2, 2 \|I\|\}, \nonumber
\end{align}
where
\begin{eqnarray}
  C_1 &=& \min \{\|H_1\|, \|H_2\|\}, \\
  C_2 &=& \frac{1}{4} \|[H_1, H_2]\|^2, \\
  C_3 &=& \frac{1}{12}\left[ \min S + \frac{1}{2} \max S\right],
\end{eqnarray}
\begin{equation}
  S = \{\|[H_1,[H_1,H_2]]\|, \|[H_2,[H_2,H_1]]\|\},
\end{equation}

We then appeal to Lemma~\ref{le:adiabatic_fidelity} to obtain the stated bound on the state infidelity \(\mathcal I(T)\) where the \(\mathcal O(\delta t)\) term comes from multiplying the \(\mathcal O(T \delta t)\) worst-case bound of \(\|U(T) - U^{(1)}(T)\|\) with the \(\mathcal O(T^{-1})\) non-adiabatic error.\qed

\end{proof}

Notice that the bound of \(2 \|I\|\) corresponds to the regime when the Trotter expansion diverges, and so this result is general (valid for any $\delta t, T$).

\section{Theorem 2 - Statement, proof, and numerical corroboration of short-time scaling}
\label{app:adiabatic_trotter_error}
\renewcommand{\theequation}{B\arabic{equation}}
\renewcommand{\thefigure}{B\arabic{figure}}
\setcounter{figure}{0}

\begin{theorem}
  Given a gapped Hamiltonian \(H[u(t)] = (1-u(t)) H_1 + u(t) H_2\) and a unitary \(U(T) = \mathcal T e^{-i\int_0^T H(t')\text d t'}\), where \(u(t) = s(t/T)\) and \(s\) is a function from \([0,1]\) to \([0,1]\) that is infinitely differentiable, if $U(t)$ is first-order Trotterized into \(H_1\) and \(H_2\) terms with fixed time steps \(\delta t \in \mathcal O(\min_t \|H[u(t)]\|^{-1})\), then in the limit that \(T \rightarrow \infty\)  with \(u(0) \rightarrow 0\) and \(u(T) \rightarrow 1\) the final state infidelity is bounded by \(\mathcal O (T^{-2} \delta t^2) + \mathcal O(T^{-2})\). Moreover, initially at a given fixed \(t/T\ll 1\), state infidelity increases as \(\mathcal O(t^2 \delta t^2)\).
The coefficient of the \(\mathcal O(T^{-2})\) term depends on the minimal energy gap at points during the evolution.
\label{th:adiabatic_trotter_error}
\end{theorem}
\begin{proof}

  Consider a \(d\)-dimensional \(H[u(t)] = (1-u(t)) H_1 + u(t) H_2\), where \(u(t) \in [0,\,1]\).

A single first-order Trotter time step \(\delta t\) of \(\ket{\psi(t)}\) produces \(\ket{\psi(t+\delta t)} = U^{(1)}(t, \delta t, T) \ket{\psi(t)}\), where \(T\) is the total time. Therefore, in the basis of instantaneous eigenstates \(\{\ket{\phi_i(t/T)}\}_{i=1,\ldots,d}\) (with corresponding energies \(\{E_i(t/T)\}_{i=1,\ldots,d}\)), it has amplitude of remaining in the \(i\)th state at time \(t+\delta t\),
\begin{align}
  \label{eq:ground_state_amp}
    B_i(t+\delta t,t/T) = \sum_{j,k} A_{kj}(t, \delta t, T) B_{j}(t,T) C_{ik}(t, \delta t, T)
  \end{align}
  where
  \begin{align}
    A_{ij}(t, \delta t, T) =& \expval{\phi_i(t/T)}{U^{(1)}(t, \delta t,T)}{\phi_j(t/T)}, \\
    B_{i}(t,t/T) =& \braket{\phi_i(t/T)}{\psi(t)},\\
    C_{ij}(t, \delta t, T) =& \braket{\phi_i((t+\delta t)/T)}{\phi_j(t/T)},
  \end{align}
  and \(U^{(1)}(t,\delta t, T)\) is a first-order Trotterization for time step \(\delta t\) at time \(t\): \(U^{(1)}(t,\delta t, T) = U_2(t,t+\delta t) U_1(t, t+\delta t)\). (\(U^{(1)}(t, \delta t, T)\) corresponds to the \(k=t/\delta t+1\) term in Eq.~\ref{eq:Utrot}). Notice that we work in terms of the adiabatic basis here, which we paramaterize by the normalized time \(0\le t/T \le 1\).

  In the case of continuous evolution through avoided crossings, given an initial state in the ground state (\textit{i.e.,} \(B_1(0,0)=1\)), we can find from adiabatic theory~\cite{Mackenzie06,Cheung11} that the final amplitudes \(B_1(0,1)\) and \(B_1(T,1)\) are related by \(\mathcal O (T^{-1})\), whose prefactor depends on the time derivatives of the Hamiltonians and energy gap. This can be a loose bound.%

  We set \(\ket{\psi(T)} = U^{(1)}(T)\ket{\psi(0)}\) with \(\ket{\psi(0)} = \ket{\phi_1(0)}\). It follows that our definition of infidelity, \(\mathcal I(T)\), is bounded by \(\sum_{i\ne 1}|B_i(T,1)|^2\):
  \begin{align}
    \mathcal I(T) =& 1 - |\braket{\phi_1(1)}{\psi(T)}|\\
    =& \sum_{i} |\braket{\phi_i(1)}{\psi(T)}|^2 - |\braket{\phi_1(1)}{\psi(T)}|\\
    \le& \sum_{i \ne 1} |\braket{\phi_i(1)}{\psi(T)}|^2 + |\braket{\phi_1(1)}{\psi(T)}| \nonumber\\
    & - |\braket{\phi_1(1)}{\psi(T)}|\\
    =& \sum_{i \ne 1} |\braket{\phi_i(1)}{\psi(T)}|^2\\
    =& \sum_{i\ne 1}|B_i(T,1)|^2.
  \end{align}

  First let us note the scaling of the \(A_{ij}\) and \(C_{ij}\) terms.
  \begin{align}
    A_{ii}(t, \delta t, T) =& 1 - i E_i(t/T) \delta t + \mathcal O(\delta t^2),\\
    C_{ii}(t, \delta t, T) =& 1 + \mathcal O(\delta t/T),
  \end{align}
  and for \(j \ne i\)
  \begin{align}
    A_{ij}(t, \delta t, T) =& \mathcal O(\delta t^2),\\
    C_{ij}(t, \delta t, T) =&  \braket{\dot \phi_i(u)}{\phi_j(u)} \delta t/T + \mathcal O((\delta t/T)^2).
  \end{align}

  With this in hand, we proceed to redefine the terms in Eq.~\ref{eq:ground_state_amp} so as to rewrite it into a more manageable form.
  First we expand out
  \begin{align}
    &A_{ii}(t, \delta t,T) C_{ii}(t, \delta t,T) \nonumber\\
    =& 1 - i E_i(t/T) \delta t + \mathcal O (\delta t/T, \delta t^2),
  \end{align}
  and for \(j\ne i\)
  \begin{equation}
    A_{ji}(t, \delta t, T) C_{ij}(t, \delta t,T) \in \mathcal O(\delta t^3/T),
  \end{equation}
  and then combine them:
  \begin{align}
    &A_{ii}(t, \delta t,T) C_{ii}(t, \delta t,T) + \Sigma_{j\ne i} A_{ji}(t, \delta t,T) C_{ij}(t, \delta t,T) \nonumber\\
    =& 1 - i E_i(t/T) \delta t + \mathcal O(\delta t/T, \delta t^2),
  \end{align}
  where we define \(\mathcal O(f,g) \equiv \mathcal O(f) + \mathcal O(g)\).

  Considering \(i \ne j\), we proceed to define three more terms:
  \begin{align}
    & A_{jj}(t, \delta t,T) C_{ij}(u, \delta t,T) \nonumber\\
    =&  \braket{\dot \phi_i(t/T)}{\phi_j(t/T)} \expval{\phi_j(t/T)}{U^{(1)}(t,\delta t,T)}{\phi_j(t/T)} \delta t/T \nonumber\\
    & + \mathcal O(\delta t^2/T^2) \\
    \equiv& \mathcal S_{ij}(t, \delta t,T)\delta t + \mathcal O(\delta t^2/T^2),
  \end{align}
  \begin{align}
    &  A_{ij}(t, \delta t,T) C_{ii}(t, \delta t,T) \nonumber\\
    =&  \expval{\phi_i(\nicefrac{t}{T})}{U^{(1)}(t,\delta t,T)}{\phi_j(\nicefrac{t}{T})} (1 + 2 \Re \braket{\dot \phi_i(\nicefrac{t}{T})}{\phi_i(\nicefrac{t}{T})}\nicefrac{t}{T}) \nonumber\\
    & + \mathcal O(\delta t^4/T^2) \\
    \equiv& \mathcal R_{ij}(t, \delta t,T)\delta t + \mathcal O (\delta t^4/T^2),
  \end{align}
  and
  \begin{align}
    &\Sigma_k A_{kj}(t, \delta t,T) C_{ik}(t, \delta t,T) \nonumber\\
    =& \Sigma_k \expval{\phi_k(t/T)}{U^{(1)}(t,\delta t,T)}{\phi_j(t/T)} \braket{\dot \phi_i(t/T)}{\phi_k(t/T)} \delta t/T \\
    & + \mathcal O(\delta t^2/T^2) \nonumber\\
    \label{eq:Q}
    \equiv& \mathcal Q_{ij}(t, \delta t, T) \delta t + \mathcal O(\delta t^2/T^2).
  \end{align}

  Dropping their arguments from now on, it follows that \(\mathcal Q_{ij} \in \mathcal O(1/T)\), \(\mathcal R_{ij} \in \mathcal O (\delta t)\) and \(\mathcal S_{ij} \in \mathcal O(1/T)\).

  With these definitions we can now rewrite Eq.~\ref{eq:ground_state_amp}:
  \begin{align}
    & B_i(t +\delta t,t/T) \nonumber\\
    =& [1 - i E_i(t/T) \delta t + \mathcal O(\nicefrac{\delta t}{T}, \delta t^2) ] B_i(t,t/T)\\
     & + \sum_{j\ne i} [\mathcal Q_{ij}\delta t +\mathcal R_{ij} \delta t + \mathcal S_{ij} \delta t + \mathcal O(\nicefrac{\delta t^2}{T^2})] B_j(t,t/T). \nonumber
  \end{align}
  It follows that
    \begin{align}
      &2i \frac{B_i(t \pm \delta t/2,t/T) - B_i(t,t/T)}{\delta t }\nonumber\\
      =& [ E_i(t/T) + \mathcal O(1/T, \delta t) ] B_i(t,t/T)\\
      & \pm 2i \sum_{j\ne i} [\mathcal Q^\pm_{ij} + \mathcal R^\pm_{ij} + \mathcal S^\pm_{ij} + \mathcal O(\delta t/T^2)] B_j(t,t/T), \nonumber
  \end{align}
    where \(\mathcal Q^\pm_{ij} = \mathcal Q_{ij}(t, \pm \delta t, T)\) and similarly for \(\mathcal R^\pm_{ij}\) and \(\mathcal S^\pm_{ij}\).

  Hence,
    \begin{align}
      &i \frac{B_i(t + \delta t/2,t/T) - B_i(t - \delta t/2,t/T)}{\delta t }\nonumber\\
      =& [ E_i(t/T) + \mathcal O(1/T, \delta t) ] B_i(t, t/T)\\
      & + i \sum_{j\ne i} [\tilde{\mathcal Q}_{ij} + \tilde{\mathcal R}_{ij} + \tilde{\mathcal S}_{ij} + \mathcal O(\delta t/T^2)] B_j(t,t/T), \nonumber\\
      \label{eq:Schrodinger}
      =& i \partder{B_i(t,t/T)}{t} + \mathcal O(\delta t^2),
    \end{align}
    where \(\partder{B_i(t,t/T)}{t}\) is a derivative w.r.t.~the first argument of \(B_i\) only (\textit{i.e.,} \(t/T\) is fixed as another independent variable) and where \(\tilde{\mathcal Q}_{ij} = \mathcal Q^+_{ij} + \mathcal Q^-_{ij}\) and similarly for \(\tilde{\mathcal R}_{ij}\) and \(\tilde{\mathcal S}_{ij}\). The \(\mathcal O(\delta t^2)\) term above is due to the finite difference error of approximating a first-order derivative using the central difference.

    Up to the given errors, Eq.~\ref{eq:Schrodinger} is the Schr\"odinger equation written in terms of the coefficients, \(B_i(t)\), of \(\ket{\psi(t)}\) in terms of the adiabatic basis \(\ket{\phi_i(t)}\). We see that the \(\tilde Q_{ij}\), \(\tilde R_{ij}\), and \(\tilde S_{ij}\) terms are responsible for coupling between energy eigenstates.

    We note that these terms scale the same as their untilde'd versions: \(\tilde{\mathcal Q}_{ij} \in \mathcal O(1/T)\), \(\tilde{\mathcal R}_{ij} \in \mathcal O (\delta t)\) and \(\tilde{\mathcal S}_{ij} \in \mathcal O(1/T)\).

    We set \(B_1(0,0) = 1\) and \(B_j(0,0) = 0\) for all \(j\ne 1\).

    In the following we will be taking the limit \(T\rightarrow \infty\) to bound the effect of these coupling terms in the adiabatic limit.

    Let us assume that the leading order error is from \(\tilde{\mathcal R}_{ij}\). It looks like the most likely candidate since it is only dependent on \(\delta t\). Adding up \(r = T/\delta t\) time steps with the triangle inequality would be expected to produce cumulative error \(\mathcal O(T)\), a quantity that increases with \(T\).

    \(\mathcal R_{ij} \propto \expval{\phi_i(t)}{U^{(1)}(t,\delta t,T)}{\phi_j(t)} = 0\) when \(t=0\) or \(t=T\) because \(U^{(1)}(0,\delta t,T) = \exp(-i H_1 \delta t)\) and \(U^{(1)}(T,\delta t,T) = \exp(-i H_2 \delta t)\). Hence, it can be expanded in a sine series, which is thus also true of \(\tilde{\mathcal R}_{ij}\):
    \begin{equation}
    \tilde{\mathcal R}_{ij} = \sum_m c_m R \sin (\pi m t/T),
    \end{equation}
    where \(R \in \mathcal O(\delta t)\) and \(\sum_m |c_m|^2 = 1\).

    Let us consider one of its sine terms as a leading-order (harmonic) perturbation, rewritten as \(V_{ij}(u) = V e^{i \pi m t/T} + V^\dagger e^{-i\pi m t/T}\) where \(V = -i R/2\), which is turned on at \(t/T \in [0,1]\), on the rest of the Hamiltonian in Schr\"odinger's equation above. It follows from time-dependent perturbation theory that the first-order correction to the \(j\)th eigenstate coefficient is
    \begin{align}
      &B^{(1)}_j(t,t/T) \nonumber\\
      =& \frac{-i R}{2} \left[ \frac{1 - e^{i(m \pi + T\Delta_{j1})t/T}}{m\pi/T + \Delta_{j1}} - \frac{1 - e^{i(-m \pi + T\Delta_{j1})t/T}}{-m\pi/T + \Delta_{j1}} \right],
      \label{eq:Rerror}
    \end{align}
    where \(\Delta_{j1} = E_j(t) - E_1(t)\).
    \begin{align}
      \implies& B^{(1)}_j(T,1)\nonumber\\
      =& \frac{-i R}{2} \left[ \frac{1 - (-1)^m e^{i T\Delta_{j1}}}{\nicefrac{m\pi}{T} + \Delta_{j1}} - \frac{1 - (-1)^m e^{i T\Delta_{j1}}}{-\nicefrac{m\pi}{T} + \Delta_{j1}} \right]\\
      =& -i \frac{m R \pi}{T \Delta^2_{j1}} (1 - (-1)^m e^{i \Delta_{j1} T}) + \mathcal O(\delta t, T^{-2})\\
      \in& \mathcal O(T^{-1}\delta t ).
    \end{align}

    The second-last line was produced by taking the \(T \rightarrow \infty\) limit. As a reminder, the \(\mathcal O(\delta t)\) dependence comes from \(R \in \mathcal O(\delta t)\).

    The same result will occur for any other sine component in the normalized decomposition and so the overall scaling is \(\mathcal O(T^{-1}\delta t )\).

    Therefore, if this is the leading-order error, then the infidelity will scale as \(\sum_{i\ne 1} |B^{(1)}_i(T,1)|^2 \in \mathcal O(T^{-2}\delta t^2)\). Notice that, unlike our intuition from applying the triangle inequality, we found that the cumulative error is decreasing with \(T\). This is because the coefficient of this \(\mathcal O(\delta t^2)\) error term is zero at \(t=0\) and \(t=1\) and so is akin to a harmonic perturbation, like a Rabi oscillation. The frequency of this Rabi oscillation, compared to the energy gap \(\Delta_{j1}\) becomes infinitesimally small when \(T \rightarrow \infty\) and so the probability of excitation decreases as the Lorentzian tail associated with Rabi oscillations: \(\mathcal O(T^{-2}\delta t^2)\).

    Perhaps this means that we did not choose the correct leading-order term. Let us instead assume that the leading order error is from \(\tilde{\mathcal S}_{ij}\). This still looks like a troublesome term since it scales as \(1/T\). Adding up \(r = T/\delta t\) time steps with the triangle inequality would be expected to produce cumulative error \(\mathcal O(1/\delta t)\), a quantity that not only does not decrease with \(T\), but increases with \(\delta t^{-1}\)!

    Unlike \(R_{ij}\), \(S_{ij}\) and is not zero at \(t=0\) or \(t=T\) since it is proportional to a ``diagonal'' expectation value \(\expval{\phi_j(t/T)}{U^{(1)}(t,\delta t,T)}{\phi_j(t/T)}\) instead of an ``off-diagonal'' one. In fact, we can bound \(S_{ij}\) by \(c/T\), where \(c \in \mathbb R\), a quantity independent of \(t\).

    Therefore, let us consider this as a leading-order (constant) perturbation, \(V_{ij}(u) = c/T\), which is turned on at \(t/T \in [0,1]\). It follows from time-dependent perturbation theory that the first-order correction to the \(j\)th eigenstate population is
    \begin{equation}
      |B_j^{(1)}(T,1)|^2 = \frac{4 |c/T|^2}{\Delta^2_{j1}} \sin^2 (\Delta_{j1} T/2) \in \mathcal O(T^{-2}).
    \end{equation}

    This is the same as the scaling found in continuous adiabatic evolution, as discussed earlier.

    Let us finally consider that the leading order error is from the remaining coupling term, \(\tilde{\mathcal Q}_{ij}\). Again, like \(\tilde{\mathcal S}_{ij}\) this can perhaps be troublesome because it scales as \(\mathcal O( 1/T)\) and so after \(r = T/\delta t\) time steps can potentially produce \(\mathcal O(1/\delta t)\) error.

    Like \(S_{ij}\), \(Q_{ij}\) is also not zero at \(t=0\) or \(t=T\) because in its sum (indexed by \(k\) in Eq.~\ref{eq:Q}) contains one ``diagonal'' expectation value \(\expval{\phi_j(t/T)}{U^{(1)}(t,\delta t,T)}{\phi_j(t/T)}\). We can proceed in the same manner as for \(S_{ij}\) and bound it by \(c/T\) and find overall scaling \(\mathcal O(T^{-2})\).

    Note that we can get away considering the untilde'd coefficients above since the tilde'd ones have the same properties: \(\tilde R_{ij}\) is equal to zero at \(t=0\) and \(t=T\) and the others are still upper bounded by \(c/T\), for some \(c \in \mathbb R\). So including them instead does not change our concluded scaling.

    We now discuss the \(\mathcal O(\delta t^2)\) error from the central-difference approximation to the first derivative in Eq.~\ref{eq:Schrodinger}. Going to higher orders in \(\delta t\) in approximating the first derivative of \(B_i\) would raise the power on the \(\delta t\) and produce new coefficients that would be weighted differences of the existing coefficients between their values at \(t\) and \(t+\delta t\) for various values of \(\delta t\). Any such new term replacing \(R_{ij}\) would still necessary be equal to zero at \(t=0\) and \(t=T\) and the other terms would still be upper bounded by \(c/T\), for some \(c \in \mathbb R\). So the conclusions would not change for any higher order approximation to the first derivative, and if the derivative is an analytic function, the error term could be made arbitrarily small by going to sufficiently high-order. In particular, for time \(T\), the order of approximation could be set at \(\Omega(\lceil -2\log_{\delta t}T \rceil)\) (the minus sign is necessary since \(\delta t \|H(t)\| < 1\) to lie within the radius of convergence, \textit{i.e.,} for the Baker-Campbell-Haussdorff (BCH) expansion to be convergent). This procedure would change this error from \(\mathcal O(\delta t^2)\) to \(\mathcal O(\delta t^{\log_{\delta t} T^{-2}}) \in \mathcal O(T^{-2})\) and so be dominated by the other terms.

    Therefore, we can neglect this term and conclude that the final state infidelity is bounded by \(\mathcal O(T^{-2} \delta t^2) + \mathcal O(T^{-2})\).

    Let us now consider a timespan early in the adiabatic evolution at fixed $t/T \ll 1$ as $T \rightarrow \infty$ (i.e.~ $t/T$ is very small but constant as $T$ goes to infinity). %

    We return back to Eq.~\ref{eq:Rerror} to consider the error from \(\mathcal R_{ij}\) but now evaluate it for \(t \ne T\):
    \begin{align}
      &B^{(1)}_j(t,t/T) \nonumber\\
      =& \frac{-i R}{2} \left[ \frac{1 - e^{i(m \pi + T\Delta_{j1})t/T}}{m\pi/T + \Delta_{j1}} - \frac{1 - e^{i(-m \pi + T\Delta_{j1})t/T}}{-m\pi/T + \Delta_{j1}} \right]\\
      =& -im\pi RT t^2/2 T^2 + \mathcal O((t/T)^3)\\
      \in& \mathcal O(\delta t T, (t/T)^3).
    \end{align}
    As before, the second-last line was produced by taking the \(T \rightarrow \infty\) limit.

    Therefore, fixing \(t/T\ll 1\) at some initial value, the leading-order error will scale as \(\sum_{i \ne 1} |B_i^{(1)}(t,t/T)|^2 \in \mathcal O(T^2 \delta t^2 ) = \mathcal O(t^2 \delta t^2)\).

    We can similarly consider the contributions from \(\tilde{\mathcal S}_{ij}\) and \(\tilde{\mathcal Q}_{ij}\) but since we bounded these by \(t/T\)-independent quantities, there is no change from considering intermediate timepoints \(t/T \ll 1\).

    Therefore, initially at a given \(\delta t/T\ll 1\), state infidelity increases as \(\mathcal O(T^2 \delta t^2)\). Note that this bound is different than the ``na\"ive'' first-order Trotter bound of \(\mathcal O(\delta t T)\) that you could expect at any given \(\delta t/T\).
    \qed

  \end{proof}

  The initial \(\mathcal O(T^2 \delta t^2 )\) scaling of Theorem~\ref{th:adiabatic_trotter_error} for fixed \(t/T\ll 1\) is numerically verified in the following Figure~\ref{fig:instantaneous_eigenstate_fidelity}.

  \begin{figure}[H]
  \centering
\includegraphics[scale=0.37]{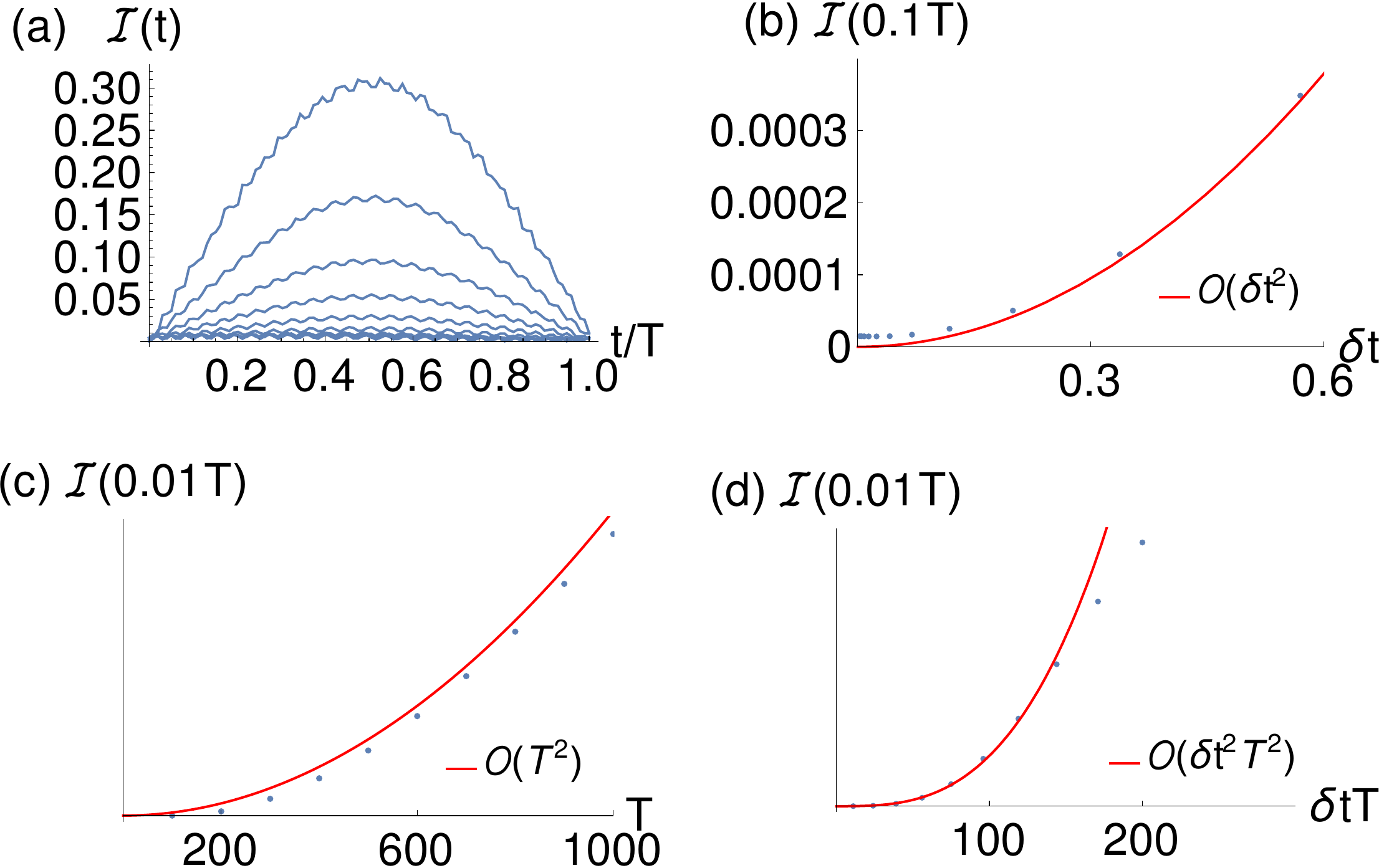}
\caption{(a) Infidelity of Trotterized initial ground state for a simple two-level system (\(H_1 = X\) and \(H_2 = Z\))  at different time steps for full ramp. The larger the time step, the lower the fidelity dips at \(u=1/2\). This figure shows how large Trotter fidelity error incurred at early points in the control curve is cancelled by Trotter error at later points in the control curve. In particular, Theorem~\ref{th:adiabatic_trotter_error} states that the initial scaling is \(\mathcal O(T^2 \delta t^2 )\) for a fixed \(t/T\). This is illustrated (b) for fixed \(t/T=0.1\) and \(T=100\), where the initial infidelity with respect to \(\delta t\) is quadratic, (c) for fixed \(t/T = 0.1\) and \(\delta t=0.1\), where it is found to be quadratic with respect to \(T\) too, and (d) just for fixed \(t/T\), where it is found to be quadratic with respect to \(\delta t T\). For larger than two-level systems, the same behavior is observed between the ground state and any excited state with fully traversed coupling (see Appendix~\ref{app:numerics}).}
\label{fig:instantaneous_eigenstate_fidelity}
\end{figure}

\section{Exploration of variable-time-step Trotterization}
\label{app:counterdiabatic}
\renewcommand{\theequation}{C\arabic{equation}}
\renewcommand{\thefigure}{C\arabic{figure}}
\setcounter{figure}{0}

Here we consider first-order Trotterization of adiabatic unitaries, with the additional freedom that we permit the time steps of the Trotter steps to be variable.

The leading error from first-order Trotterization can be used as a natural counter-diabatic term for reduced infidelity in Trotterized adiabatic evolution, which can correct diabatic errors present at finite \(T\)~\cite{wurtz2022counterdiabaticity}. The magnitude of this correction term can be optimized by allowing for variable Trotter time steps. In this way, the regime of applicability of Theorem~\ref{th:layden_timedep}'s scaling \(\mathcal O(\delta t) + \mathcal O(T \delta t^2)\) can be extended to reach further for both small and large \(\delta t\), corresponding to the green and red regimes in Figure~\ref{fig:partialramps}. A bound on the amount that this regime can be increased for these two extremes of \(\delta t\) is given by the following Theorem:

\begin{theorem}
\label{th:bound}
If Trotterization is allowed to have variable time step, for total time \(T\), the state infidelity \(\mathcal I(T)\) under first-order Trotterization is less than or equal to \(C_2 T^2 \delta t^2\) for a regime in \(\delta t\) that is larger than Theorem~\ref{th:layden_timedep}'s by \(o(\sqrt{|H|})\) for large time steps and \(o(T^{-1})\) for small time steps.
\end{theorem}

\begin{proof}
We initially follow~\cite{wurtz2022counterdiabaticity}.

An adiabatic Hamiltonian can be supplemented, in the short time limit when the Magnus expansion is convergent, with a leading-order counter-diabatic term, producing the following counter-diabatic Hamiltonian:
\begin{equation}
H_{CD} = \gamma(t) H_1 + \beta(t) H_2 + i(s + \dot \gamma(t) \alpha(\gamma)) [H_1, H_2],
\end{equation}
where \(\alpha\) parametrizes an adiabatic gauge potential \(A(\alpha)\)~\cite{wurtz2022counterdiabaticity} with value chosen to minimize
\begin{equation}
  \|\partial_\lambda H_A - i[A(\alpha), H_A]\|_{\mathcal P},
\end{equation}
\(\lambda\) is the time, \(s\) is an added auxiliary counter-diabatic field, and \(\mathcal P\) is a positive semidefinite projector~\cite{Kolodrubetz17}.

We match the leading order Trotter (BCH) error with the counter-diabatic Magnus term in \(\mathcal T e^{-i \int^t_0 H_{CD}(t')\mbox d t'} \) for time step \(j\) corresponding to time \(t'\):
\begin{eqnarray}
&&\gamma_j H_1 \delta t {+} \beta_j H_2 \delta t {-} i \frac{\gamma_j \beta_j \delta t^2}{2} [H_1, H_2]\\
&=& \tau \bar \gamma(t') H_1 {+} \tau \bar \beta(t') H_2 + i \tau (\bar s + \dot \gamma(t') \bar \alpha) [H_1, H_2], \nonumber
\end{eqnarray}
where \(\delta t \equiv T/r\) and
\begin{equation}
    \alpha(t')  = \frac{\|[H, \partial_{t'} H] \|_{\mathcal P}}{\|[[H, \partial_{t'} H], H]\|_{\mathcal P}}.
\end{equation}

We set the respective terms on the left-side equal to the right-side:
\begin{eqnarray}
\label{eq:gammaj}
\gamma_j \delta t &=& \tau \bar \gamma(t'),\\
\label{eq:betaj}
\beta_j \delta t &=& \tau \bar \beta(t'),\\
\label{eq:crossterm}
-\gamma_j \beta_j \delta t^2 /2 &=& \tau (\bar s + \dot \gamma(t') \bar \alpha(t')).
\end{eqnarray}
The overbar indicates a time-averaged variable over the interval \([0, \tau]\).

The unique solution for \(\tau\) from this system of equation is
\begin{equation}
    \tau = -2\frac{\bar s + \dot \gamma(t')\bar \alpha(t')}{\bar \gamma \bar \beta}.
\end{equation}
This solution exists as long as \(\bar s\) is not too large~\cite{wurtz2022counterdiabaticity}.

Using this, for a given \(\delta t\), the auxiliary field \(\bar s(t')\) can be found from Eq.~\ref{eq:crossterm},
\begin{equation}
\label{eq:s}
    \bar s(t') = -\bar \gamma(t') \bar \beta(t') \frac{\delta t^2}{2 \tau} - \dot \gamma \bar \alpha(t'),
\end{equation}
which can be used to find \(\gamma_i\) and \(\beta_i\) through Eqs~\ref{eq:gammaj}-\ref{eq:betaj}:
\begin{equation}
    \gamma_j = \frac{\bar \gamma(t') \delta t}{\tau},
\end{equation}
and
\begin{equation}
    \beta_j = \frac{\bar \beta(t') \delta t}{\tau}.
\end{equation}
So the only free variable we have is \(\delta t\).

Setting \(\delta t = \tau\) makes \(\gamma_j = \gamma(t')\), \(\beta_j = \beta(t')\), and corresponds to regular Trotterization. Therefore, by choosing an optimum \(\delta t\), we are guaranteed to produce a Trotter error, and thereby a bound on fidelity error, at most as large as regular Trotterization.

(This optimal time step \(\delta t \ne \tau\) in general, hence the description of this process as a ``variable-time-step Trotterization''.)

Eq.~\ref{eq:s} shows that setting the time step to a more optimal value will further cancel fidelity decline by either increasing or decreasing the magnitude of the counter-diabatic auxiliary term, which corresponds to decreasing or increasing \(\delta t\) w.r.t.~\(\tau\), respectively, since the term with \(\delta t^2\) is negative while the term with \(\bar \alpha(t')\) is positive (because \(\bar \alpha(t') < 0\)).

For fixed total time, when \(\tau \rightarrow \infty\), the magnitude of the counter-diabatic term should increase to prevent diabatic excitation, and so \(\delta t < \tau\). On the other hand, when \(\tau \rightarrow 0\), the magnitude of the counter-diabatic term should decrease since diabatic excitation decreases, and so \(\delta t > \tau\).

The former regime corresponds to the breakdown of the BCH and Magnus expansion, which occurs when \(\tau \in \Omega (H)\). Decreasing \(\delta t\) to be less than \(\tau\) will therefore stave off when this regime begins.

On the other hand, the latter regime corresponds to the infidelity limit of continuous adiabatic annealing due to the finite total time, which occurs when \(\tau \in \mathcal O((\Delta^2 T)^{-1})\) for gap \(\Delta\) and total time \(T\). Increasing \(\delta t\) to be greater than \(\tau\) will therefore stave off and decrease this regime.

Since \(\bar \gamma(t'), \bar \beta(t'), \dot \gamma(t'), -\bar \alpha(t') > 0\), it follows that there must exist \(\epsilon >0\) such that \(\tau = \delta t - \epsilon\) (\(\tau = \delta t + \epsilon\)) and \(s > 0\) (\(s < 0\)). This will make the counter-diabatic Magnus term, \(i \tau (\bar s {+} \dot \gamma(t') \bar \alpha) [H_1, H_2]\), larger (smaller) by \(\mathcal O(\epsilon)\) since \(\bar s \in \mathcal O (\epsilon)\), while changing the magnitude of the first-order Trotter error by \(\mathcal O(\epsilon^2)\) since it is \(\mathcal O(\gamma \beta)\) and \(\gamma,\,\beta \in \mathcal O(\epsilon)\). This means that the effective evolution will be in the BCH/Magnus convergence regime (or more adiabatic) and will not appreciably effect the Trotter error for small enough \(\epsilon\).

Since the first-order Trotter error is \(\mathcal O(\delta t^2)\), it follows that we want \(\epsilon \in o(\sqrt{\delta t})\). The breakdown of the BCH/Magnus expansion corresponds to \(\epsilon \in o(\sqrt{|H|})\) and the breakdown of the continuous adiabatic regime occurs when \(\epsilon \in o(T^{-1})\).

A more adiabatic evolution by \(o(\epsilon)\) with negligible Trotter error change means that \(|\expval{\psi_f}{(U(t) - U_1)}{\psi_i}|\) will decrease by \(o(\epsilon)\).

This has reduced the first-order Trotter error \emph{for each time step}. We can then consider all the time steps together, apply the same technique as in Theorem~\ref{th:layden_timedep}, to obtain a tighter bound by \(o(\epsilon)\).\qed
\end{proof}

Instances of this scaling for this larger domain of \(\delta t\) can be found in Fig.~\ref{fig:Ising}-\ref{fig:3_App_E} in the subsequent Appendix.

We further note that this counter-diabatic effect arises naturally to improve upon the \(\mathcal O(T^{-2})\) bound of Theorem~\ref{th:adiabatic_trotter_error} as well, in the regime corresponding to intermediate values of $\delta t$. As seen in the intermediate region of Fig.~\ref{fig:scaling}, $\mathcal{I}(T)$ can be significantly smaller than this asymptotic bound suggested by Theorem~\ref{th:adiabatic_trotter_error} because the finite first-order Trotter error cancels or augments diabatic error according to Theorem~\ref{th:bound} that would not otherwise occur in a continuous evolution. This diabatic error's effect is necessarily dependent on the unitary order of the first-order Trotterization (\textit{i.e.,} \(U_1 U_2\), or vice-versa).

\clearpage

Here we present some numerical results demonstrating the increased regime of applicability in \(\delta t\) that Theorem~\ref{th:bound} provides compared to Theorem~\ref{th:layden_timedep}, by allowing for variable Trotter time step.

The increase of the infidelity scaling in Th.~\ref{th:layden_timedep} to the small and large time step regime (corresponding to the green and red regions in Fig.~\ref{fig:scaling}, respectively) is shown here in Fig.~\hyperref[fig:Ising]{\ref*{fig:Ising}(a)} and Fig.~\hyperref[fig:3_App_E]{\ref*{fig:3_App_E}(a)} for an Ising system and a transverse-field Ising model (TFIM) system, respectively. They serve as demonstrations of an application of Theorem~\ref{th:bound}. These systems still have the same \(\delta t\), as in their ``normal Trotter'' implementations, but their linear ramps, \(u(t)\), are no longer constrained to be in \([0,1]\). As a result, the effective time step is no longer proportional to \(\delta t\).

The insets of the figures in (b) show the difference between adjacent time step \(\gamma\)s (a time step-independent quantity for a linear ramp) w.r.t.~\(\delta t\) for ordinary Trotterization compared to under the variable time step Trotterization plotted in the main figures. This quantity is proportional to the effective time step. It shows how the variable effective time step is relatively shorter for long time steps (red region) and relatively longer for short time steps (green region), which is why the scaling is able to extend into these regions (see proof of Th.~\ref{th:bound}).

\begin{figure}[H]
\centering
\includegraphics[scale=0.65]{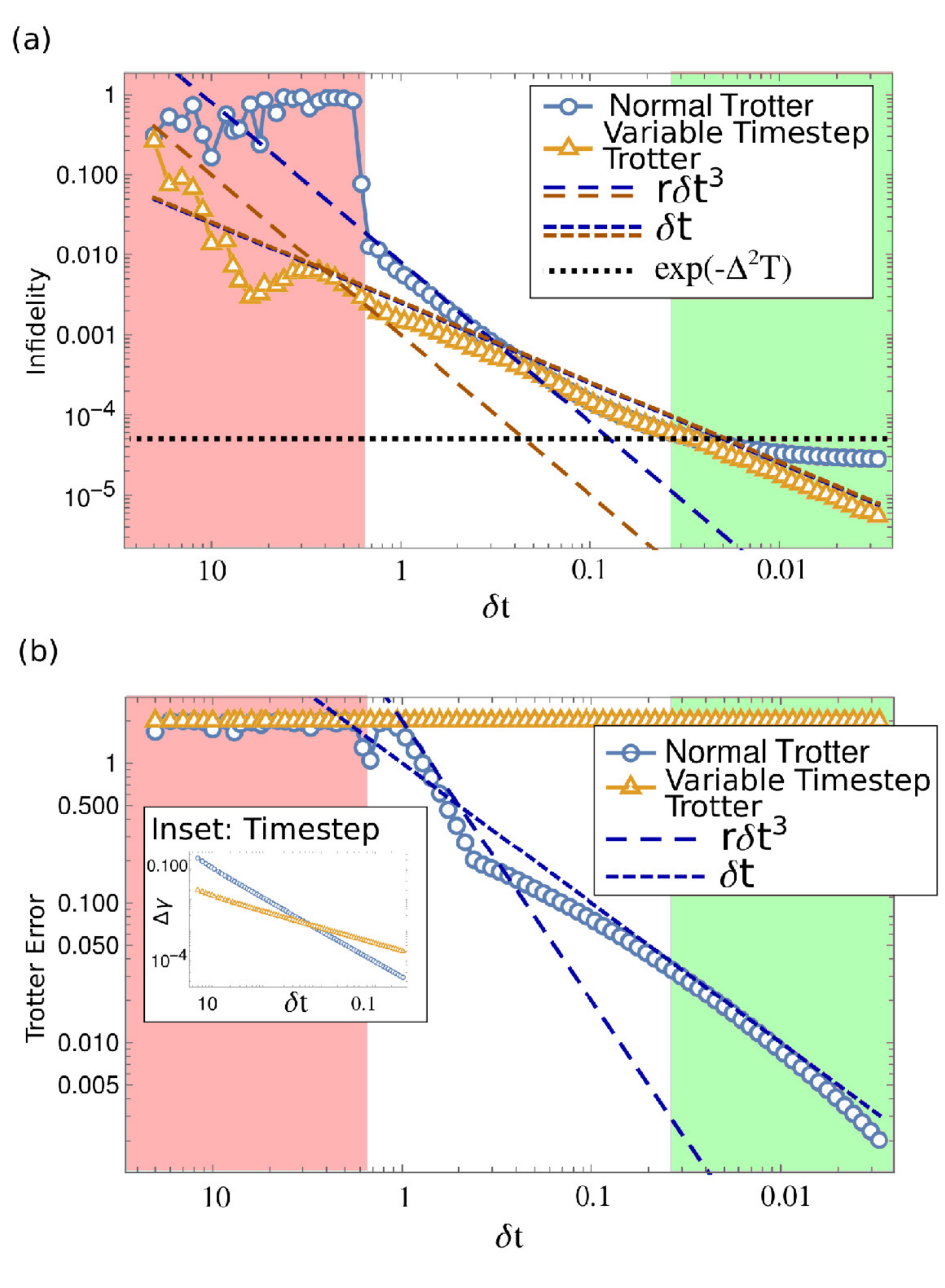}
\caption{Ising \(n=8\) system a) infidelity and b) Trotter error for an initial ground state with total time \(T=100\) evolved to \(0.9\) of the full ramp. Note that the infidelity of the state is w.r.t.~the same initial state evolved under the un-Trotterized continuous unitary, not the ground state of the Hamiltonian at the final \(t=T\).}
\label{fig:Ising}
\end{figure}

Part (b) of these figures plot the corresponding operator norm error and demonstrate how the counter-diabatic effect of variable Trotter time steps specifically targets only one energy transition, and so does not in general improve operator norm error too; counter-diabatic terms for one energy level frequently serve as pro-diabatic terms of others.

\begin{figure}[hbt]
\includegraphics[scale=0.6]{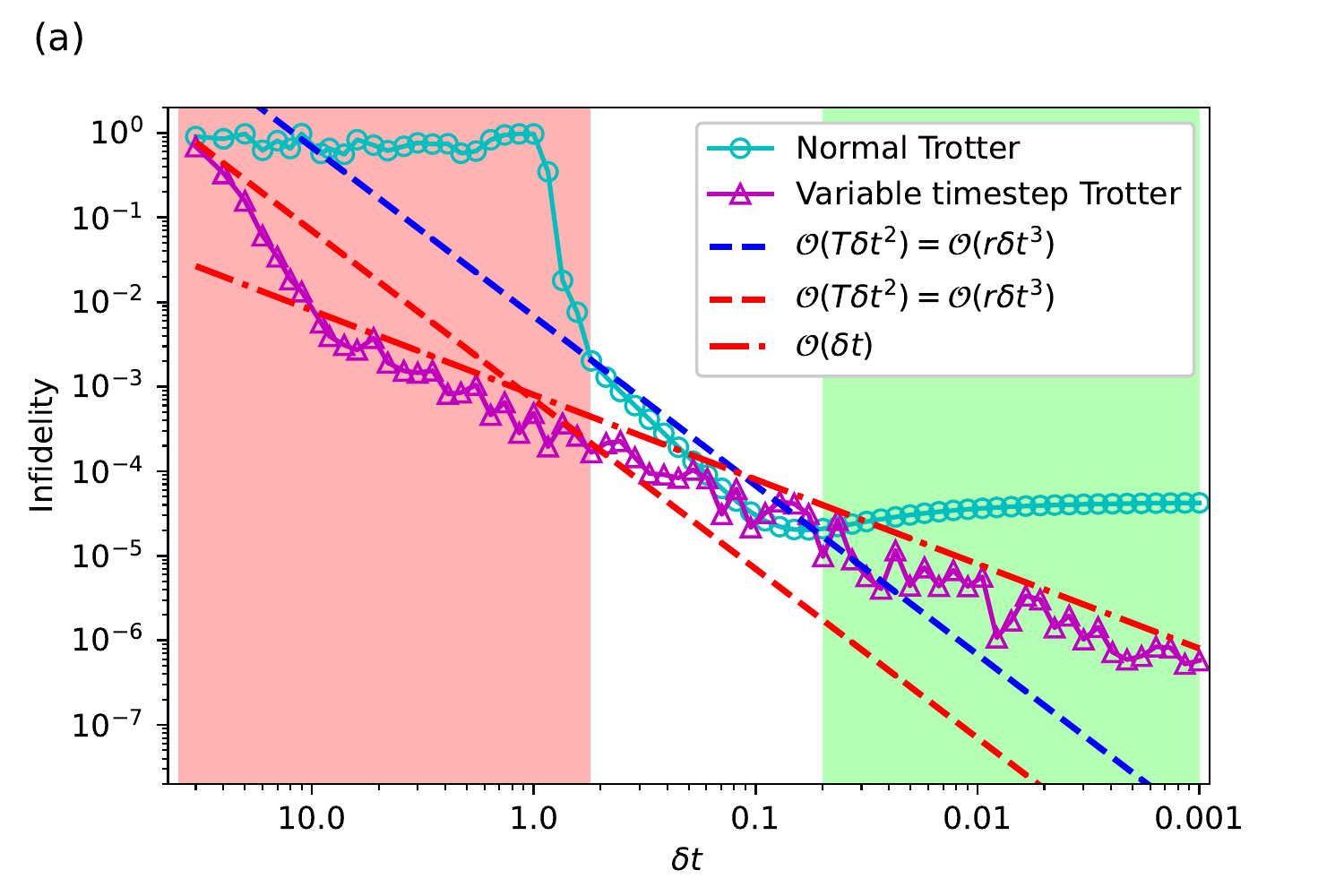}
\includegraphics[scale=0.6]{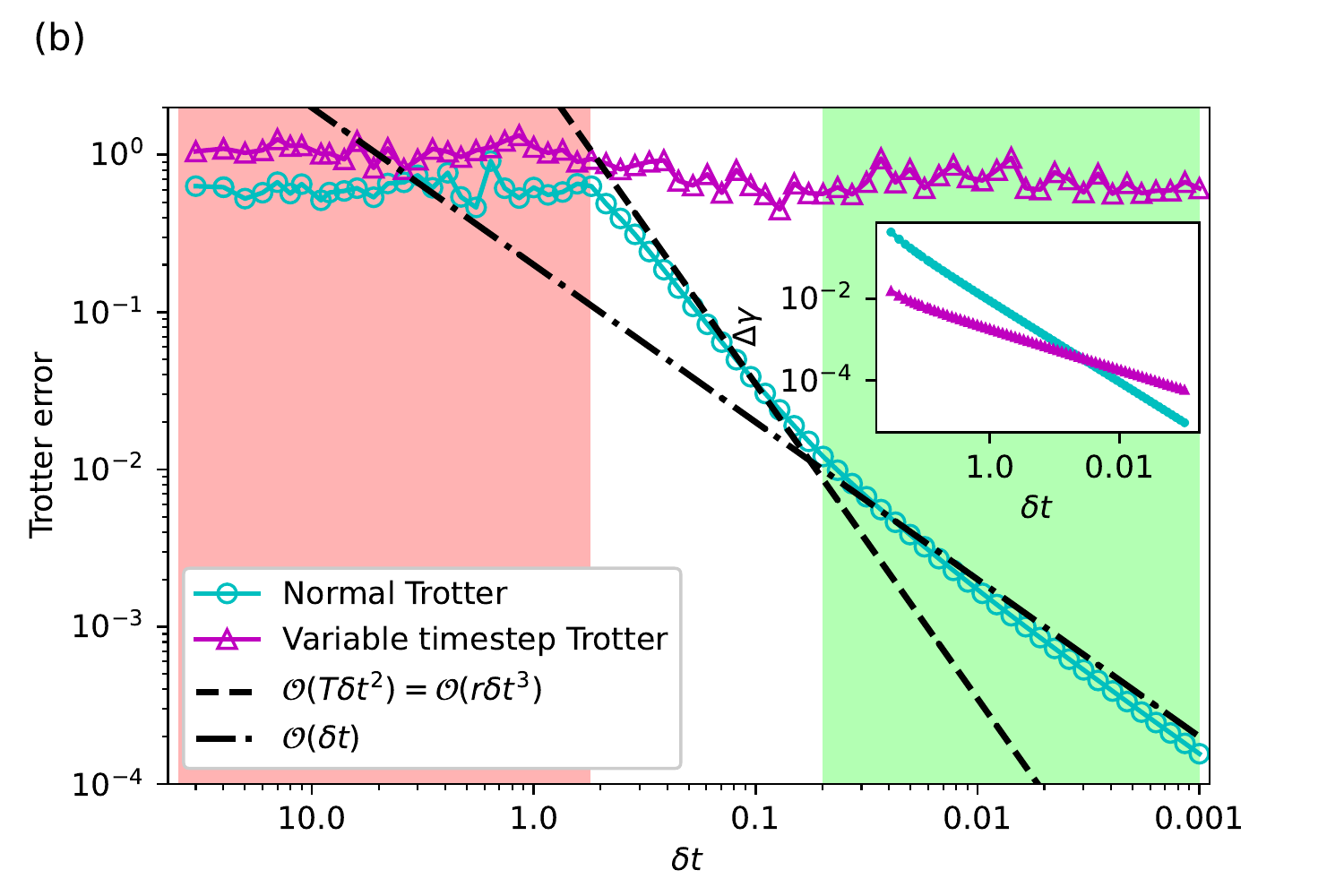}
\caption{(a) State infidelity and (b) Trotter error for the 6-qubit TFIM ground state w.r.t. the Trotterized states obtained with normal Trotter and variable time step Trotter. Both TFIM ground state and Trotterized states were obtained with a total time $T=100$ and evolved to $0.9$ of the full ramp. The TFIM system Hamiltonian is $H_2=\sum_i(Z_iZ_{i+1}+0.8Z_i+0.9X_i)$ and the mixer Hamiltonian is $H_1=\sum_i X_i$.}
\label{fig:3_App_E}
\end{figure}

\clearpage

\section{Correspondence between QAOA and Trotterized annealing}
\label{app:QAOA}
\renewcommand{\theequation}{D\arabic{equation}}
\renewcommand{\thefigure}{D\arabic{figure}}
\setcounter{figure}{0}

Here we provide details on the methods and parameters used to produce the QAOA curves, their corresponding anneal curves, and their final Trotterizations of the \(3\)-regular MAXCUT problem shown in Figure.~\ref{fig:QAOA_anneal}. We also show these intermediate curves and compare them (see Figs.~\ref{fig:optimal_anneal_curves_3reg}-\ref{fig:anneal_trotter_diff_3reg}). We also show similar results for a \(2\)-regular MAXCUT instance below in Fig.~\ref{fig:QAOA_anneal_2reg}.

The intention of illustrating these results is to show how at large enough depth, QAOA can become a constant time step digitization of an underlying set of quantum annealing curves with similar total integrated time. This is a limit that differs from the vanishing time step digitization of the adiabatic limit that has traditionally been considered~\cite{Farhi14}.

\subsection{QAOA Curve Generation}

To begin, \(P=2\) and \(P=3\) optimal QAOA angles are found by optimizing starting with a seed that is an increasing linear ramp with a particular slope and intersection. Specifically, for \(2\)-regular MAXCUT, the initial seed   is \(\gamma_1=0\) and \(\gamma_2\approx 1\). Similarly, for \(3\)-regular MAXCUT, the initial seed is \(\gamma_1=\pi/16\) and   \(\gamma_2=3\pi/16\). Optimal sets of angles for \(P+1\) are then found by ``bootstrapping,'' \textit{i.e.,} seeding the optimization with interpolated optimal angles found for \(P\) \cite{Pichler18, Mbeng19, Mbeng19_2}. This generates a set of optimal angles that approach a universal curve, as shown in Fig.~\hyperref[fig:QAOA_anneal]{\ref*{fig:QAOA_anneal}(a)}.

To generate the initial family of \(2\)-regular QAOA curves, a fixed pairwise total time \(\gamma_i+\beta_i=1.1\) was imposed. This was made to address the fact that the \(2\)-regular ring of disagrees has a continuum of reachable optimal angles for \(P \ge N/2\), which was the \(P\)-regime we explored with \(N=8\) and \(P\in \{8,\,\ldots,\l,32\}\). This constraint was not necessary for the \(3\)-regular example shown in Figure~\ref{fig:QAOA_anneal}.

\subsection{Optimal Control Anneal Curves}

In this section, we provide additional detail on the procedure used to generate the optimal anneal curves plotted in the main text in Fig.~\hyperref[fig:QAOA_anneal]{\ref*{fig:QAOA_anneal}(b)}. Each anneal curve $u(t)$, $t\in[0,T]$, is sought by solving the quantum optimal control problem
\begin{equation}
\min_{u(t)} J[T,u(t)].
\label{Eq:QOC}
\end{equation}
with the objective functional
\begin{equation}
  J[T,u(t)] = \expval{\Psi(T)}{H_2}{\Psi(T)},
  \label{Eq:controlobjfun}
\end{equation}
where $|\Psi(T)\rangle$ denotes the solution to the Schr\"odinger equation at time $T$ with Hamiltonian
\begin{equation}
H[u(t)] = H_1+(H_2-H_1)u(t).
\end{equation}
It can be shown that optimal control solutions of Eq. (\ref{Eq:QOC}) must satisfy the corresponding Euler-Lagrange equations:
 \begin{align}
 	|\dot{\Psi}{(t)}\rangle &= -i H[u(t)] ~|\Psi{(t)}\rangle, \mbox { with } |\Psi{(0)}\rangle = |\Psi_0\rangle, \label{Eq:TDSEinQOC}
 \end{align}
 \begin{align}
 	|\dot{\lambda}{(t)}\rangle &= -i H[u(t)] ~|\lambda{(t)}\rangle, \mbox { with }|\lambda{(T)}\rangle = H_2|\Psi(T)\rangle,
 \end{align}
 \begin{align}
\frac{\delta J}{\delta u(t)} = 2\big(\mbox{Im}\left[\langle\lambda(t)|(H_2-H_1)|\Psi(t)\rangle\right]\big)=0.
 \end{align}
 where $|\lambda(t)\rangle$ is a Lagrange multiplier introduced to enforce that the dynamics of $|\Psi(t)\rangle$ obey the  Schr\"odinger Eq. (\ref{Eq:TDSEinQOC}) \cite{stengel_optimal_1994,PhysRevA.37.4950}.

 In order to find optimal solutions $u(t)$, a variety of methods can be employed \cite{khaneja2005optimal,reich2012monotonically,ho2010accelerated}. In this work, we consider local optimization methods, which are fed an initial seed curve, and proceed to seek a solution $u(t)$ by iteratively updating $u(t)$ in the direction of the negative gradient until $\frac{\delta J}{\delta u(t)}=0$, indicating a local minimum is reached.

Here, we carry out this procedure for several different seed curves. In particular, we utilize the optimal QAOA angles obtained for a range of different $p$ values in order to generate corresponding seed curves for $u(t)$ associated with each value of $p$. These seed curves are produced by performing a cubic spline interpolation between the QAOA angles (shown in Fig.~\ref{fig:optimal_anneal_curves_3reg}) to generate curves that are discretized into a series of small time steps \(\Delta t \ll 1\) with total integrated time equal to the sum of the QAOA angles, \textit{i.e.,} $T = \sum_i(\gamma_i + \beta_i)$. The seed curves are then further modified by explicitly including an initial ``bang'' of the problem Hamiltonian \(H_2\) and a final ``bang'' of the mixer Hamiltonian \(H_1\) of duration a single time step \(\Delta t\), such that $u(0) = 1$ and $u(T)=0$. These initial and terminal values of $u(t)$ are fixed during the subsequent optimization of the anneal curves in order to preferentially find the closest optimal anneal curve with such initial and final bangs, as per~\cite{Brady21_2}, whose findings indicate that the optimal anneal curves should obey this bang-anneal-bang structure.

 After these seed curves have been generated, they are each subsequently optimized in order to obtain optimal annealing curves minimizing Eq. (\ref{Eq:controlobjfun}). At each iteration of the optimization algorithm, the state $|\Psi(t)\rangle$ is evolved via forward simulation, by approximating the full quantum time evolution operator $U(T)$ as a time-ordered product of propagators over $N$ small time steps $\Delta t$, \textit{i.e.,} as
\begin{equation}
  U(T) \approx U_N U_{N-1} \cdots U_1
\end{equation}
where $H(t)$ is approximated to be time-independent over each small time step such that
\begin{equation}
  U_k = e^{-i(H_1 + (H_2-H_1)u_k)\Delta t},
\end{equation}
where in the above, we have introduced the abbreviated notation \(u_k \equiv u(t_k)\).
Following the forward simulation of $|\Psi(t)\rangle$, the dynamics of $|\lambda(t)\rangle$ are computed in the same manner but via backwards propagation from the terminal condition $|\lambda(T)\rangle=H_2|\Psi(T)\rangle$. The gradients $\frac{\delta J}{\delta u_k}$ are then evaluated according to
\begin{equation}
  \frac{\delta J}{\delta u_k} =  2\Delta t \big( \mbox{Im}\left[\langle\lambda_k|(H_2-H_1)|\Psi_k\rangle\right]\big)
\end{equation}
for each value of $k$, where \(\ket{\Psi_k} = U_k \cdots U_1 \ket{\Psi_0}\) and $\ket{\lambda_k} = U_{k-1}^\dagger \cdots U_N^\dagger \ket{\lambda(T)}$, where \(\ket{\Psi_0}\) is taken to be the ground state of \(H_1\). These gradients are then passed to the gradient algorithm L-BFGS-B, which we utilize to perform the optimization with \(u(t)\) constrained to lie between \(0\) and \(1\). The solutions obtained using L-BFGS-B correspond to the results plotted in Fig.~\hyperref[fig:QAOA_anneal]{\ref*{fig:QAOA_anneal}(b)}.

A \(3\)-regular MAXCUT example is shown in Fig.~\ref{fig:QAOA_anneal} in the main text. More detailed data that underlie this figure are shown in Figures~\ref{fig:optimal_anneal_curves_3reg}-\ref{fig:anneal_trotter_diff_3reg}.

\begin{figure}[h]
    \includegraphics[scale=0.5]{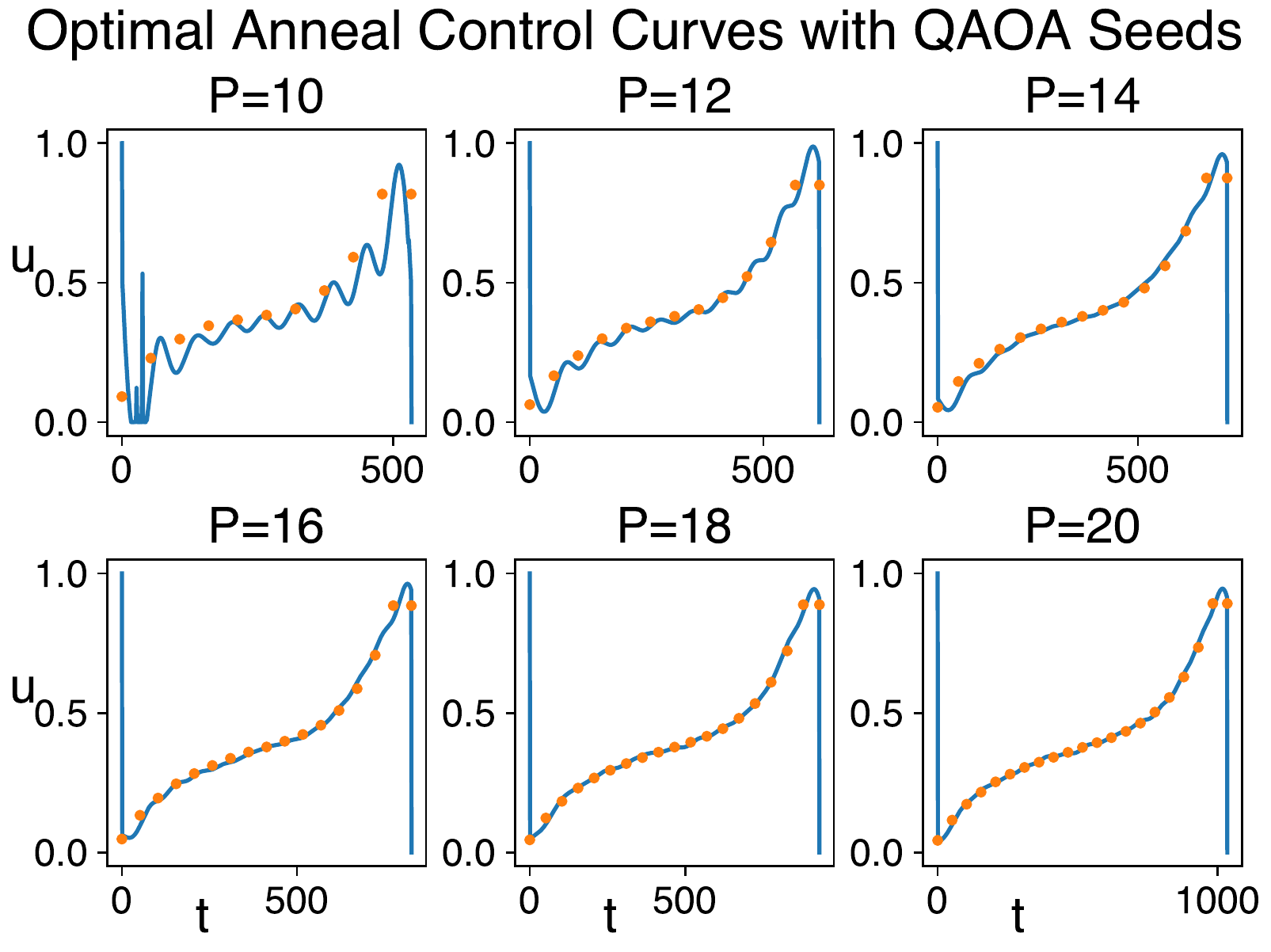}
\caption{Optimal anneal control curves with QAOA seeds (last seed repeated) for Fig.~\ref{fig:QAOA_anneal}.}
  \label{fig:optimal_anneal_curves_3reg}
\end{figure}

\begin{figure}[h]
    \includegraphics[scale=0.5]{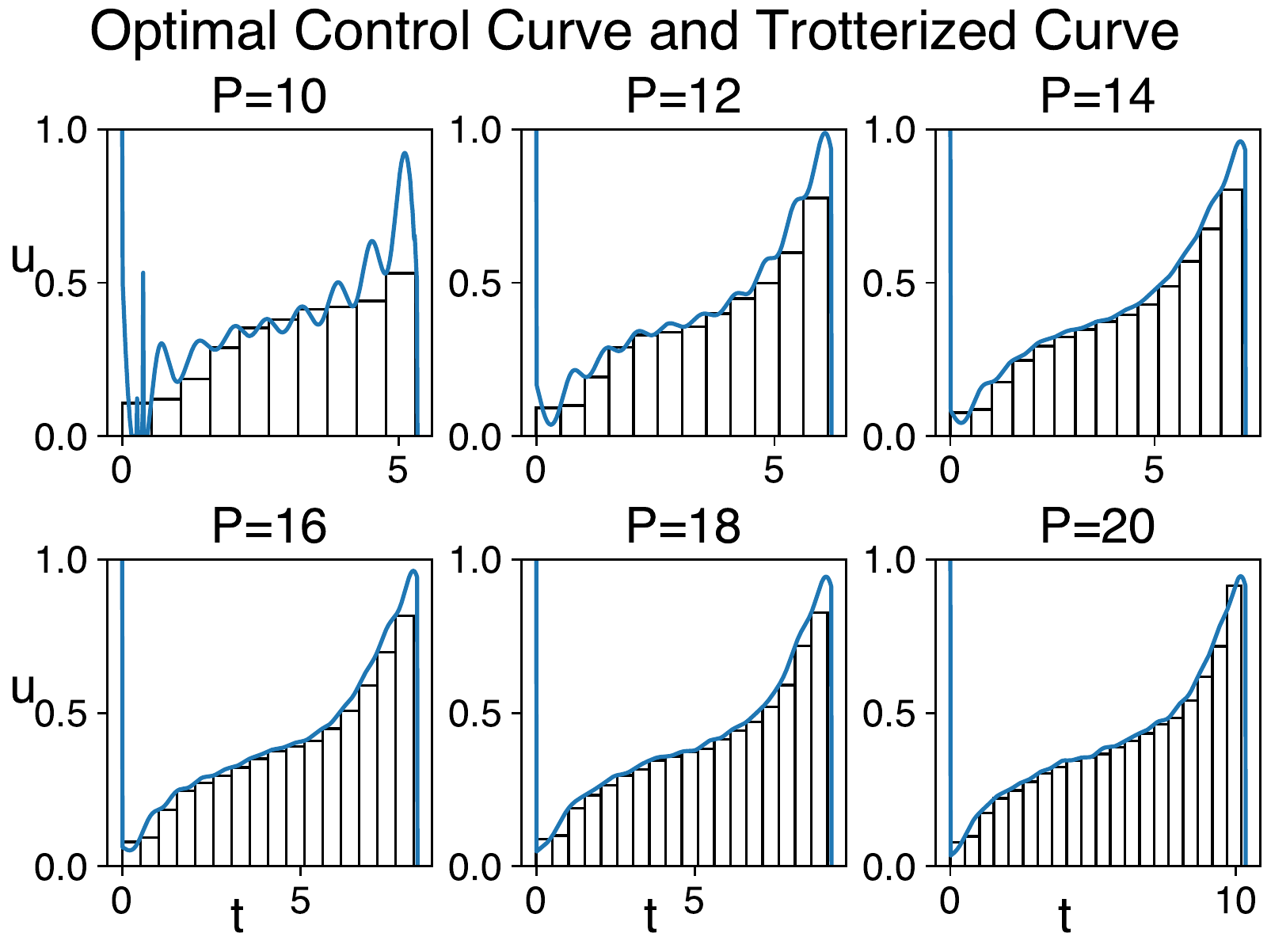}
\caption{Comparison between optimal anneal control curves and their Trotterization for Fig.~\ref{fig:QAOA_anneal}.}
  \label{fig:anneal_trotter_comp_3reg}
\end{figure}

\begin{figure}[h]
    \includegraphics[scale=0.5]{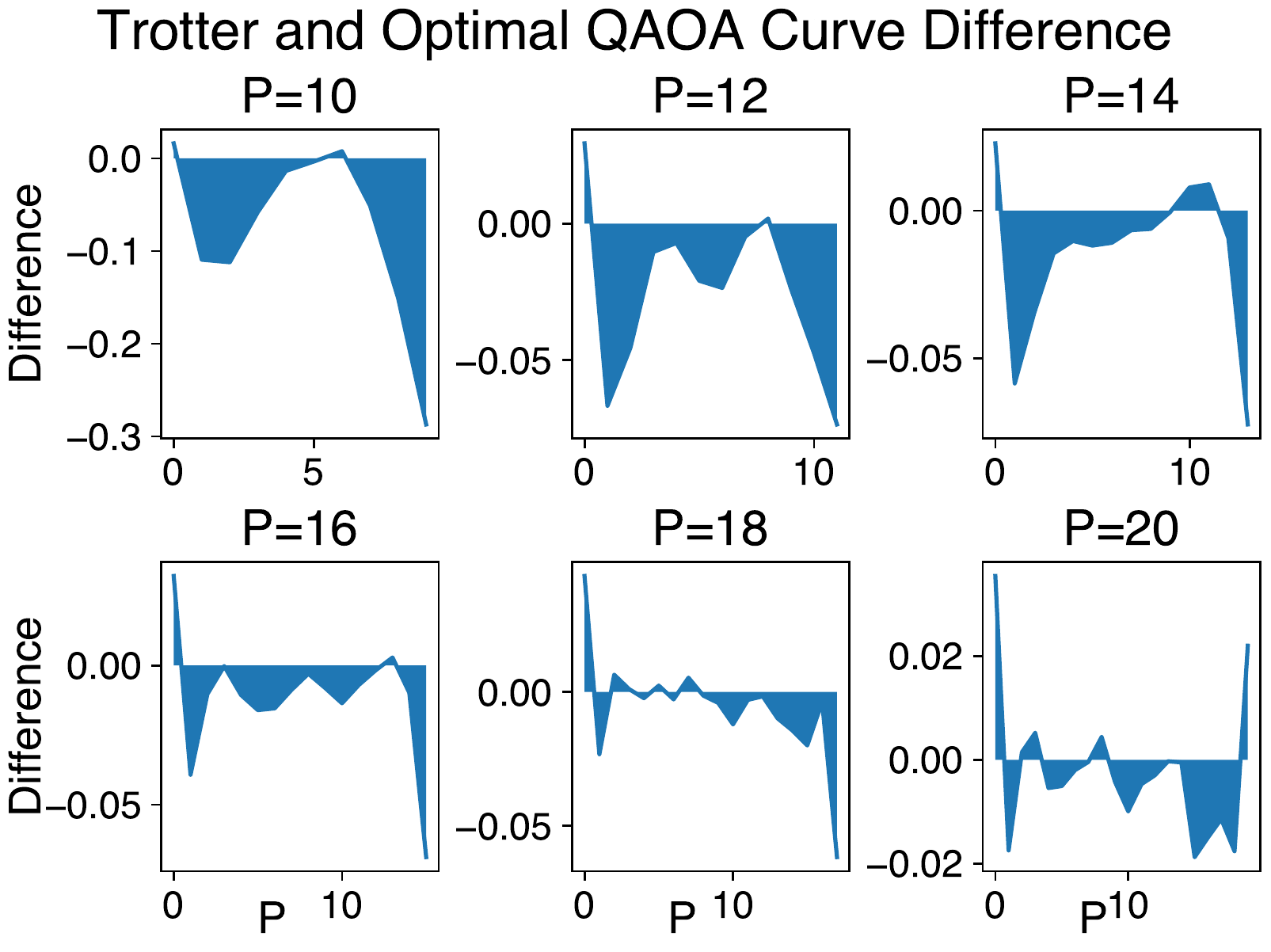}
\caption{Difference between optimal anneal control curves and their Trotterization for Fig.~\ref{fig:QAOA_anneal}.}
  \label{fig:anneal_trotter_diff_3reg}
\end{figure}

We also present a \(2\)-regular MAXCUT example below in Fig.~\ref{fig:QAOA_anneal_2reg}.

\begin{figure}[h]
\includegraphics[scale=0.6]{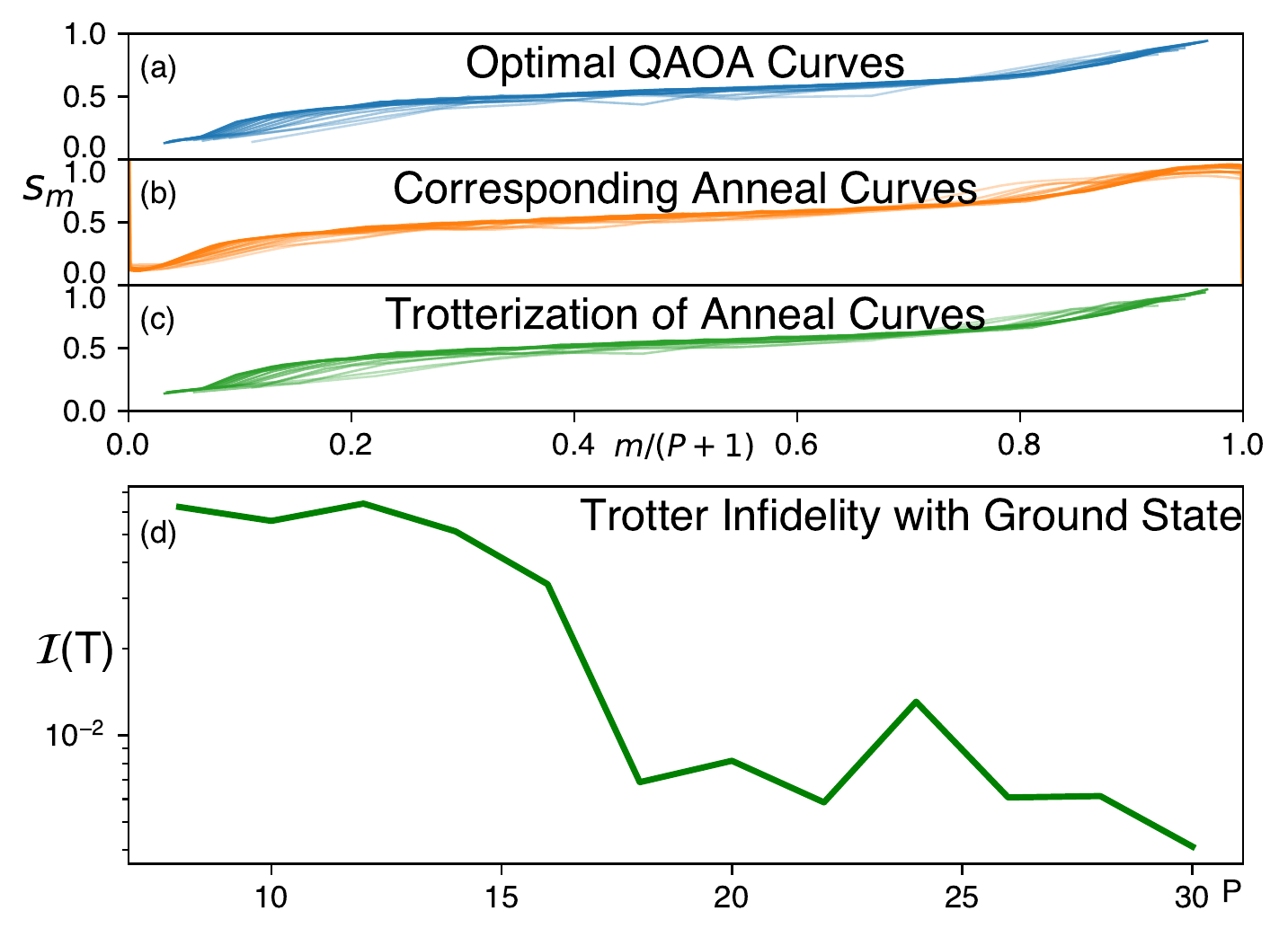}
\caption{A problem Hamiltonian of \(N{=}8\) \(2\)-regular periodic Ising is examined (MAXCUT). Shown are an optimal set of (a) QAOA curves found for \(P=\{8,\,\dots,\, 30\}\) by bootstrapped seeding~\cite{Mbeng19,Mbeng19_2} from lower \(P\) (darker curves correspond to larger \(P\)), (b) the optimal anneal curves found by seeding from these corresponding QAOA curves, (c) the Trotterization of these anneal curves, and (d) the infidelity with the problem Hamiltonian ground state after evolution under the QAOA, anneal, and Trotter curves. See this Appendix for more details about the numerics.}
\label{fig:QAOA_anneal_2reg}
\end{figure}

More detailed data that underlie Figure~\ref{fig:QAOA_anneal_2reg} are shown in Figures~\ref{fig:optimal_anneal_curves_2reg}-\ref{fig:anneal_trotter_diff_2reg}.

\begin{figure}[h]
    \includegraphics[scale=0.5]{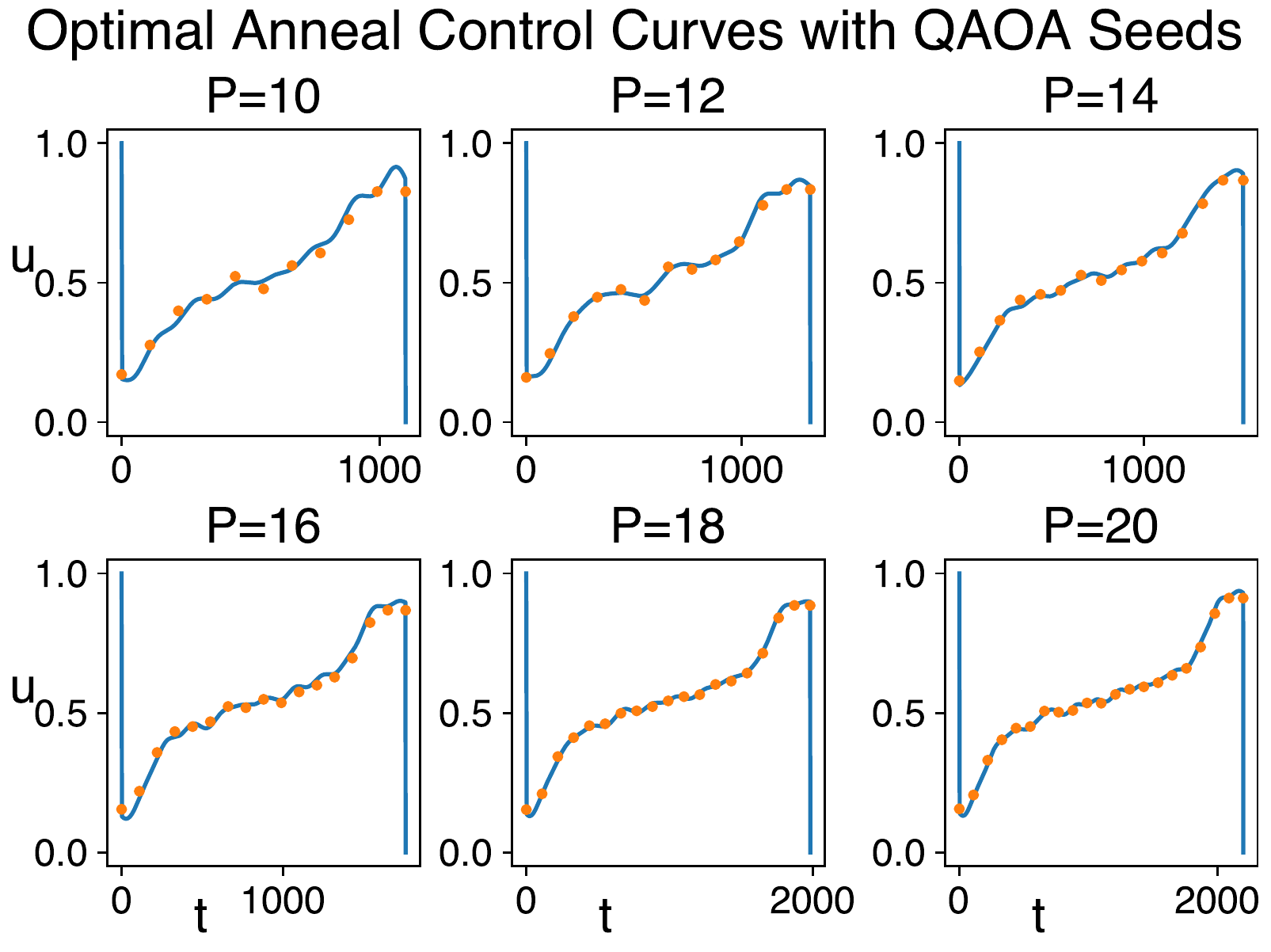}
\caption{Optimal anneal control curves with QAOA seeds (last seed repeated) for Fig.~\ref{fig:QAOA_anneal_2reg}.}
  \label{fig:optimal_anneal_curves_2reg}
\end{figure}

\begin{figure}[ht]
    \includegraphics[scale=0.5]{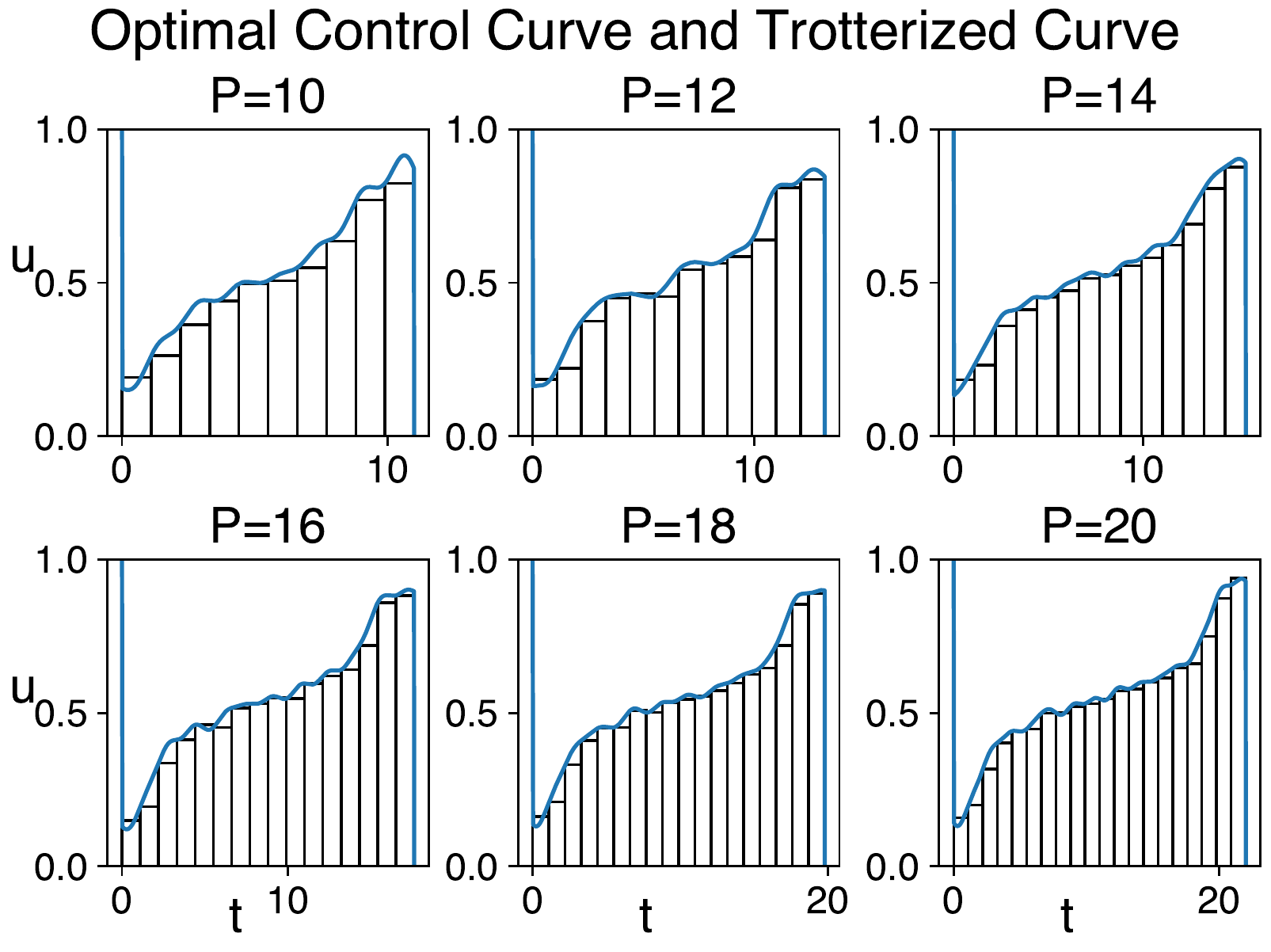}
\caption{Comparison between optimal anneal control curves and their Trotterization for Fig.~\ref{fig:QAOA_anneal_2reg}.}
  \label{fig:anneal_trotter_comp_2reg}
\end{figure}

\begin{figure}[ht]
    \includegraphics[scale=0.5]{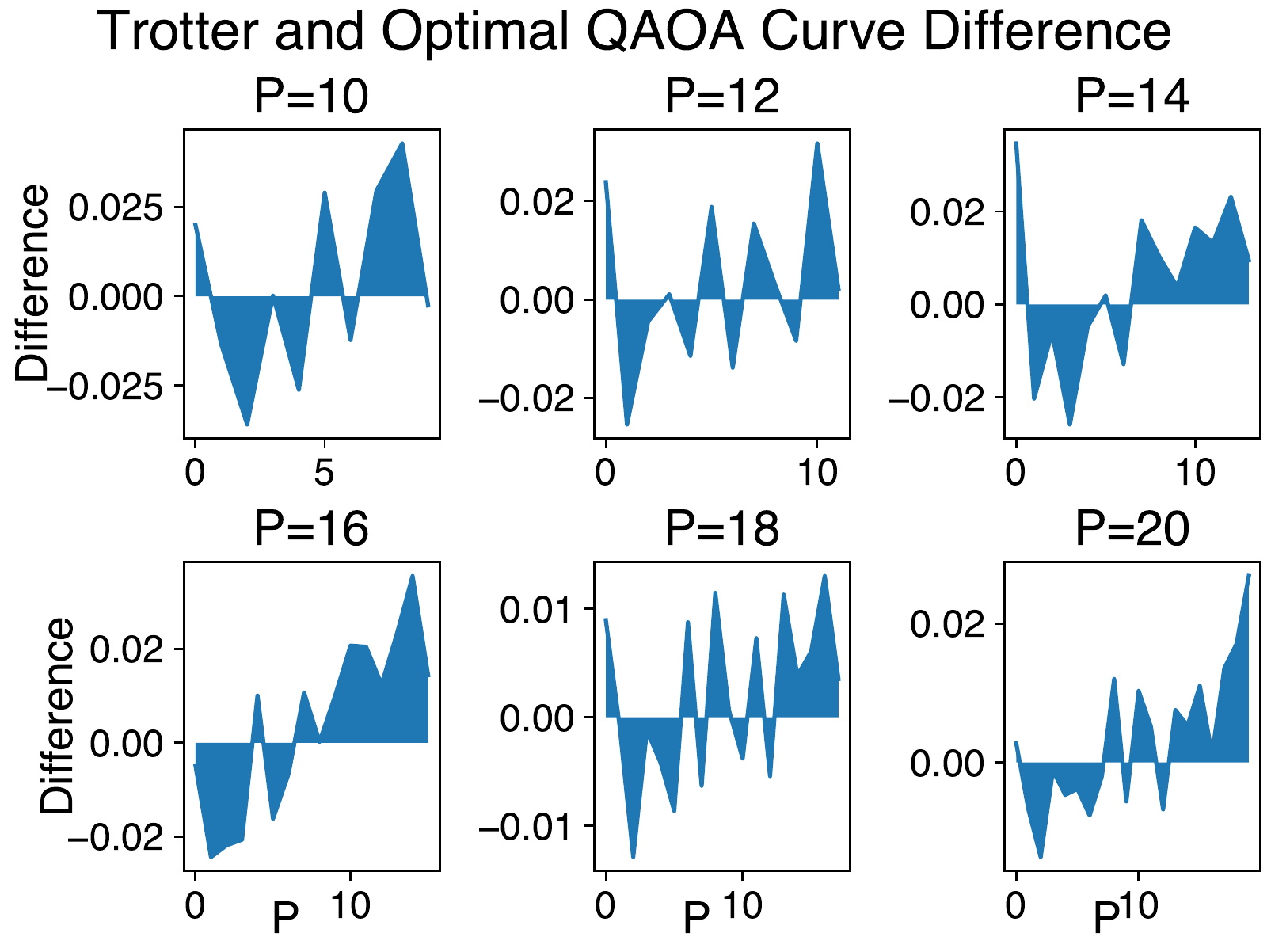}
\caption{Difference between optimal anneal control curves and their Trotterization for Fig.~\ref{fig:QAOA_anneal_2reg}.}
  \label{fig:anneal_trotter_diff_2reg}
\end{figure}
\clearpage
\subsection{Trotterization of Anneal Curves}

The optimal anneal curves are subsequently Trotterized with a time step equal to the average integrated time of the corresponding QAOA bangs in two manners. First, the initial and final time steps, where a bang is enforced, obtain values equal to the integrated area under the curve. Second, the remaining middle time steps obtain values equal to the continuous curve at the beginning of their respective ``bins''. This Trotterization is visualized in Fig.~\ref{fig:anneal_trotter_comp_3reg} and Fig.~\ref{fig:anneal_trotter_comp_2reg}. The corresponding Trotterized curves (see Fig.~\hyperref[fig:QAOA_anneal]{\ref*{fig:QAOA_anneal}(c)} and Fig.~\hyperref[fig:QAOA_anneal_2reg]{\ref*{fig:QAOA_anneal_2reg}(c)}) produce final ground state infidelity that is decreasing with \(T\) (or, equivalently, with \(P\)) as shown in Fig.~\hyperref[fig:QAOA_anneal]{\ref*{fig:QAOA_anneal}(d)} and Fig.~\hyperref[fig:QAOA_anneal_2reg]{\ref*{fig:QAOA_anneal_2reg}(d)}. This is contrary to traditional Trotter error scaling which indicates infidelity should increase with \(T\) as \(\mathcal O (T \delta t^2)\) for such a fixed time step Trotterization~\cite{Layden21}.

\section{Numerical examples that corroborate the predicted scaling}
\label{app:numerics}
\renewcommand{\theequation}{E\arabic{equation}}
\renewcommand{\thefigure}{E\arabic{figure}}
\setcounter{figure}{0}

The figures in the main text present numerical results for a simple two-level system.
Here we present similar results extended to the Ising model and transverse field Ising model, as well as a molecular Hamiltonian.
Our intent here is to demonstrate the general applicability of Theorems~\ref{th:layden_timedep}-\ref{th:bound} and the results presented in the main text.

Fig.~\ref{fig:1_App_E} is similar to Fig.~\ref{fig:scaling} and Fig.~\ref{fig:partialramps}, except that it plots infidelity for a TFIM system $(n=6)$ at various fractions of the full ramp and both operator ordering in the Trotterization. It demonstrates how for an evolution over \(90\%\) of the full ramp, the infidelity for the TFIM system scales according to Th.~\ref{th:layden_timedep}, and then transitions to scaling according to Th.~\ref{th:adiabatic_trotter_error} when the evolution control curve traverses the full ramp from \(H_1\) to \(H_2\).

\begin{figure}[hbt]
\centering
\includegraphics[scale=0.6]{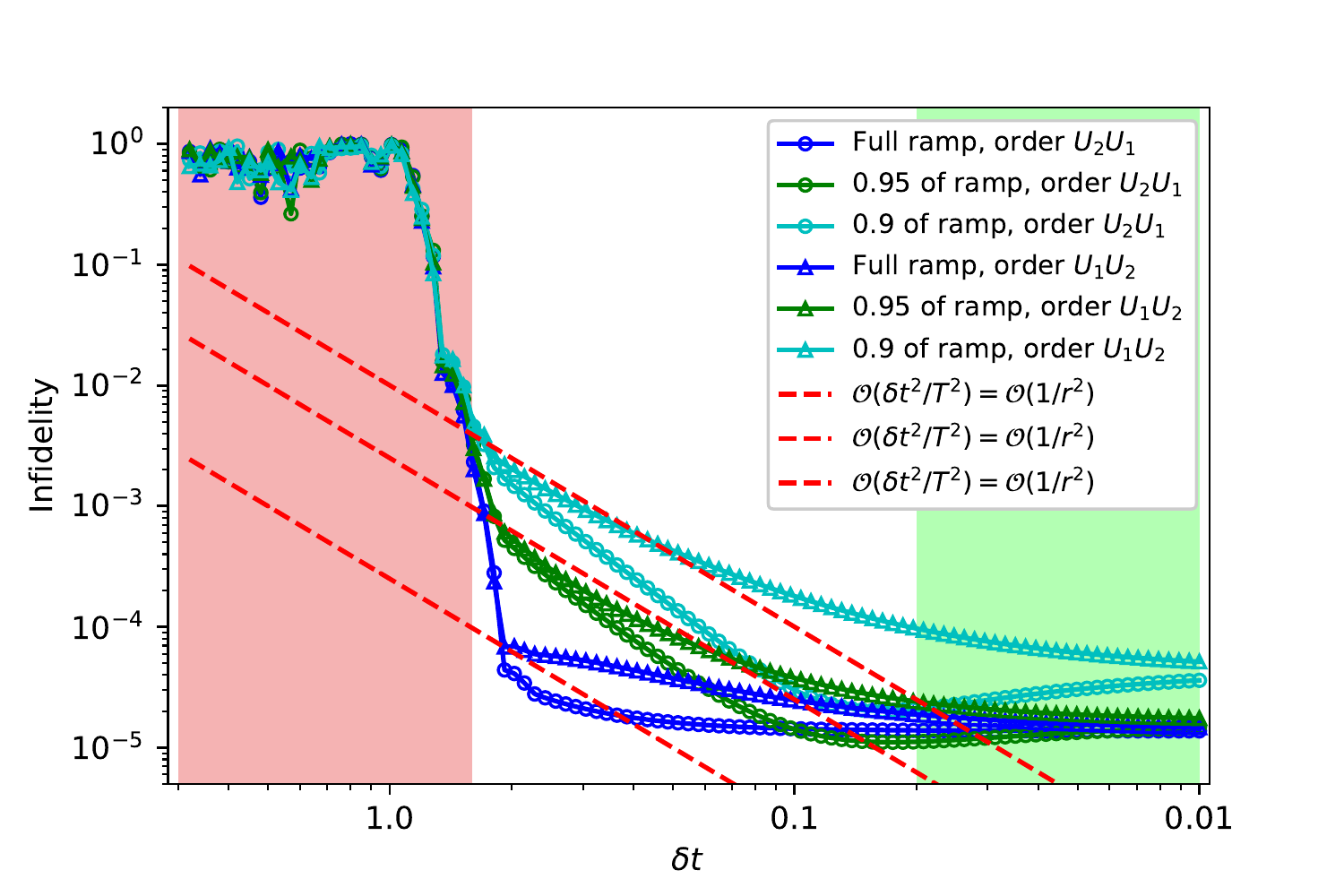}
\caption{State infidelity between the Trotterized ground state and the target ground state of a TFIM ($n=6$), with a total time $T=100$ over different fractions of the full ramp. The target state is defined to be the ground state of the TFIM Hamiltonian evaluated at the fraction of the ramp evolved to. The TFIM system Hamiltonian is $H_2=\sum_i(Z_iZ_{i+1}+0.8Z_i+0.9X_i)$ and the mixer Hamiltonian is $H_1=\sum_i X_i$.}
\label{fig:1_App_E}
\end{figure}

Fig.~\ref{fig:2_App_E} is similar to Fig.~\hyperref[fig:scaling]{\ref*{fig:scaling}(b)}, except that it plots infidelity for the same TFIM system of Fig.~\ref{fig:1_App_E}. It exhibits the same \(\mathcal O(T^{-2})\) scaling expected from Th.~\ref{th:adiabatic_trotter_error}.

\begin{figure}[hbt]
\centering
\includegraphics[scale=0.6]{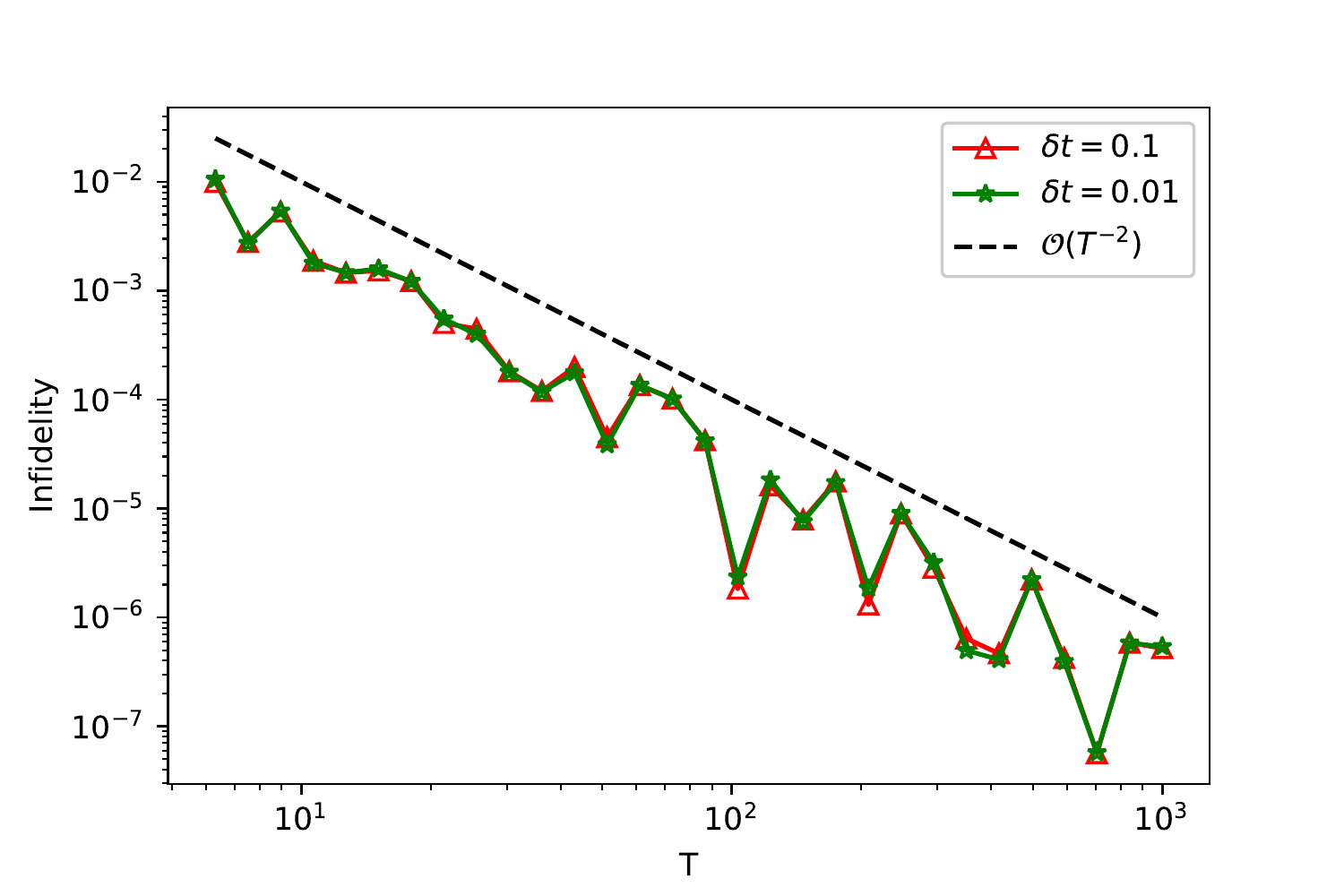}
\caption{State infidelity between the Trotterized ground state and the TFIM ground state with different time steps over the full ramp for different total times $T$. The 6-qubit TFIM system Hamiltonian is $H_2=\sum_i(Z_iZ_{i+1}+0.8Z_i+0.9X_i)$ and the mixer Hamiltonian is $H_1=\sum_i X_i$.}
\label{fig:2_App_E}
\end{figure}

Fig.~\hyperref[fig:5_App_E]{\ref*{fig:5_App_E}(a)} is similar to Fig.~\ref{fig:scaling}, except that it plots infidelity for the ground state of a final Hamiltonian corresponding to an H$_2$ molecule's bond length, using a minimal basis in four different fermionic encodings: Jordan-Wigner (J-W), Bravyi-Kitaev (B-K), Checksum (C), and Up-Down (UD).
\begin{equation}\label{eq:J-K_Hamiltonian}
\begin{aligned}
H_{\text{J-W}}&=-0.106 I+0.045 (X_0 Y_1 Y_2 X_3- X_0 X_1 Y_2 Y_3+\\
&Y_0 X_1 X_2 Y_3-Y_0 Y_1 X_2 X_3) +0.17 (Z_0+Z_1) +\\
&0.168 Z_0 Z_1 + 0.12 (Z_0 Z_2+Z_1 Z_3) +\\
& 0.166 (Z_0 Z_3 + Z_1 Z_2)  -0.22 (Z_2+Z_3) + 0.174 Z_2 Z_3
\end{aligned}
\end{equation}
\begin{equation}\label{eq:B-K_Hamiltonian}
\begin{aligned}
H_{\text{B-K}}&= -0.106 I +0.045( X_0 Z_1 X_2 + X_0 Z_1 X_2 Z_3 +\\
& Y_0 Z_1 Y_2 + Y_0 Z_1 Y_2 Z_3) -0.22 (Z_1 Z_2 Z_3+Z_2)+\\
&0.17 (Z_0 + Z_0 Z_1) + 0.166 (Z_0 Z_1 Z_2 + Z_0 Z_1 Z_2 Z_3) +\\
&0.12 (Z_0 Z_2 + Z_0 Z_2 Z_3) + 0.168 Z_1  + 0.174 Z_1 Z_3,
\end{aligned}
\end{equation}
\begin{equation}\label{eq:C_Hamiltonian}
\begin{aligned}
H_{\text{C}}&=-0.106 I + 0.091 (X_0 Y_1 Y_2 - Y_0 Y_1 X_2) +\\
&0.17 (Z_0 + Z_1)+ 0.342 Z_0 Z_1 + 0.24 Z_0 Z_2 + 0.331 Z_1 Z_2 -\\
&0.22 (Z_0 Z_1 Z_2 + Z_2),
\end{aligned}
\end{equation}
\begin{equation}\label{eq:U-D_Hamiltonian}
\begin{aligned}
H_{\text{U-D}}&=-0.347I + 0.182 X_0 X_1 + 0.011 Z_0 Z_1 +\\
&0.39 (Z_0 + Z_1).
\end{aligned}
\end{equation}
The encodings were done using OpenFermion~\cite{mcclean_openfermion_2020}. The Jordan-Wigner encoding, Eq.~\ref{eq:J-K_Hamiltonian}, is on 4 qubits and it involves a relatively large number of 4-local terms. The Bravyi-Kitaev encoding, Eq.~\ref{eq:B-K_Hamiltonian}, is also on 4 qubits and it involves fewer 4-local terms. For the last two encodings we use checksum codes~\cite{steudtner_fermion--qubit_2018,mcclean_openfermion_2020} to reduce the number of qubits. We use an even-weight checksum code to reduce one qubit in  Eq.~\ref{eq:C_Hamiltonian}. Finally, we use an odd-weight checksum code on both spin-up and -down modes to reduce two qubits in Eq.~\ref{eq:U-D_Hamiltonian}.

Similarly, Fig.~\hyperref[fig:5_App_E]{\ref*{fig:5_App_E}(b)} plots infidelity for the ground states of the Hamiltonian for the water molecule (H$_2$O) in its equilibrium geometry and the Hamiltonian for lithium hydride (LiH) with a diatomic bond length of 1.45 $\text{\normalfont\AA}$. Both Hamiltonians were obtained with the Bravyi-Kitaev encoding using OpenFermion. For H$_2$O we use the STO-3G basis and, since for the purposes of this work we do not need the exact value of the ground state energy $E_0$, we approximate $E_0$ by reducing the active space to 4 orbitals and 8 qubits (it usually requires 14 qubits~\cite{tranter_comparison_2018}). Similarly, for LiH we use the STO-6G basis and reduce the active space from 6 orbitals and 12 qubits to 3 orbitals and 6 qubits~\cite{maupin_variational_2021}.

In Fig.~\ref{fig:Molecular_H_Infid_vs_T}, which is similar to Figs.~\ref{fig:2_App_E} and~\hyperref[fig:scaling]{\ref*{fig:scaling}(b)}, we plot infidelity for the same molecular systems of Fig.~\ref{fig:5_App_E}, but using only the Bravyi-Kitaev encoding for the three of them. For sufficiently large  total time $T$, they exhibit a $\mathcal{O}(T^{-2})$ scaling, as expected from Th.~\ref{th:adiabatic_trotter_error}.
\begin{figure}[H]
\includegraphics[scale=0.6]{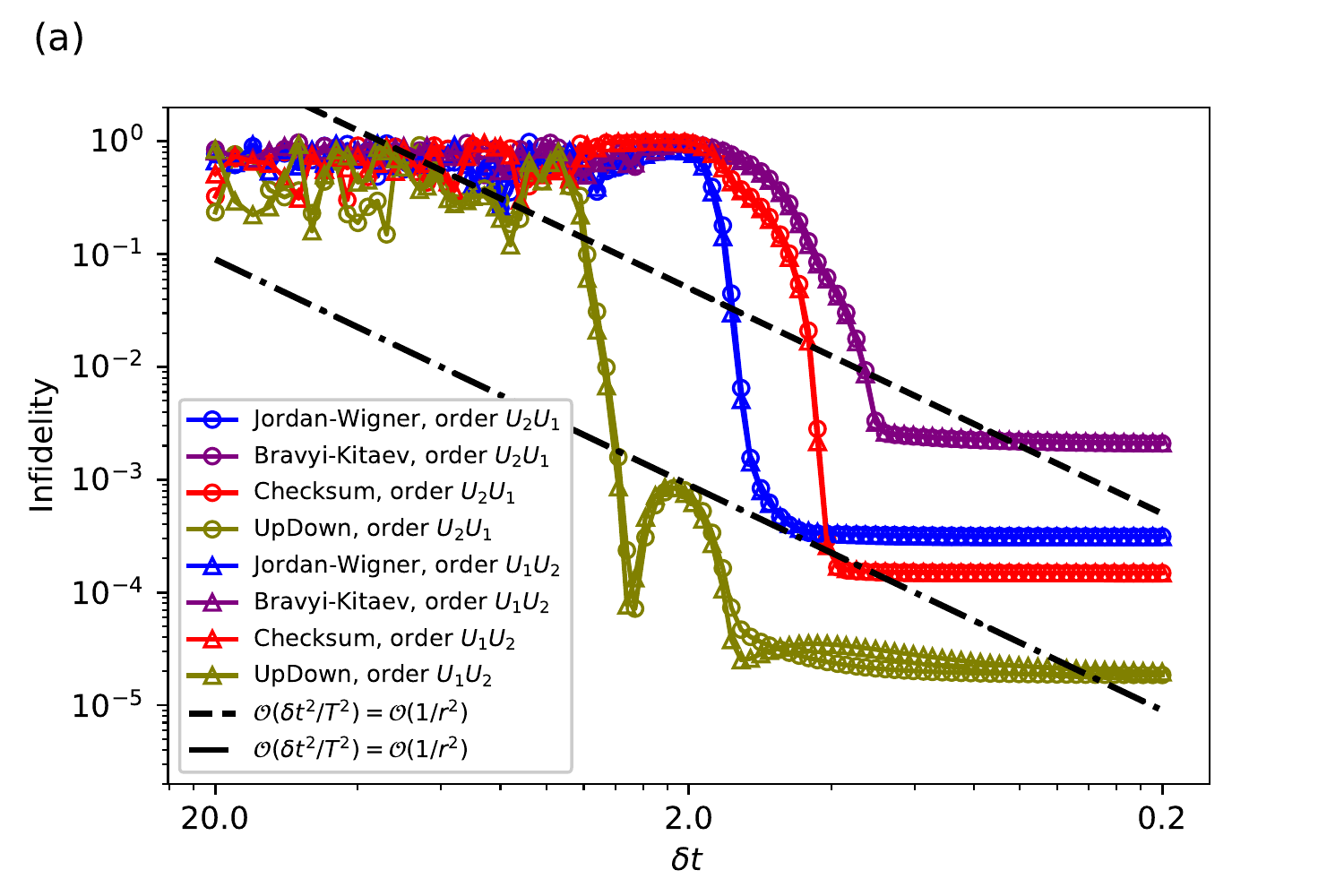}
\includegraphics[scale=0.6]{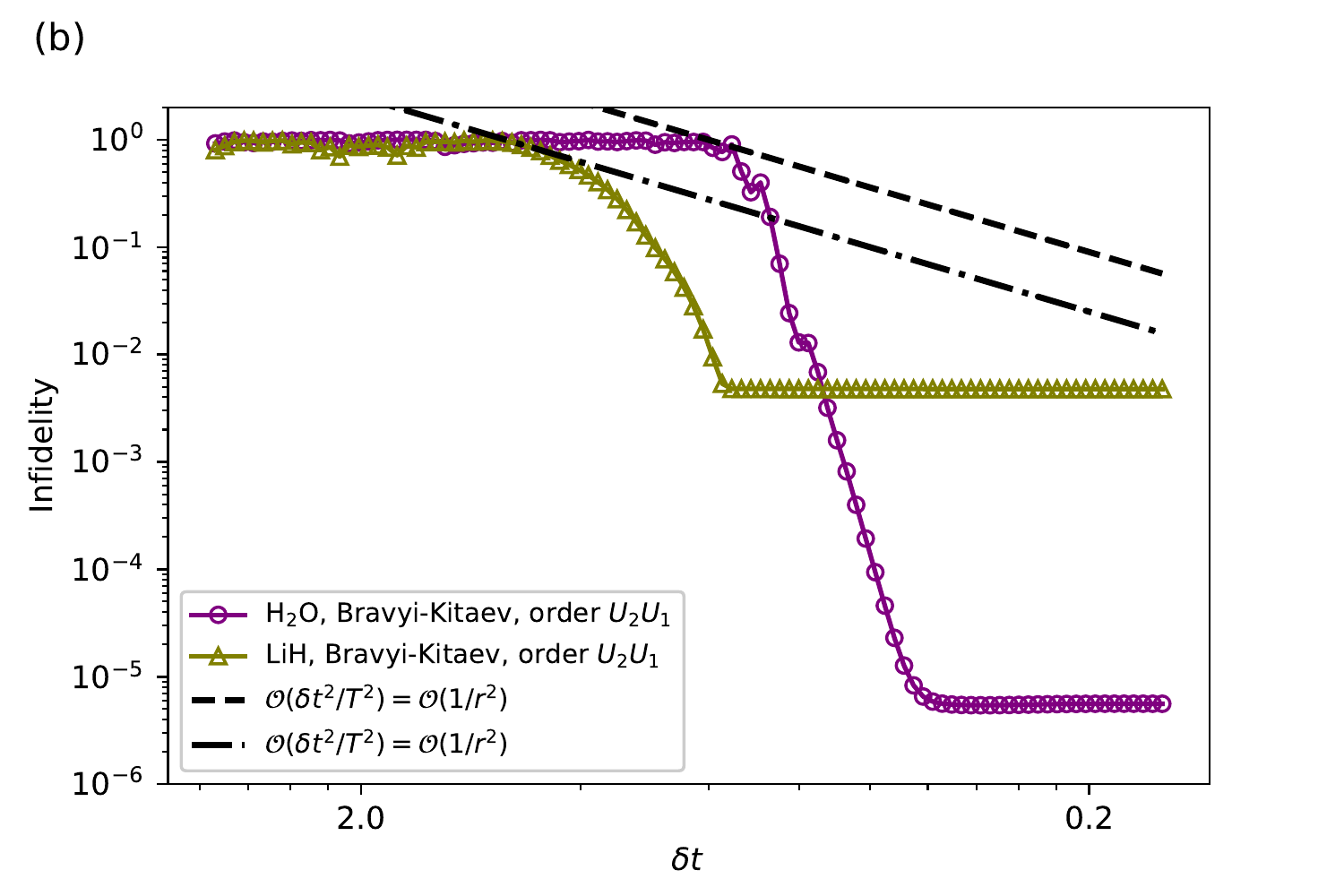}
\caption{ State infidelity between the Trotterized ground state and the target ground state of (a) the Hamiltonian corresponding to an H$_2$ molecule's bond length (different encodings), (b) the Hamiltonian corresponding to the water molecule (H$_2$O) in its equilibrium geometry (purple curve, Bravyi-Kitaev encoding) and the Hamiltonian for lithium hydride (LiH) with diatomic bond length equal to 1.45 $\text{\normalfont\AA}$ (dark gold color, Bravyi-Kitaev encoding). In both (a) and (b), the evolution is over the full ramp and total times (a) $T=200$ and (b) $T=2000$. The same mixer Hamiltonian, $H_1=\sum_i X_i$, is used in both (a) and (b).}
\label{fig:5_App_E}
\end{figure}

\begin{figure}[H]
\includegraphics[scale=0.6]{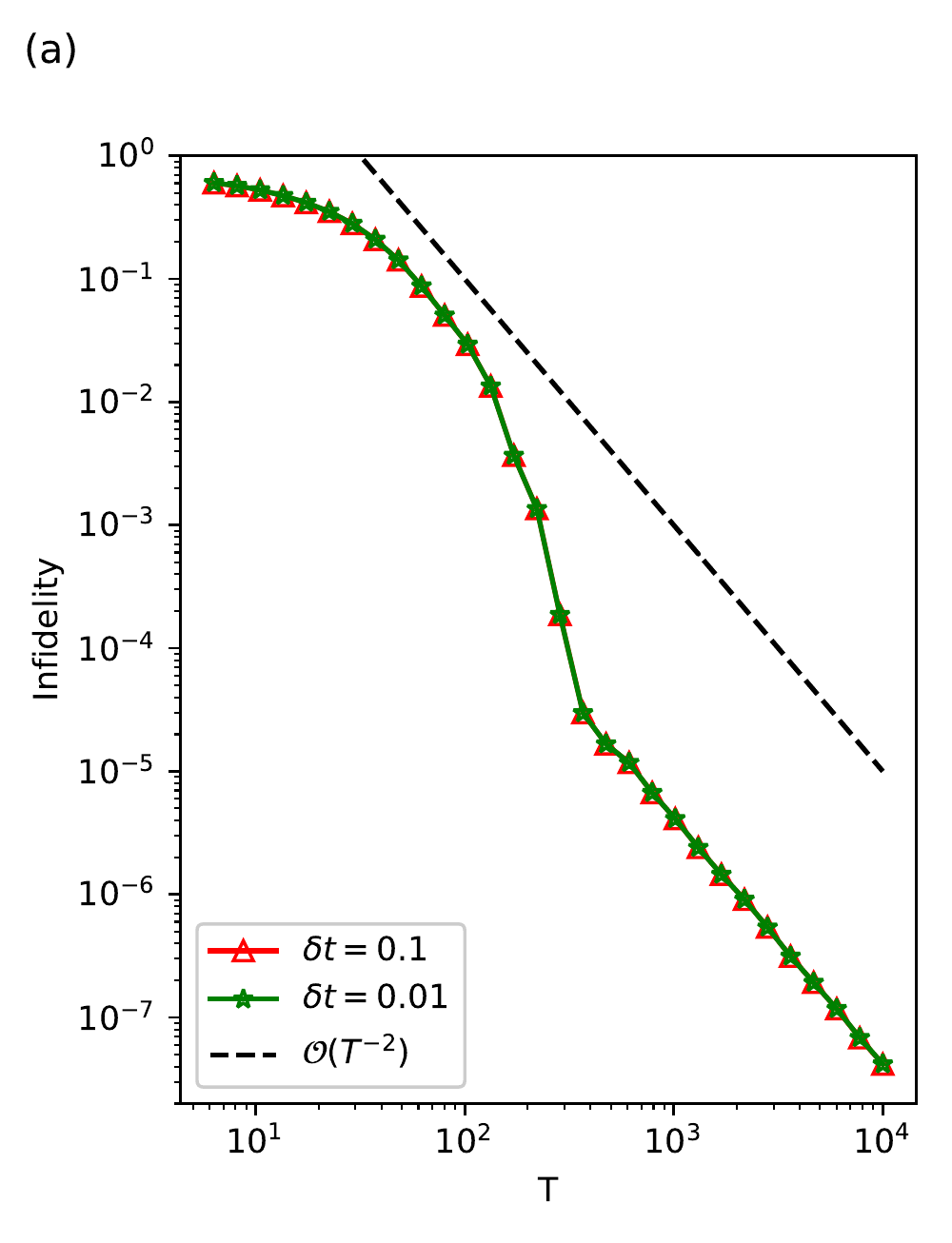}
\includegraphics[scale=0.6]{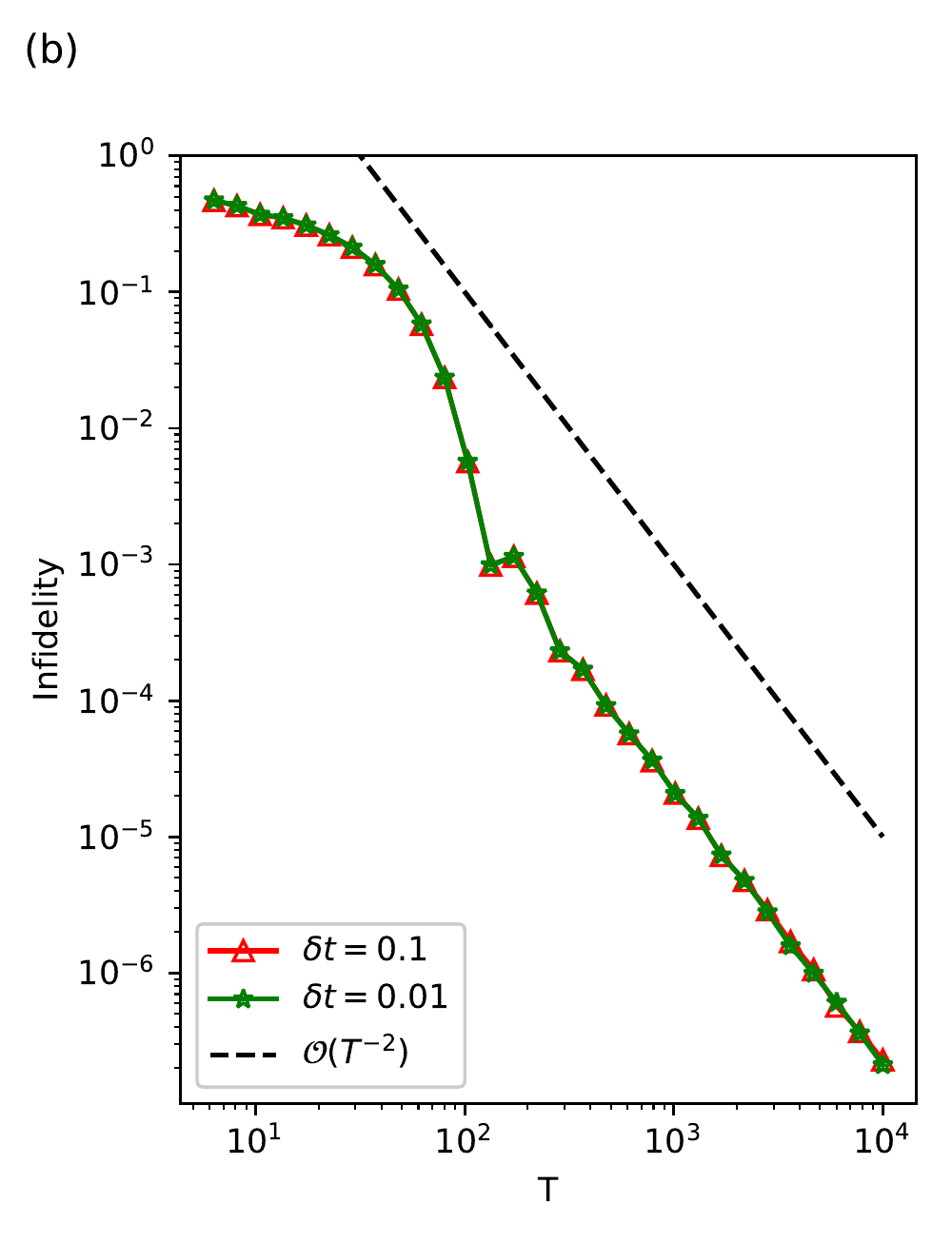}
\includegraphics[scale=0.6]{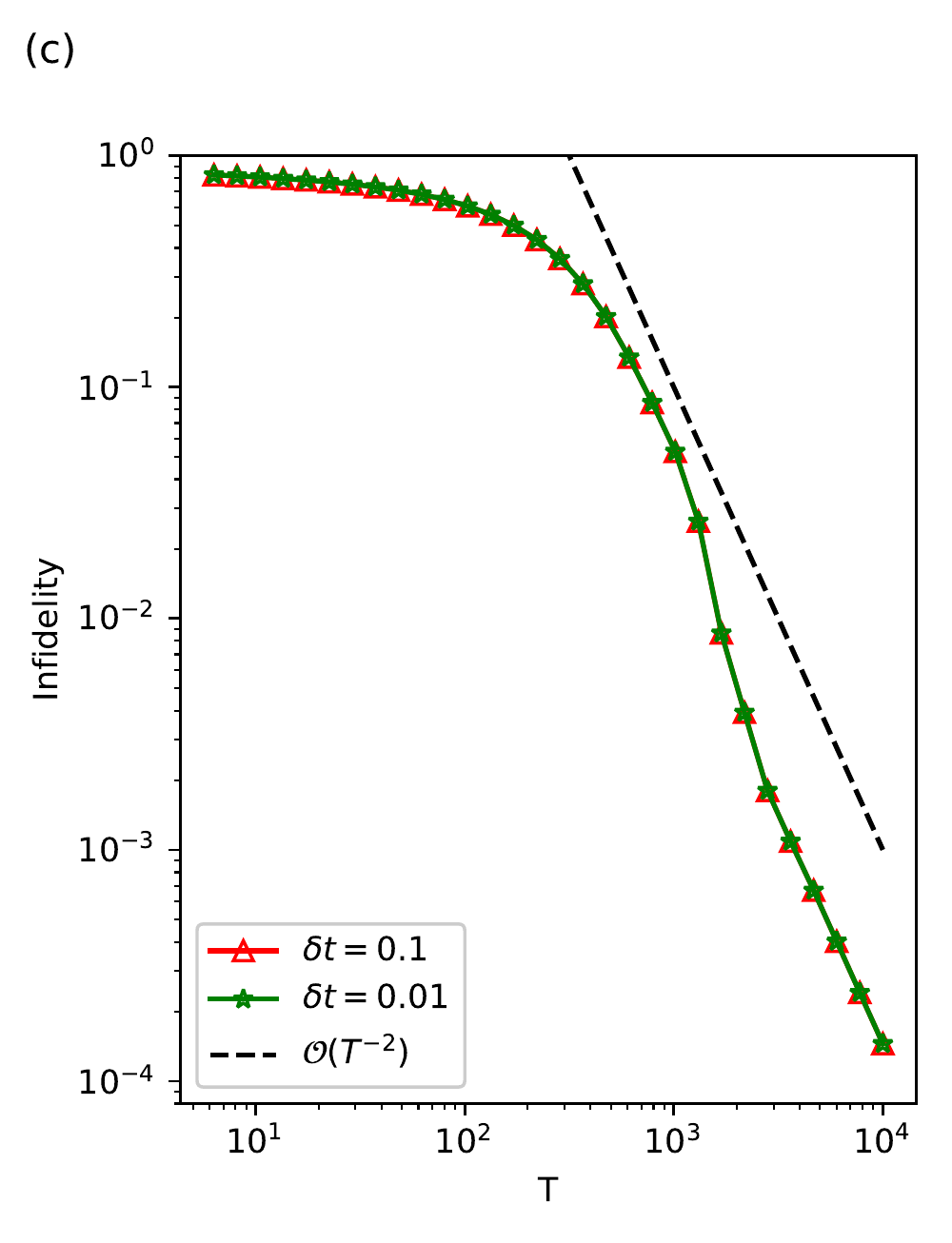}
\caption{ With different time steps $\delta t$ over the full ramp for different total times $T$, we  plot the state infidelity between the Trotterized ground state and the target ground state of (a) the Hamiltonian corresponding to an H$_2$ molecule's bond length, (b) the Hamiltonian corresponding to the water molecule (H$_2$O) in its equilibrium geometry, and (c) the Hamiltonian for lithium hydride (LiH) with diatomic bond length equal to 1.45 $\text{\normalfont\AA}$. For these three cases, we use the Bravyi-Kitaev encoding and the same mixer Hamiltonian $H_1=\sum_i X_i$.}
\label{fig:Molecular_H_Infid_vs_T}
\end{figure}

The \emph{self-healing} or self-correcting characteristic property of the infidelity scaling in Th.~\ref{th:adiabatic_trotter_error} can be observed in the TFIM system as well. Fig.~\ref{fig:TFIM_fidelity} and Fig.~\ref{fig:4_App_E} are similar to Fig.~\ref{fig:instantaneous_eigenstate_fidelity}, except for a TFIM system with $16$ and $64$ energy levels, respectively. In the smaller system, it is easiest to see how the error builds up and then self-cancellation occurs after the corresponding excited state energy levels diverge with respect to the ground energy level.

\begin{figure}[hbt]
\includegraphics[scale=0.3]{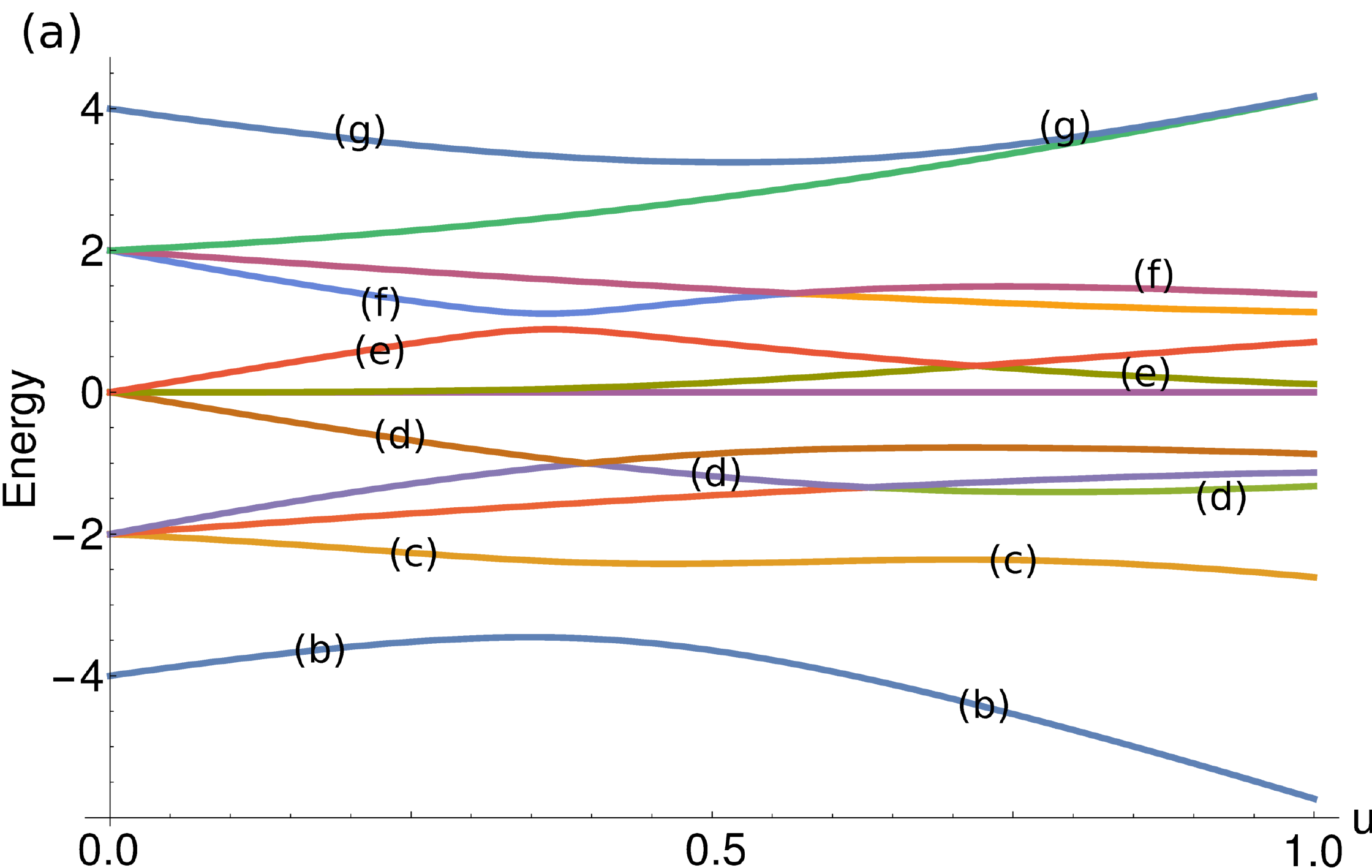}
\includegraphics[scale=0.19]{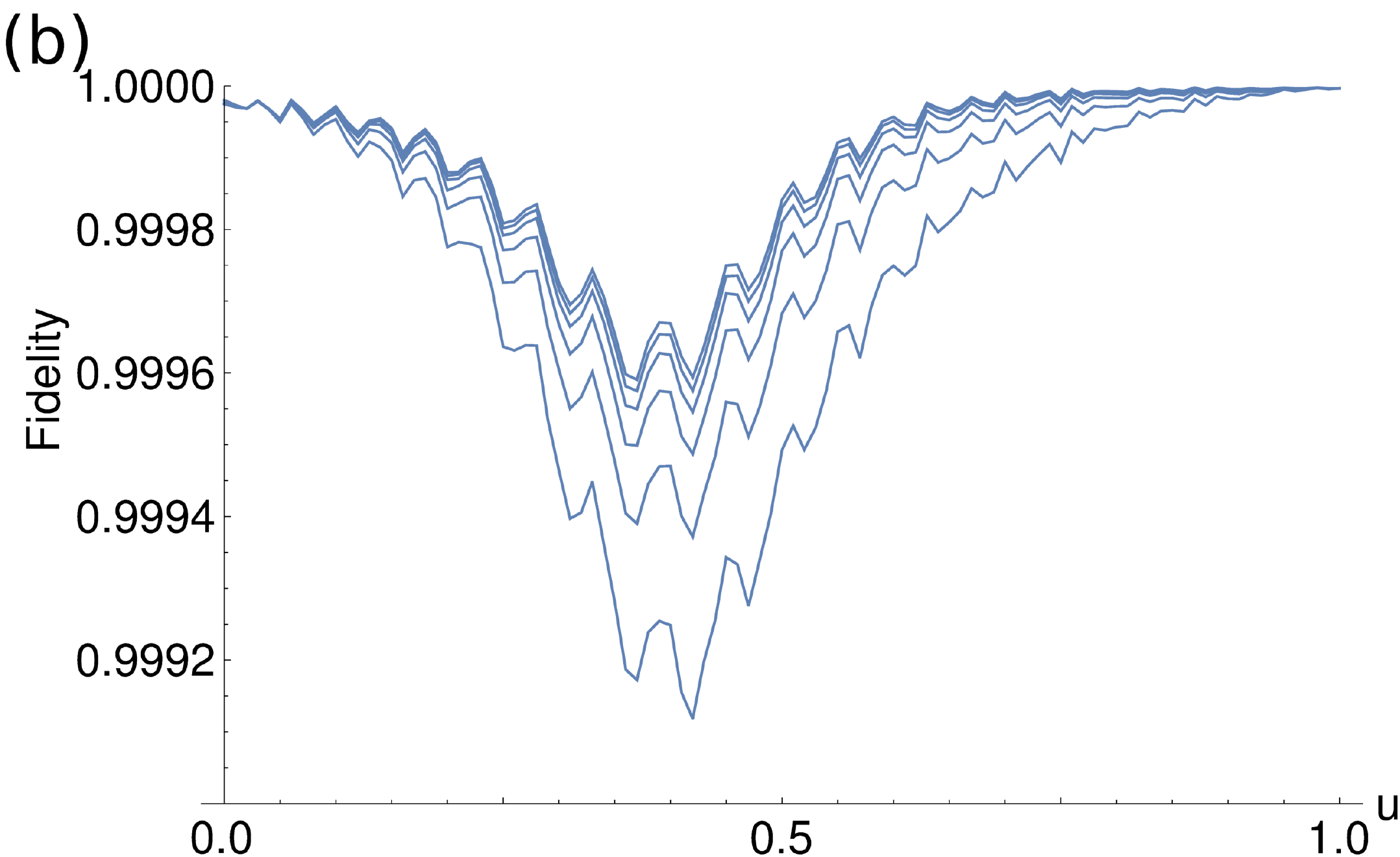}
\includegraphics[scale=0.19]{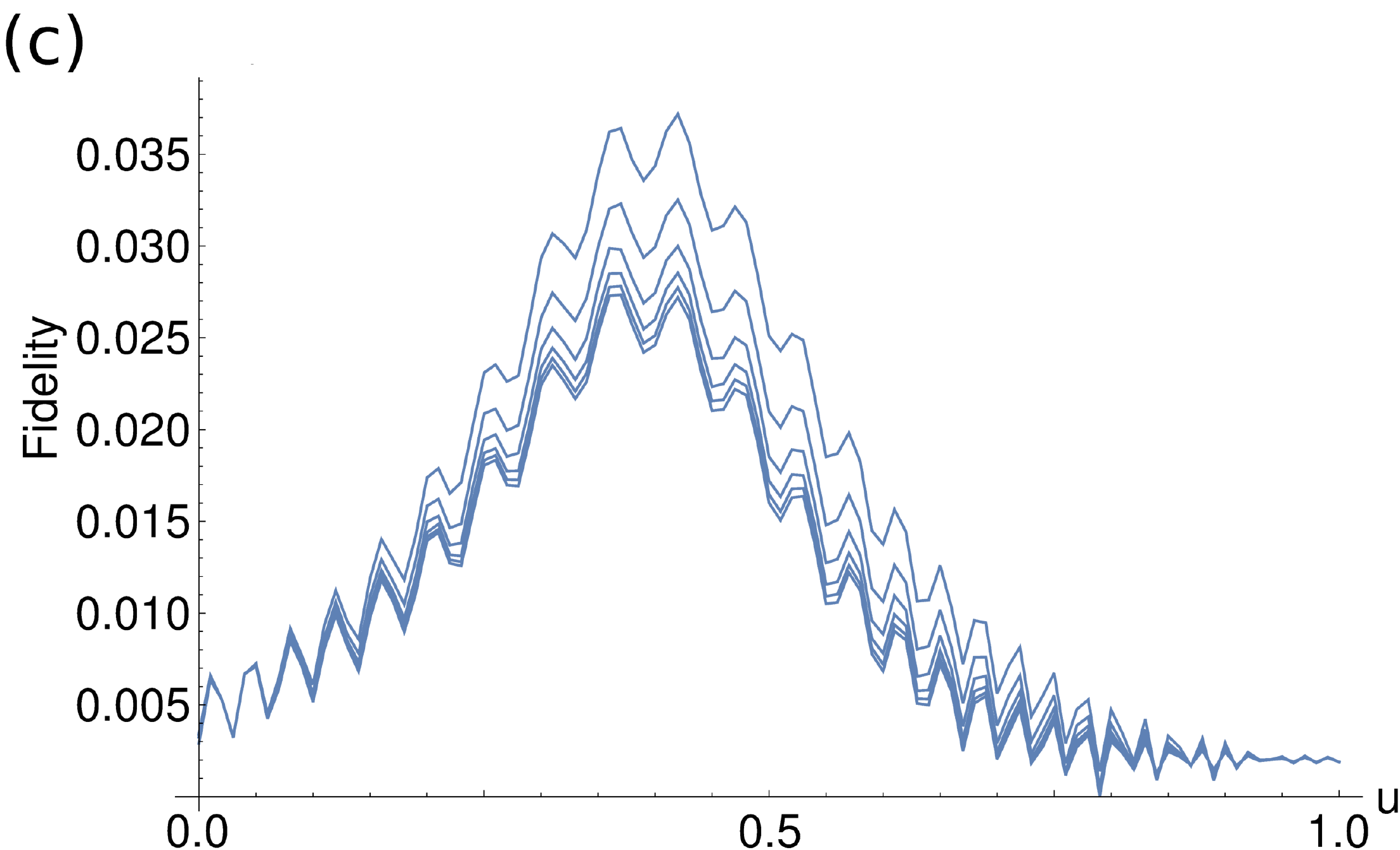}
\includegraphics[scale=0.18]{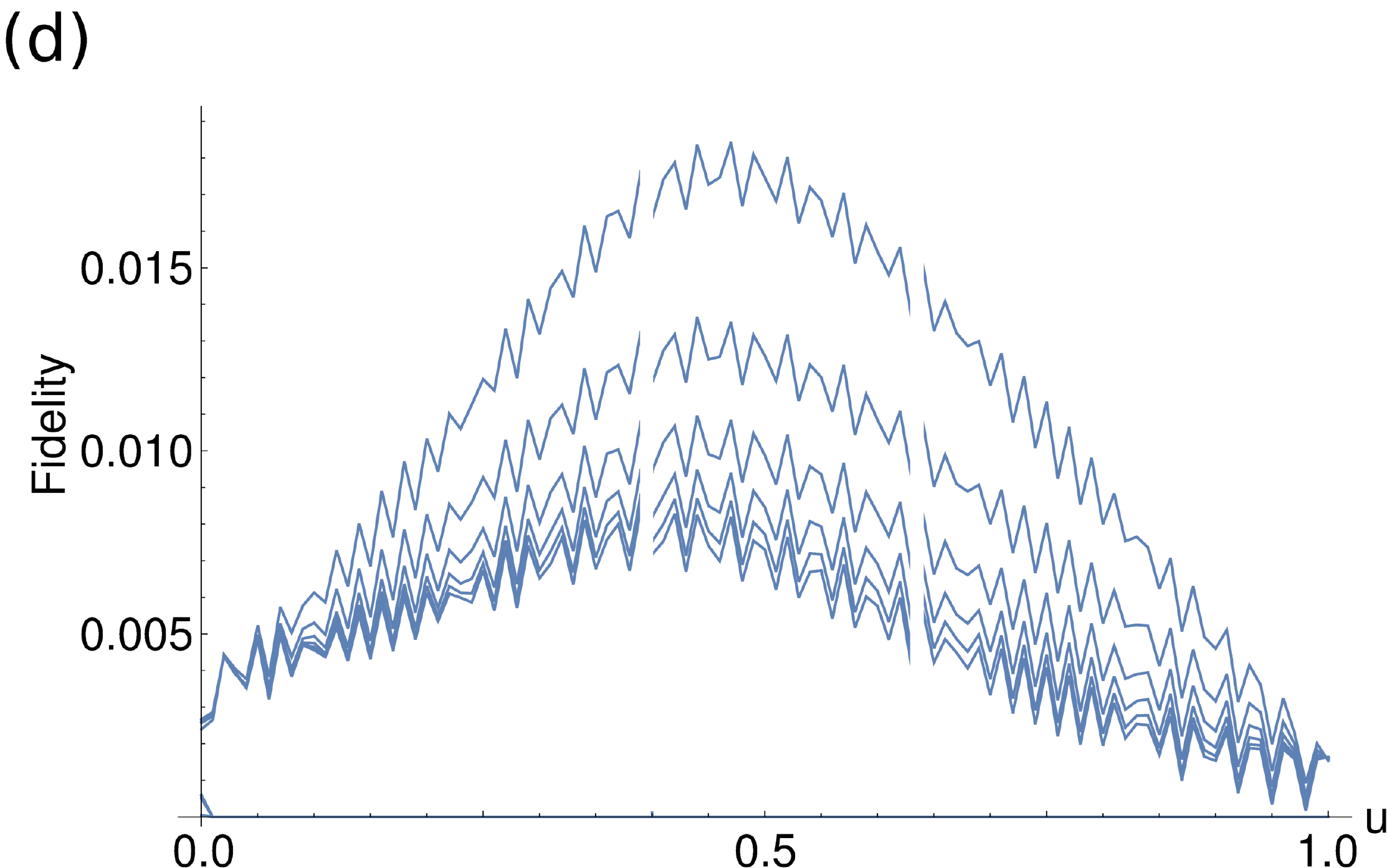}
\includegraphics[scale=0.18]{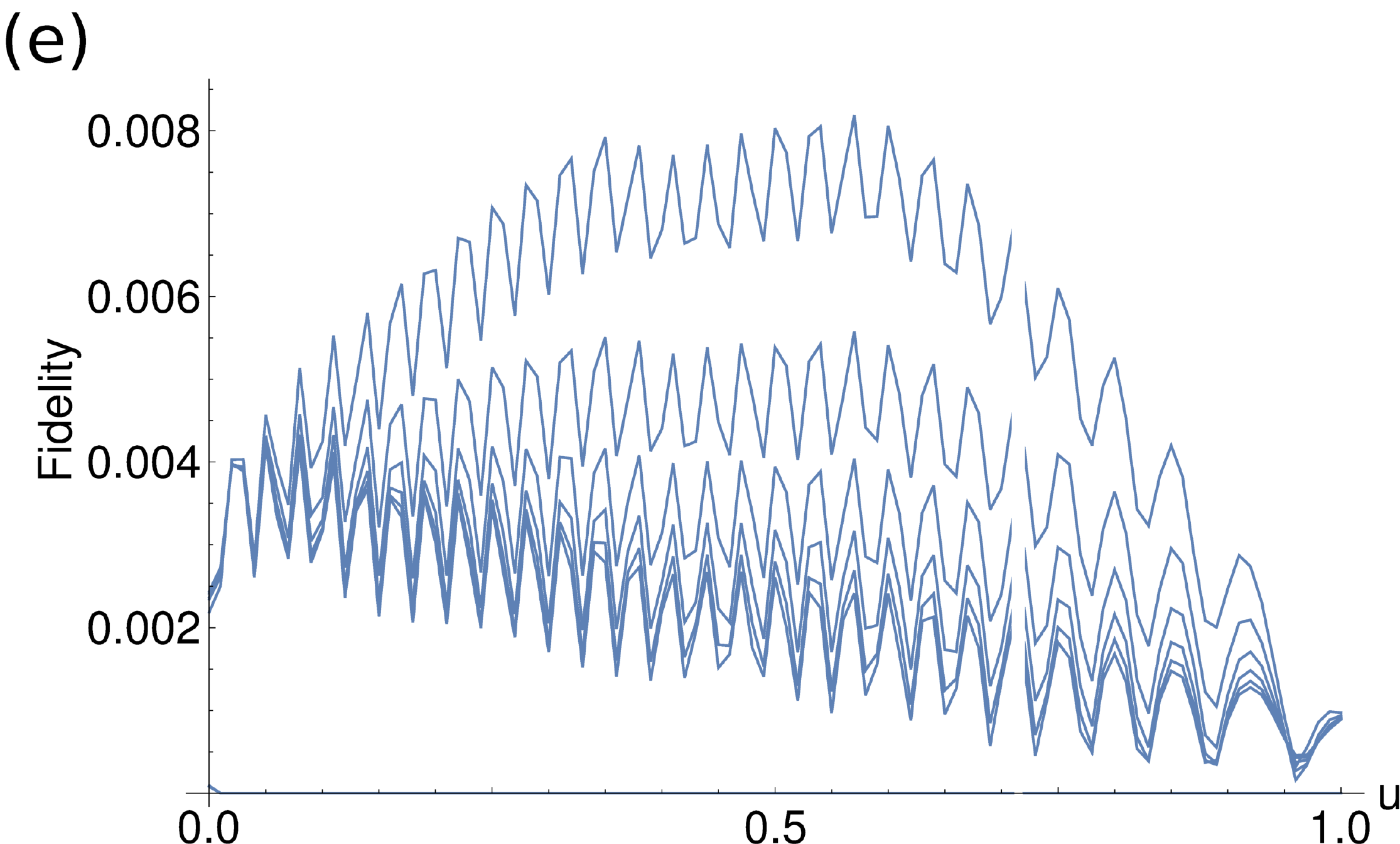}
\includegraphics[scale=0.18]{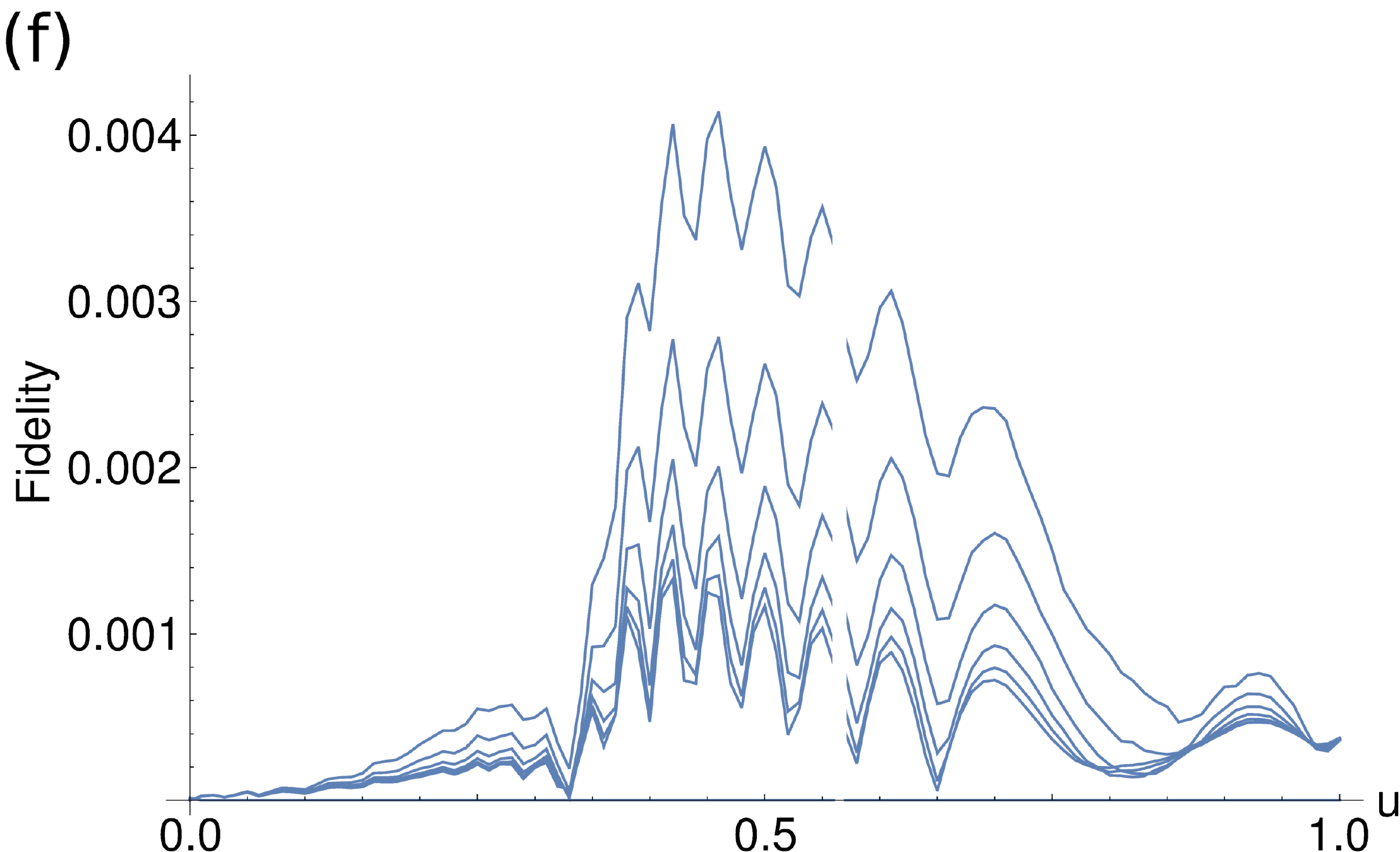}
\includegraphics[scale=0.19]{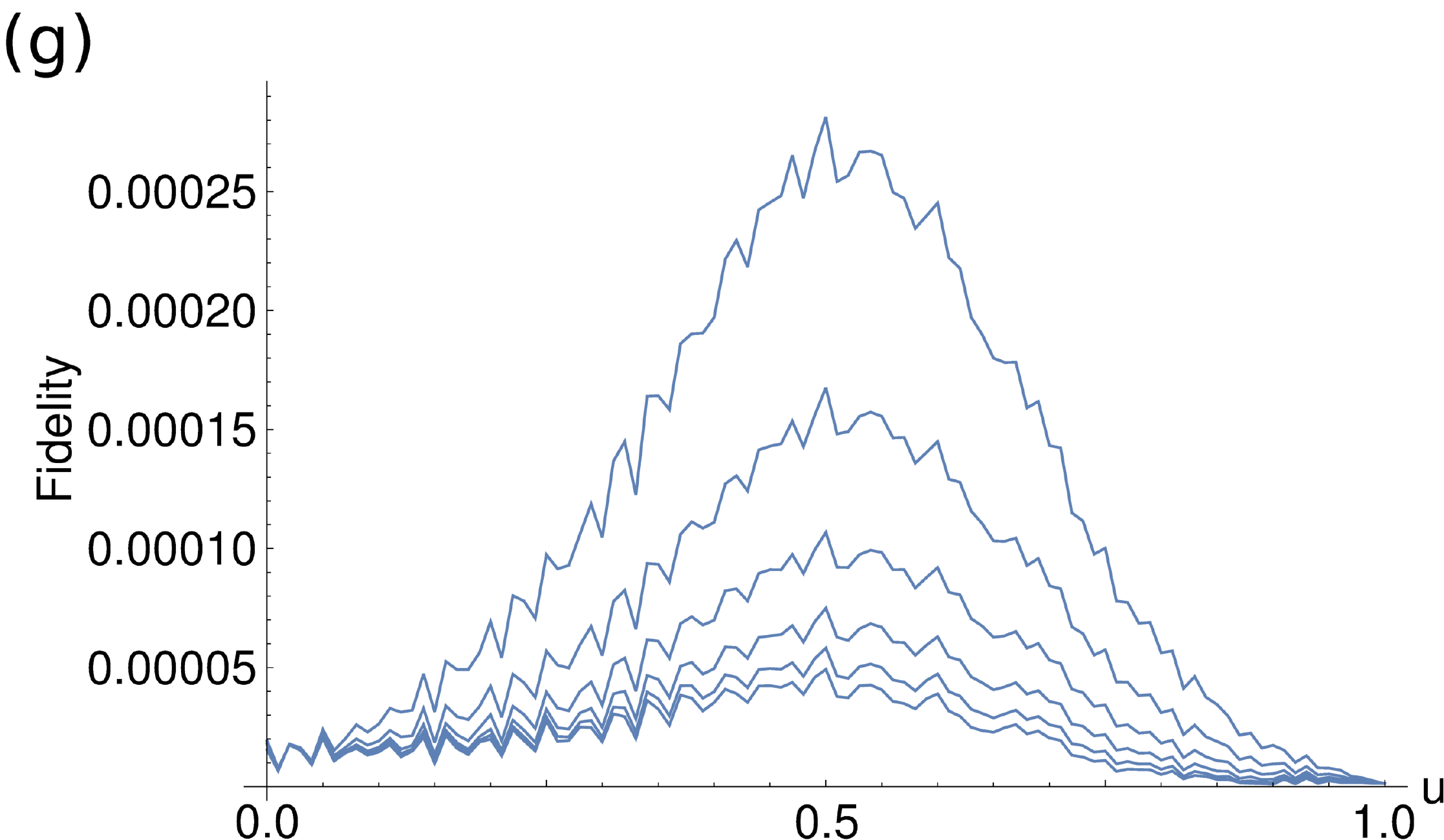}
\caption{(a) TFIM spectrum (\(n=4\)) and state fidelity between Trotterized state and (b) instantaneous true ground state, (c) first excited state, (d) 2/4/5th excited state, (e) 9/10th excited state, (f) 11th/13th excited state, and (g) 15th excited state. Specifically, the Hamiltonian is \(H_1 = \sum_i (-Z_i Z_{i+1} + 0.4 Z_i +0.4 X_i)\) and \(H_2 = \sum_i X_i\) and the total time \(T = 100\) and \(0.05 \le \delta \le 0001 t\). The corresponding energy levels of (b-g) are indicated in (a).}
\label{fig:TFIM_fidelity}
\end{figure}

\begin{figure}[H]
\includegraphics[scale=0.6]{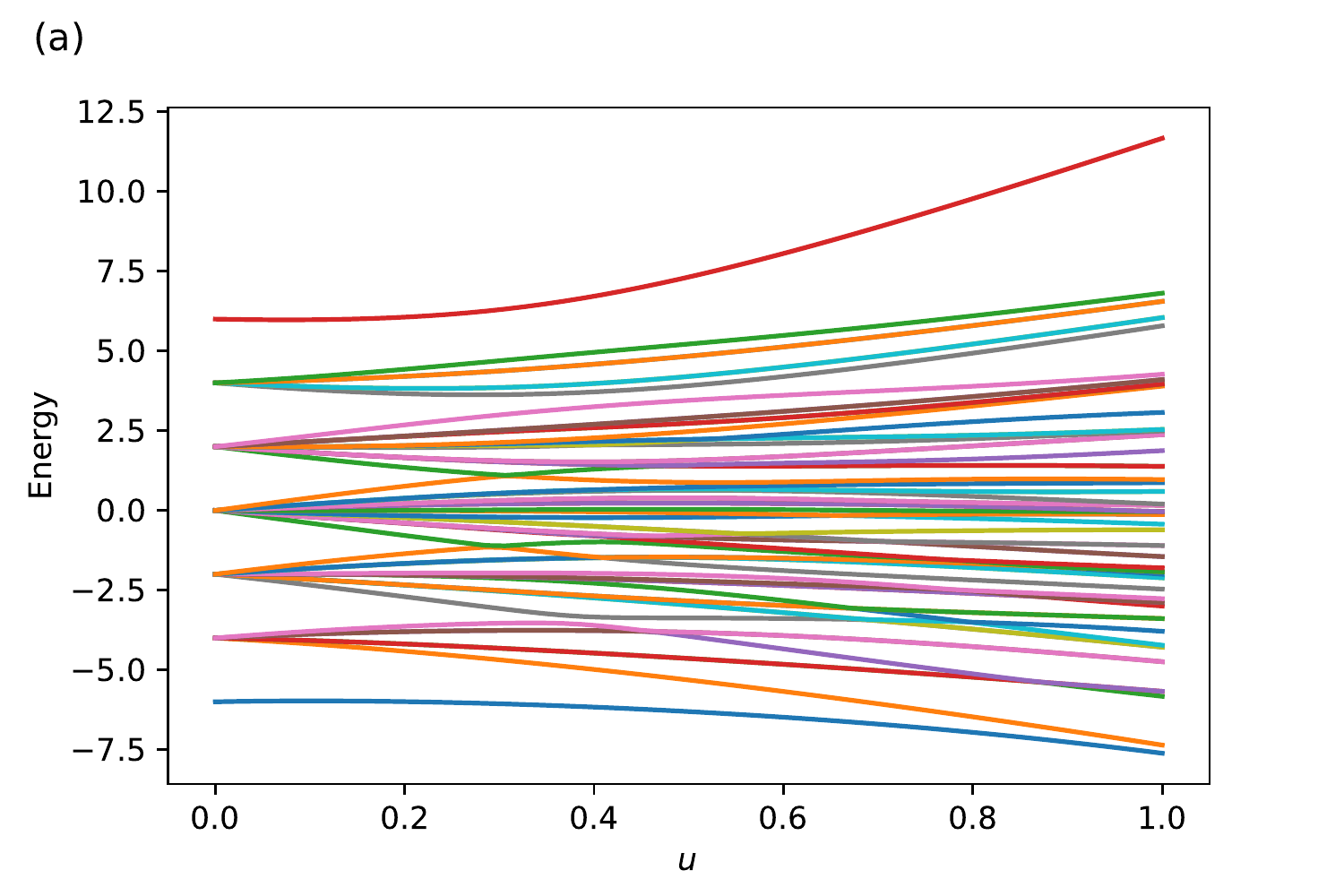}
\includegraphics[scale=0.28]{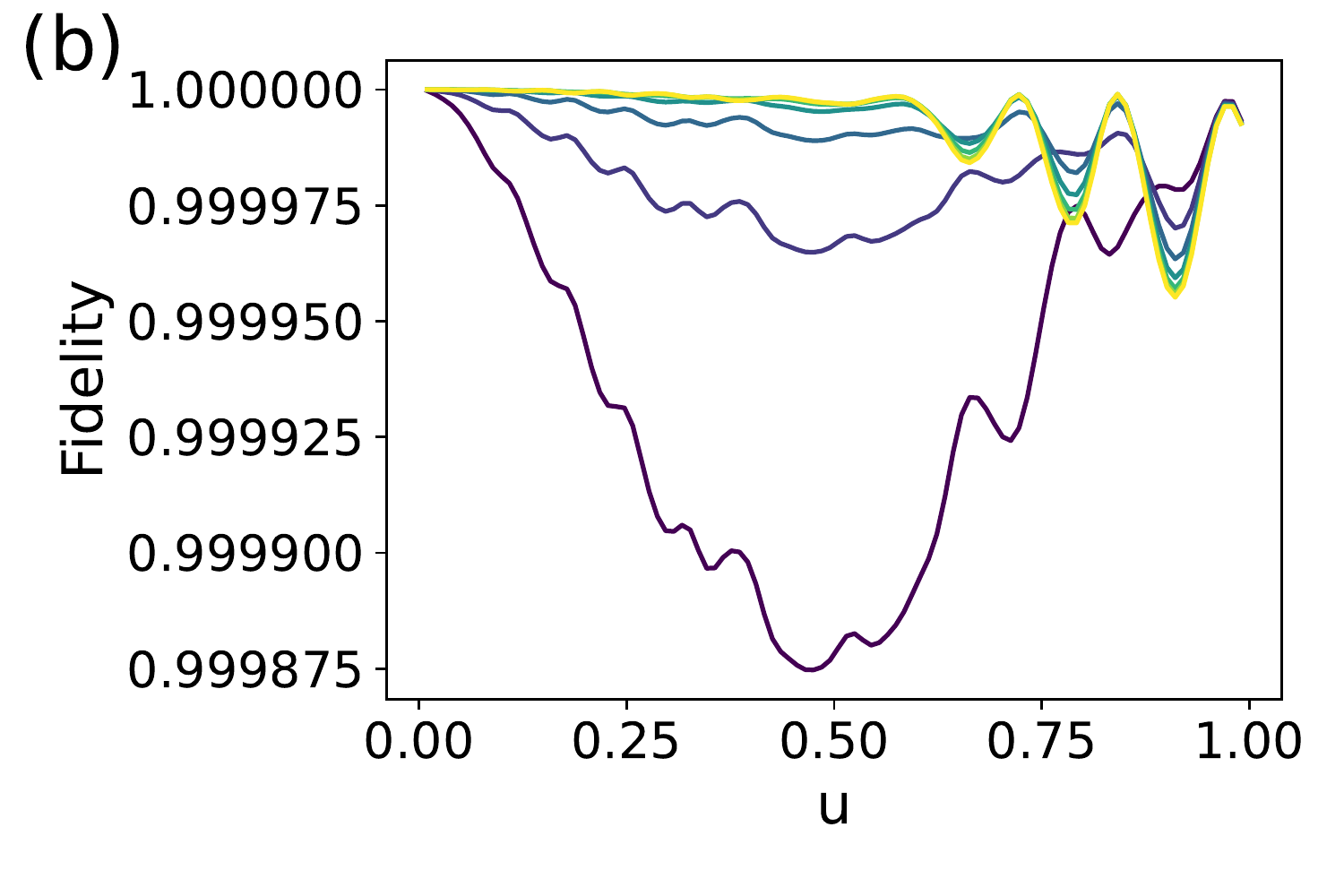}
\includegraphics[scale=0.28]{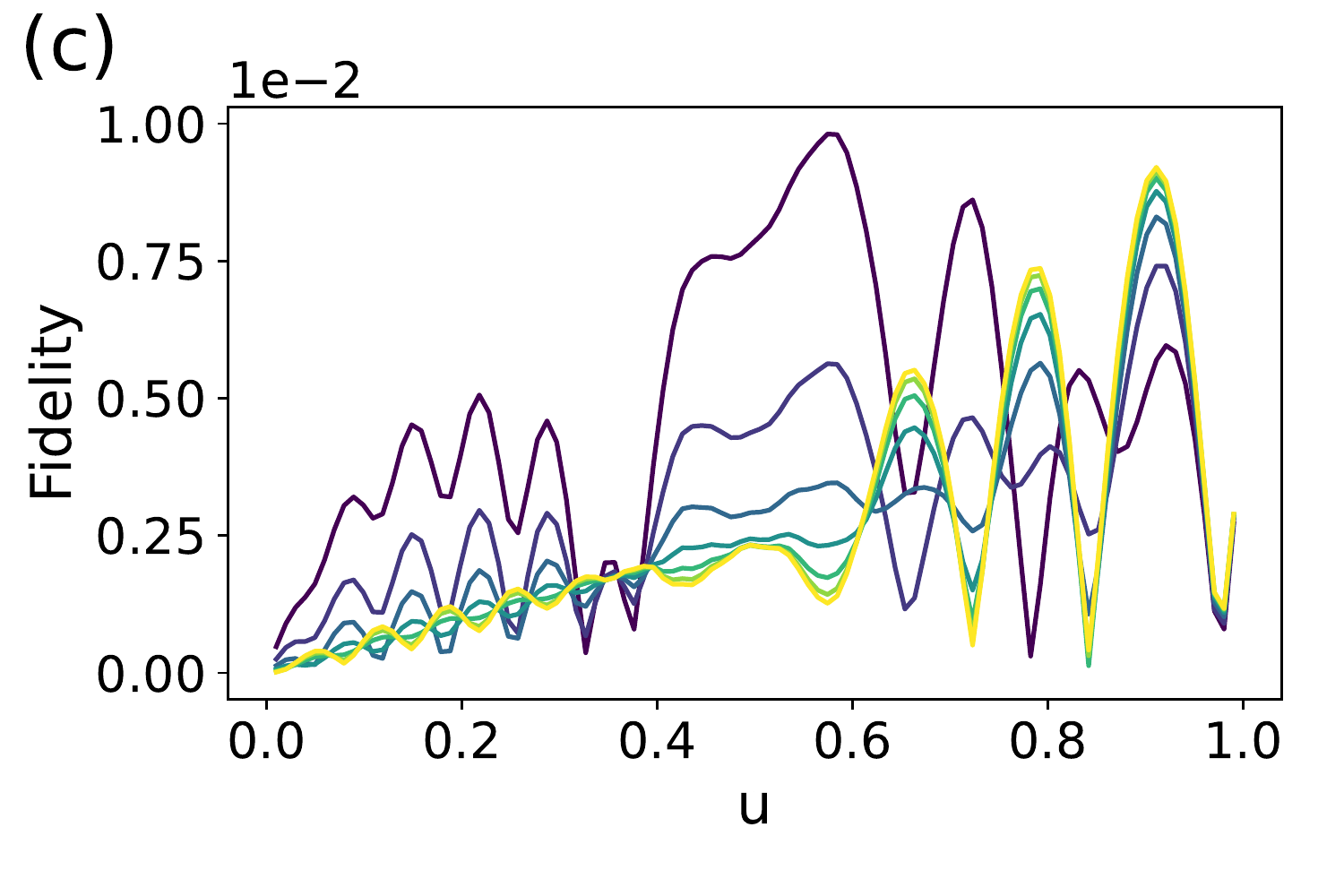}
\includegraphics[scale=0.28]{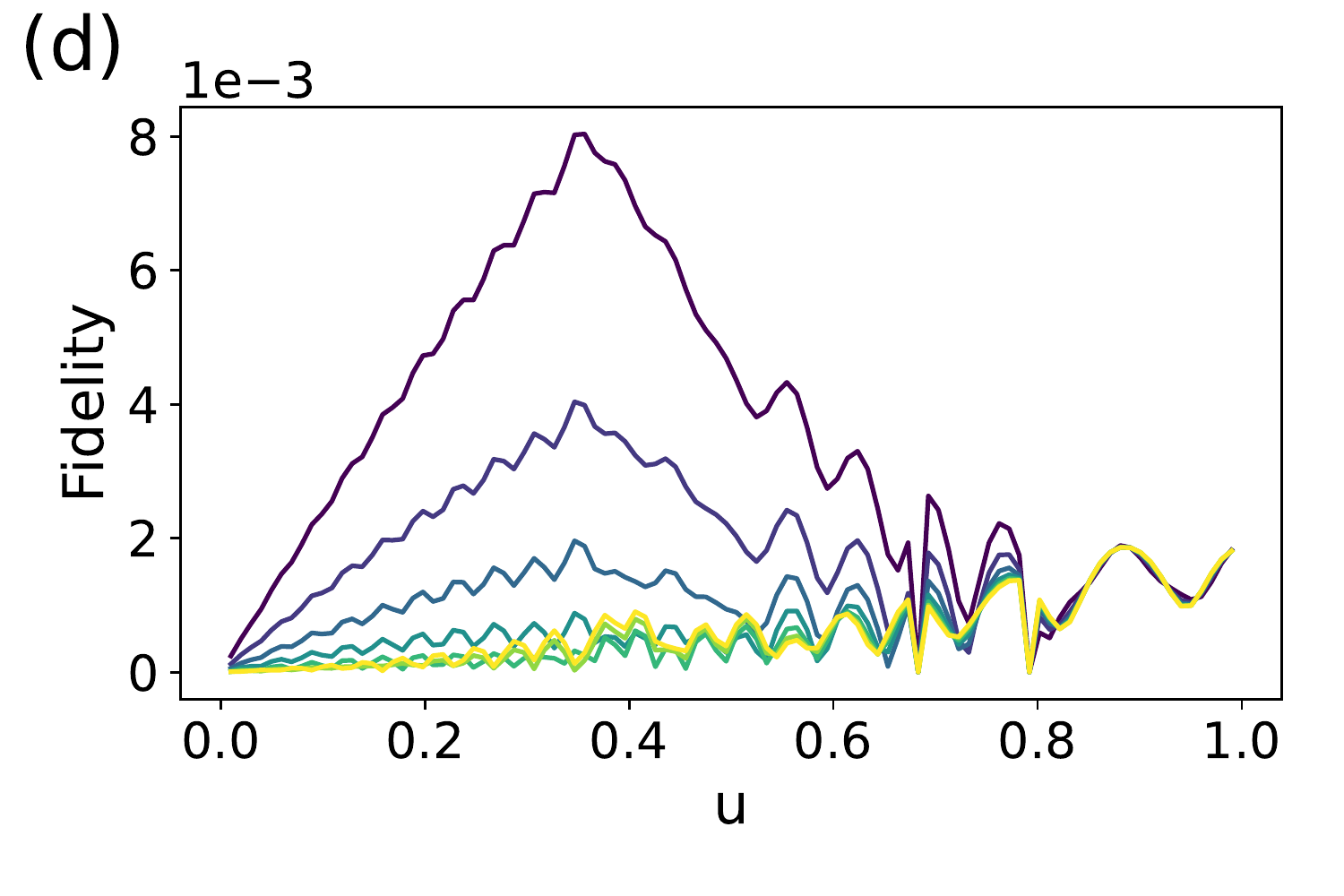}
\includegraphics[scale=0.28]{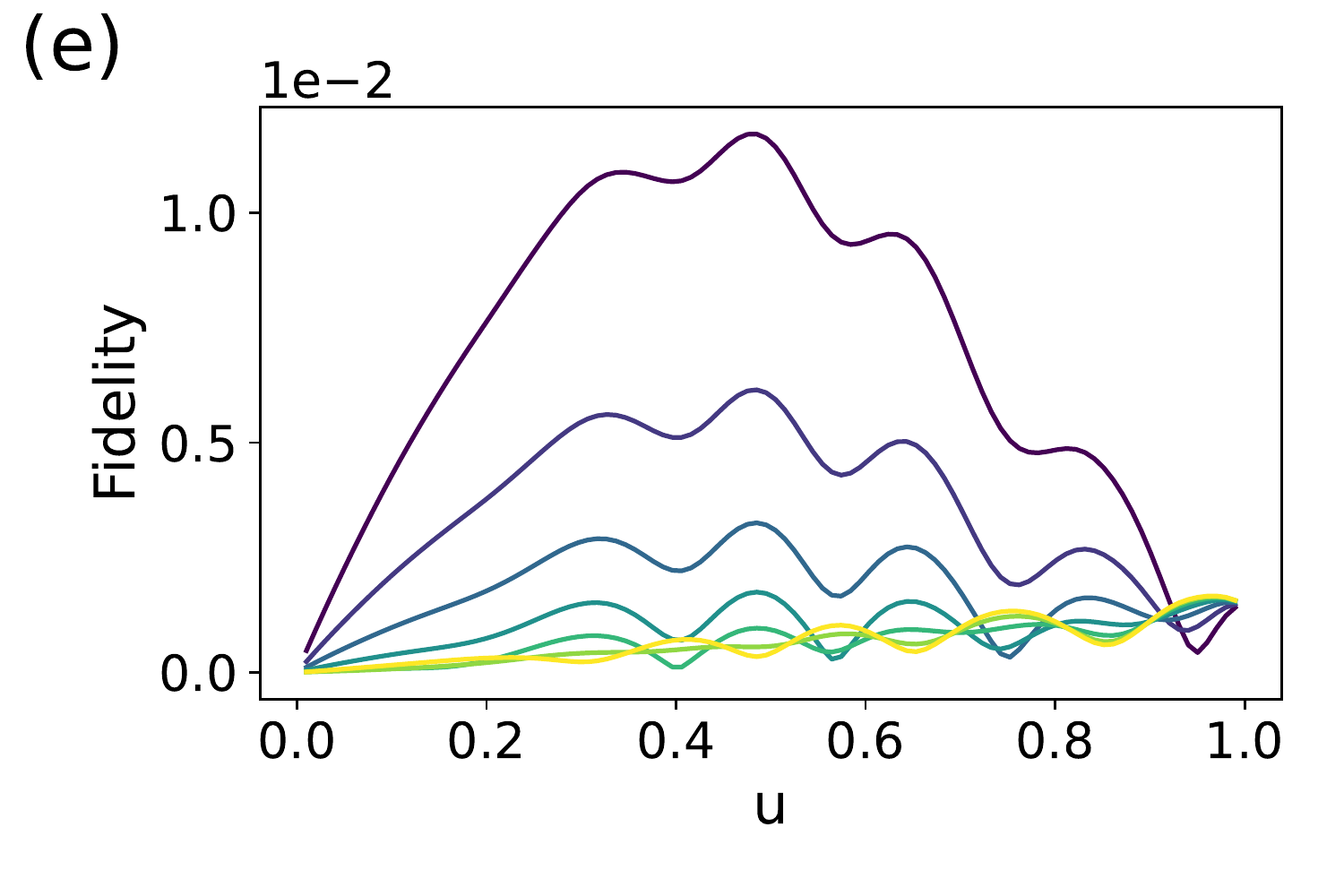}
\includegraphics[scale=0.28]{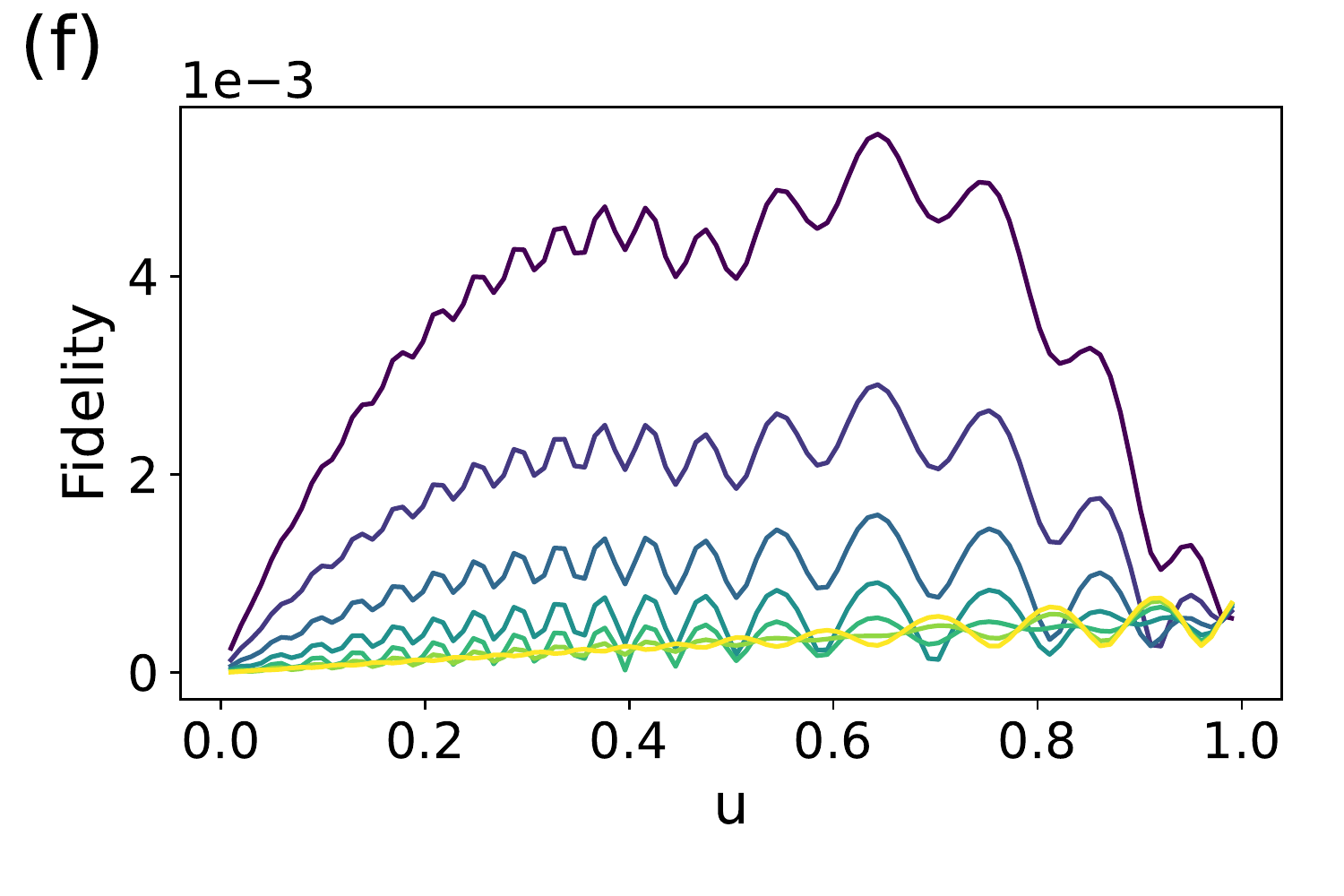}
\includegraphics[scale=0.28]{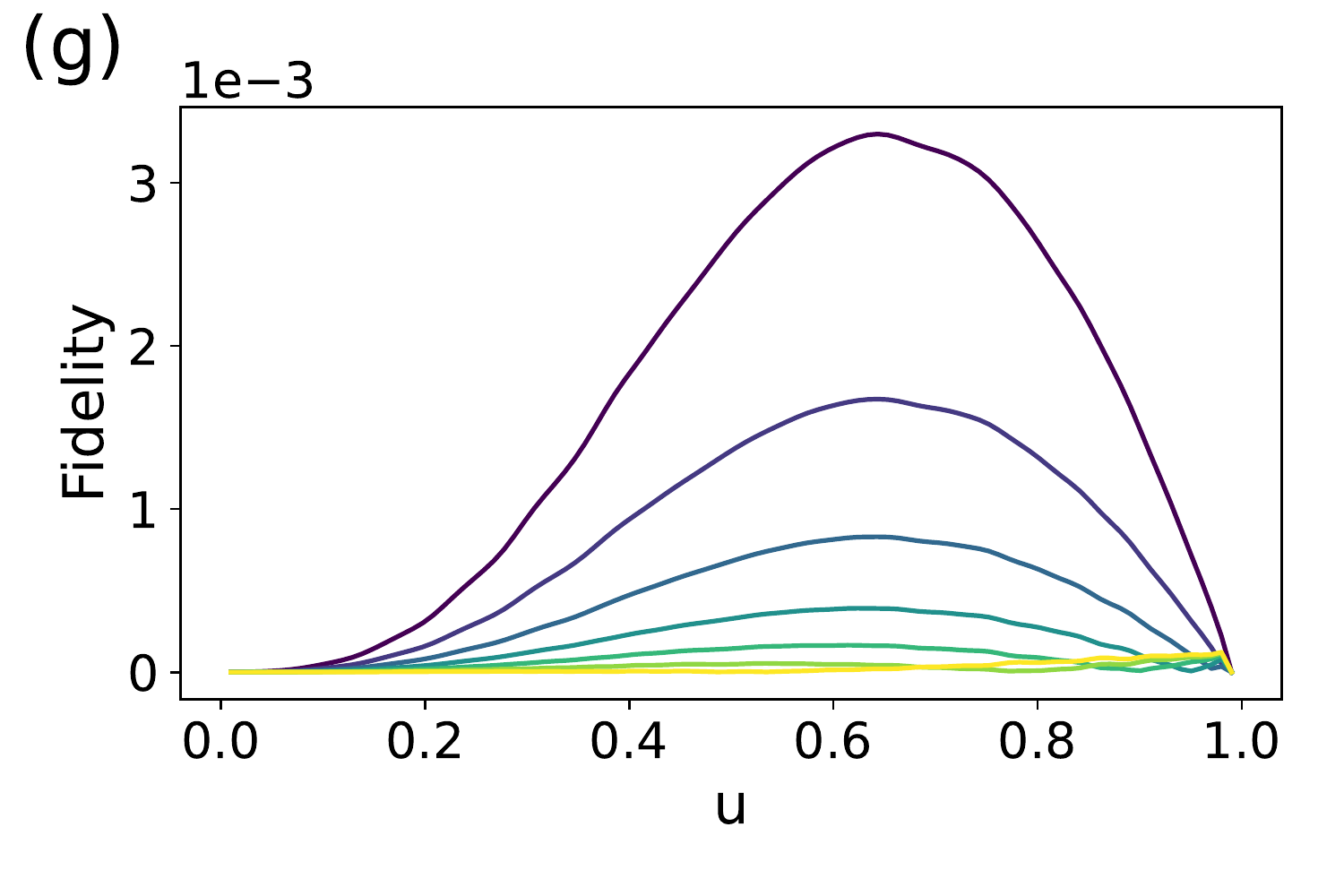}
\caption{(a) TFIM spectrum ($n=6$) and state fidelity between Trotterized state and (b) instantaneous true ground state, (c) 2nd/4th/6th excited state, (d) 7th/9th/10th excited state, (e) 13th/15th/16th excited state, (f) 21st/22nd/23rd excited state, and (g) 32nd/34th excited state. Specifically, the TFIM system Hamiltonian is $H_2=\sum_i(Z_iZ_{i+1}+0.8Z_i+0.9X_i)$, the mixer Hamiltonian is $H_1=\sum_i X_i$, total time $T=100$, and $0.05\leq\delta t\leq 0.001$.}
\label{fig:4_App_E}
\end{figure}

\section{Digitized State Preparation Resource Estimates}
\label{app:resource_estimates}
\renewcommand{\theequation}{F\arabic{equation}}
\renewcommand{\thefigure}{F\arabic{figure}}
\setcounter{figure}{0}

Here we will consider our result in the context of higher-order Trotterization.
Generic \(p\)th-order Trotter error bounds are proportional to the \(p\)th power of \(\delta t\). We will find that this implies that the circuit depth of an adiabatic evolution with a target infidelity \(\epsilon\) will scale as \(\mathcal O(\epsilon^{-\nicefrac{1}{2} - \nicefrac{1}{p}})\).
However, Theorem~\ref{th:adiabatic_trotter_error}'s bound scales independently of \(\delta t\) in the worst case.
We will find that this suggests a circuit depth scaling as \(\mathcal O(\epsilon^{-1/2})\), which is the same as the $p\rightarrow \infty$ limit of the generic bound.
In this way, generic bounds suggest that a reduction in $\epsilon$ scaling comparable to Theorem~\ref{th:adiabatic_trotter_error}'s would require infinite depth, which is clearly far too pessimistic.

Given a target maximum infidelity \(\epsilon\), in the adiabatic limit we want \(\epsilon \in \mathcal O(T^{-2}) \Leftrightarrow T \in \mathcal O(\epsilon^{-1/2})\).
We also want to pick a small enough time step \(\delta t\) such that the first-order Trotter error's contribution to the infidelity is on the same order as the infidelity.
Luckily, from Theorem~\ref{th:adiabatic_trotter_error} this means that we want \(\epsilon \in \mathcal O(T^{-2}) + \mathcal O(T^{-2} \delta t^2) \in \mathcal O(T^{-2})\), which means that the time step is independent of the target infidelity for small enough time step.
This implies that the number of time steps \(r = T \delta t^{-1} \in \mathcal O(T) \in \mathcal O(\epsilon^{-1/2})\).

On the other hand, from the na\"ive \(p\)th-order Trotter bound, we want to pick a small enough time step \(\delta t\) such that the \(p\)th-order Trotter error's contribution to the infidelity is on the same order as the infidelity. Thus we want \(\delta t\) such that \(\epsilon \in \mathcal O(T^2 \delta t^{2p})\).
Since \(T \in \mathcal O(\epsilon^{-1/2})\) from before, this implies that \(\delta t^{-1} \in \mathcal O(\epsilon^{-\nicefrac{1}{p}})\).
This means that the number of time steps \(r' = T \delta t^{-1} \in \mathcal O(\epsilon^{-1/2 -\nicefrac{1}{p}})\).

Notice that \(r' \rightarrow r\) when \(p \rightarrow \infty\), \textit{i.e.,} Theorem~\ref{th:adiabatic_trotter_error}'s favorable scaling with respect to target fidelity is only attainable from na\"ive Trotter bounds if the order of the Trotter expansion goes to infinity.

The cost of digitized state preparation is equal to the number of time steps, \(r\), multiplied by the cost of implementing a constant time step, \(C(N, p)\).
Generally, \(C(N, p)\) scales with system size (number of qubits), \(N\), and with the order of the Trotterization, \(p\).
For instance, for first-order Trotterization of a second-quantized quantum chemistry Hamiltonian, \(C(N, 1) \in \mathcal O(N^5)\)~\cite{Whitfield11}.
Higher-order Trotterizations of a second-quantized Hamitonian would require ever higher powers of \(N\).

This allows us to state a comparison between the digitized state preparation cost scaling of Theorem~\ref{th:adiabatic_trotter_error} with the na\"ive Trotter cost scaling in the following manner: Theorem~\ref{th:adiabatic_trotter_error}'s favorable cost scaling with respect to target fidelity is only attainable from prior Trotter bounds that have infinite power dependence in system size.

\end{document}